\documentclass{LMCS}

\def\doi{9(2:11)2013}
\lmcsheading%
{\doi}
{1--74}
{}
{}
{Mar.~\phantom06, 2012}
{Jun.~26, 2013}
{}

\usepackage{enumerate}
\usepackage{hyperref}

\usepackage[utf8]{inputenc}

\usepackage{latexsym}
\usepackage{url}
\usepackage{pst-all} 
\usepackage{array}

\usepackage[all]{xy}

\usepackage{amsfonts}
\usepackage{amssymb} 
\usepackage{bbding}
\usepackage{float}
\usepackage{wrapfig}
\usepackage[section]{placeins}
\usepackage{slashbox}
\usepackage{multirow}

\newcommand{\act}{\mbox{\it Act}}
\newcommand{\cero}{\mathbf{0}}

\newcommand{\lsem}{[\![}
\newcommand{\rsem}{]\!]}

\newcommand{\nat}{\mathrm{I}\!\mathrm{N}}
\newcommand{\barra}{\;|\;}
\newcommand{\init}{\mbox{\textit{I}}} 

\newcommand{\proc}{{\mathcal P}} 

\newcommand{\flee}{\longrightarrow}
\newcommand{\Flee}{\Longrightarrow}

\newcommand{\tran}[1]{\stackrel{#1}{\flee}}
\newcommand{\Tran}[1]{\stackrel{#1}{\Flee}}

\newcommand{\ol}{\overline}

\newcommand{\hnfx}[1]{\textit{hnf\/}^Z(#1)}
\newcommand{\tehnfx}[1]{\textit{tehnf\/}_N^Z(#1)}

\newcommand{\gtp}{\sqsupseteq} 
\newcommand{\ltp}{\sqsubseteq} 
\newcommand{\eip}{\equiv} 

\newcommand{\ltpF}{\ltp_{F}} 

\newcommand{\gtlbi}{\mathrel{\lower 0.9ex\hbox{$\stackrel{\displaystyle\sqsupset}{\sim}$}}} 
\newcommand{\ltlbi}{\mathrel{\lower 0.9ex\hbox{$\stackrel{\displaystyle\sqsubset}{\sim}$}}} 

\newcommand{\deter}[1]{\textit{Det}(#1)} 

\newcommand{\comment}[1]{}

\newcommand{\calB}{\mathcal{B}}

\newcommand{\calP}{\mathcal{P}}

\newcommand{\calT}{\mathcal{T}}

\newcommand{\formyacc}[1]{
if $\varphi \in \mathcal{L}^{\prime}_{#1}$ and $a \in \emph{Act}$ then $a\varphi \in \mathcal{L}^{\prime}_{#1}$}

\newcommand{\conjU}[2]{
if $\sigma, \sigma_j \in \mathcal{L}^{\prime}_#2$ for all $j\in J$ then $(\sigma\wedge\bigwedge_{j \in J}\neg \sigma_j\top) \in \mathcal{L}_{#1}$}

\newcommand{\conjsim}[2]{
If $\sigma \in \mathcal{L}^{\prime}_#2$ then $\sigma\in \mathcal{L}^{\prime}_{#1}$; \\ \item \vspace{-0.4cm}if $\sigma \in \mathcal{L}^{\prime}_#2$ then $\neg \sigma\in \mathcal{L}^{\prime}_{#1}$}
\newcommand{\conjbis}[2]{
if $\sigma \in \mathcal{L}^{\equiv}_#2$ then $\sigma \in \mathcal{L}^{\prime}_{#1}$}
\newcommand{\conjparticular}[2]{
if $X_1 \subseteq \mathcal{L}^{\prime}_#2$ and $X_2 \subseteq \mathcal{L}^{\prime}_#2$ then $(\bigwedge_{a \in X_1} a\top \wedge \bigwedge_{b \in X_2}\neg b\top) \in \mathcal{L}^{\prime}_{#1}$}
\newcommand{\conjneg}[2]{
if $\sigma \in \mathcal{L}^{\neg}_#2$ then $\sigma \in \mathcal{L}^{\prime}_{#1}$}
\newcommand{\conjmo}[2]{
if $\sigma \in \mathcal{L}^{\surd}_#2$ then $\sigma \in \mathcal{L}^{\prime}_{#1}$}
\newcommand{\conjnegparticular}[2]{
if $X_1 \subseteq \mathcal{L}^{\prime}_#2$ then $(\bigwedge_{a \in X_1}\neg a\top) \in \mathcal{L}^{\prime}_{#1}$}
\newcommand{\conjyform}[2]{
if $\varphi \in \mathcal{L}^{\prime}_{#1}$ and $\sigma \in \mathcal{L}^{\equiv}_#2$ then $\sigma \wedge \varphi \in \mathcal{L}^{\prime}_{#1}$}
\newcommand{\conjyformparticular}[2]{
if $\varphi \in \mathcal{L}^{\prime}_{#1}$ and $X_1, X_2 \subseteq \mathcal{L}^{\prime}_#2$ then $(\bigwedge_{a \in X_1} a\top \wedge \bigwedge_{b \in X_2}\neg b\top) \wedge \varphi\in \mathcal{L}^{\prime}_{#1}$}
\newcommand{\conjyformneg}[2]{
if $\varphi \in \mathcal{L}^{\prime}_{#1}$ and $\sigma \in \mathcal{L}^{\neg}_#2$ then $\sigma \wedge \varphi \in \mathcal{L}^{\prime}_{#1}$}
\newcommand{\conjyformmo}[2]{
if $\varphi \in \mathcal{L}^{\prime}_{#1}$ and $\sigma \in \mathcal{L}^{\surd}_#2$ then $\sigma \wedge \varphi \in \mathcal{L}^{\prime}_{#1}$}
\newcommand{\conjyformnegparticular}[2]{
if $\varphi \in \mathcal{L}^{\prime}_{#1}$ and $X_1 \subseteq \mathcal{L}^{\prime}_#2$ then $(\bigwedge_{a \in X_1}\neg a\top)\wedge \varphi \in \mathcal{L}^{\prime}_{#1}$}

\newcommand{\bran}[2]{\begin{iteMize}{$\bullet$}\item #1;\item if $\varphi_i \in \mathcal{L}^{\prime}_{#2}$ for all $i \in I$ then $\bigwedge_{i\in I} \varphi_i \in \mathcal{L}^{\prime}_{#2}$; \item\formyacc{#2}.\end{iteMize}}
\newcommand{\brannew}[2]{\begin{iteMize}{$\bullet$}\item#1; \item if $\varphi_i \in \mathcal{L}^{\prime}_{#2}$ for all $i \in I$ then $\bigwedge_{i\in I} \varphi_i \in \mathcal{L}^{\prime}_{#2}$; \item\formyacc{#2}.\end{iteMize}}

\newcommand{\dbran}[2]{\begin{iteMize}{$\bullet$}\item$\top \in \mathcal{L}^{\prime}_{#2}$; \item #1; \item if $X\subseteq \emph{Act}$ and $\varphi_a \in \mathcal{L}^{\prime}_{#2}$ for all $a \in X$ then $\bigwedge_{a\in X} a\varphi_a \in \mathcal{L}^{\prime}_{#2}.$ \end{iteMize}}

\newcommand{\lin}[2]{\begin{iteMize}{$\bullet$}\item$\top \in \mathcal{L}^{\prime}_{#2}$; \item#1; \item\formyacc{#2}.\end{iteMize}}

\newcommand{\linf}[2]{\begin{iteMize}{$\bullet$}\item$\top \in \mathcal{L}^{\prime}_{#2}$; \item\conjyformneg{#2}{#1}; \item\formyacc{#2}.\end{iteMize}}
\newcommand{\linfmo}[2]{\begin{iteMize}{$\bullet$}\item$\top \in \mathcal{L}^{\prime}_{#2}$; \item\conjyformmo{#2}{#1}; \item\formyacc{#2}.\end{iteMize}}
\newcommand{\linfparticular}[2]{\begin{iteMize}{$\bullet$}\item$\top \in \mathcal{L}^{\prime}_{#2}$; \item\conjyformnegparticular{#2}{#1}; \item\formyacc{#2}.\end{iteMize}}

\newcommand{\linc}[2]{\begin{iteMize}{$\bullet$}\item$\top \in \mathcal{L}^{\prime}_{#2}$; \item\conjbis{#2}{#1}; \item\formyacc{#2}.\end{iteMize}}
\newcommand{\lincparticular}[2]{\begin{iteMize}{$\bullet$}\item$\top \in \mathcal{L}^{\prime}_{#2}$; \item\conjparticular{#2}{#1}; \item\formyacc{#2}.\end{iteMize}}

\newcommand{\linfc}[2]{\begin{iteMize}{$\bullet$}\item$\top \in \mathcal{L}^{\prime}_{#2}$; \item\conjneg{#2}{#1}; \item\formyacc{#2}.\end{iteMize}}
\newcommand{\linfcmo}[2]{\begin{iteMize}{$\bullet$}\item$\top \in \mathcal{L}^{\prime}_{#2}$; \item\conjmo{#2}{#1}; \item\formyacc{#2}.\end{iteMize}}
\newcommand{\linfcparticular}[2]{\begin{iteMize}{$\bullet$}\item$\top \in \mathcal{L}^{\prime}_{#2}$; \item\conjnegparticular{#2}{#1}; \item\formyacc{#2}.\end{iteMize}}

\newcommand{\conjsimold}[2]{
If $\sigma \in \mathcal{L}_#2$ then $\sigma\in \mathcal{L}^{\prime}_{#1}$; \\ \item \vspace{-0.4cm} if $\sigma \in \mathcal{L}_#2$ then $\neg \sigma\in \mathcal{L}^{\prime}_{#1}$}

\begin{document}

\title[Unifying the ltbt spectrum]{Unifying the Linear time-Branching time spectrum of strong process semantics}

\author[D.~de Frutos Escrig]{David de Frutos Escrig\rsuper a}	
\address{{\lsuper{a,b,c,d}}Departamento de Sistemas Inform\'aticos y Computaci\'on
  \\ Universidad Complutense de Madrid}	
\email{\{defrutos, cgr, miguelpt\}@sip.ucm.es, dromeroh@pdi.ucm.es}  
\thanks{{\lsuper{a,c,d}}David de Frutos Escrig, Miguel Palomino and David Romero Hernández were partially supported by the Spanish
MEC project  DESAFIOS10 TIN2009-14599-C03-01 and the project
PROMETIDOS S2009/TIC-1465.}	

\author[C.~Gregorio Rodr\'iguez]{Carlos Gregorio Rodr\'iguez\rsuper b}	
\thanks{{\lsuper b}Carlos Gregorio Rodríguez was partially supported by the Spanish MEC project ESTuDIo TIN2012-36812-C02-01.}	

\author[M.~Palomino]{Miguel Palomino\rsuper c}	

\author[D.~Romero Hern\'andez]{David Romero Hern\'andez\rsuper d}	


\ACMCCS{[{\bf Theory of computation}]: Formal languages and automata
  theory---Formalisms---Algebraic language theory; Formal languages
  and automata theory---Semantics and reasoning---Program semantics;
  Logic; Models of computation---Concurrency}


\begin{abstract}
  Van Glabbeek's linear time-branching time spectrum is one of the most
  relevant work on comparative study on process semantics, in which semantics
  are partially ordered by their discrimination power. In this paper we bring
  forward a refinement of this classification and show how the process
  semantics can be dealt with in a uniform way: based on the very natural
  concept of constrained simulation we show how we can classify the spectrum in
  layers; for the families lying in the same layer we show how to obtain in a
  generic way equational, observational, logical and operational
  characterizations; relations among layers are also very natural and
  differences just stem from the constraint imposed on the simulations that
  rule the layers. Our methodology also shows how to achieve a uniform
  treatment of semantic preorders and equivalences.
\end{abstract}

\maketitle

\newpage

\tableofcontents

\section{Introduction}
\label{sec:intro}
  Since the foundational work by Robin Milner~\cite{Mil80CCS,Mil89CC} 
    and Tony Hoare~\cite{Hoa85CSP} on process semantics, there has been a 
    multitude of proposals to endow processes with meaning and to define
    equivalence and preorder relations over them.
    Among the most relevant work are those of Matthew Hennessy~\cite{Hen88ATP},
    who introduced the testing methodology defining process semantics from
    test cases, and those of Jan Bergstra and Jan Willen Klop~\cite{BK84}, later
    continued by Jos Baeten and Peter Weijland~\cite{Bae90PA}, which were
    based on an axiomatic approach.    

    These proposals define algebraic languages for the specification of 
    processes, diverging in subtle details concerning the treatment of
    non-determinism and parallelism.
    These aspects are captured by means of certain operators which may 
    (strongly) vary in each particular language.

    Focusing on equivalences, it is interesting to note how the pioneering
    work in this area already established two fundamental notions,
    bisimulation and traces/failures, that constitute an upper and a lower
    bound on the natural framework in which other process equivalences
    can be studied.
    Hoare---with his characteristic clarity---summarizes the situation
    in the following paragraph.

    \begin{quote}
      \emph{CCS makes many distinctions between processes which would be 
        regarded as identical in this book. The reason for this is that CCS is
        intended to serve as a framework for a family of models, each of which 
        may make more identifications than CCS but cannot make less. To avoid 
        restricting the range of models, CCS makes only those identifications 
        which seem absolutely essential. In the mathematical model of this 
        book [CSP] we have pursued exactly the opposite goal ---we have made as 
        many identifications as possible, preserving only the most essential 
        distinctions.}~\cite{Hoa85CSP}
    \end{quote}

    In between these two fundamental notions of equivalence---bisimulation and
    traces---the last two decades of the 20th century witnessed the surge of
    a large variety of new equivalences associated to new calculi and process 
    algebras, whose aim was to explore the different needs for expressivity
    and distinction capabilities in many applications.

    The most important taxonomic work on process semantics was carried out
    by Rob van Glabbeek as part of his doctoral dissertation~\cite{Gla90tesis}.
    In two papers, titled~\emph{Linear time-branching time 
    spectrum}~\cite{Gla90,Gla93}, he collected the most important of these
    equivalences establishing, among other results that we will comment on,
    a classification based on their capability to distinguish processes.
    The first of the papers concentrate on \emph{strong} semantics, in 
    the sense that they consider each action processes
    perform as being observable by their environment. The second paper
    consider the inclusion in the language of a new and invisible action $\tau$;
    process semantics considering this internal action are
    usually called \emph{weak} semantics.
    Figure~\ref{fig:ltbtsAmpliado} shows a slightly expanded version of the
    spectrum proposed by van Glabbeek for the case of strong semantics in \cite{Gla90}.
    These strong semantics, that do not consider at all the special role of 
    internal actions $\tau$, are the only ones that we consider in this paper.

    This array of semantics is supported by many authors who claim that
    there is no single ``good'' definition.
    Process theory can be applied in a wide spectrum of contexts and situations
    and the concrete uses will have a decisive influence in the election of 
    what a suitable semantics should be.

    \begin{quote}
      \emph{The choice of a suitable semantics may depend on the tools an
      environment has, to distinguish between certain processes. It is
      conceivable that a concurrency theory is equipped with different
      semantics, and has the capacity to express equality on different 
      levels.}~\cite{Gla97}
    \end{quote}

    The possibility to define several and varied semantics can then be
    considered to be an advantage of the theory, since it allows for the
    necessary flexibility to reflect different notions of processes and
    equivalence and preorder relations over them.

    Nevertheless, this multiplicity has gone hand in hand in the literature
    with an individual study of each of the semantics that somehow 
    makes the whole theory less appealing because such a cornucopia can
    become a handicap both for its study and its practical application.
    For instance, although most of the semantic notions defined for 
    processes simultaneously
    induce both a (pre)order and an equivalence,~\footnote{A remarkable
    exception, however, is the bisimulation notion, for which no non-trivial
    order relation is known.}  
    the literature has frequently overlooked the fact that
    these two notions are mutually intertwined, as we will show later.
    Likewise, the study of the concrete models has been usually undertaken
    paying little attention to the other semantics or to the relations among 
    them, even though it is well-known that there exist ``families'' of 
    semantics---such as the linear semantics---which are undoubtedly related.

    A unified study of semantics has both methodological and practical 
    implications that have been explored along the last years by the authors
    of this work, for example 
    in~\cite{FG05concur,FG08ifiptcs,FG09ic,FGP09sofsem,DeFrutosEtAl08b}, 
    and also in work by important researchers in the 
    area~\cite{AFIL05,AFI07,CFG08,LV10}.
    This research shows that a unified view of process semantics is indeed
    possible.

  This is precisely the main goal we set to reach with our work: we
    aim to study process semantics in a generic way, making the equivalence and
    (pre)order relations our object of study in order to find patterns,
    to identify families, to search for properties among these relations,
    so that we obtain generic results that need not be proved again and again 
    for each of the semantics.

    We aim, in a nutshell, at a unifying view of process semantics that
    can be used to understand them both jointly and individually and that
    allows to continue with their theoretical study in a more focused manner,
    helping to identify those properties a semantics should have for a
    particular application.

  \begin{figure}[tbp]
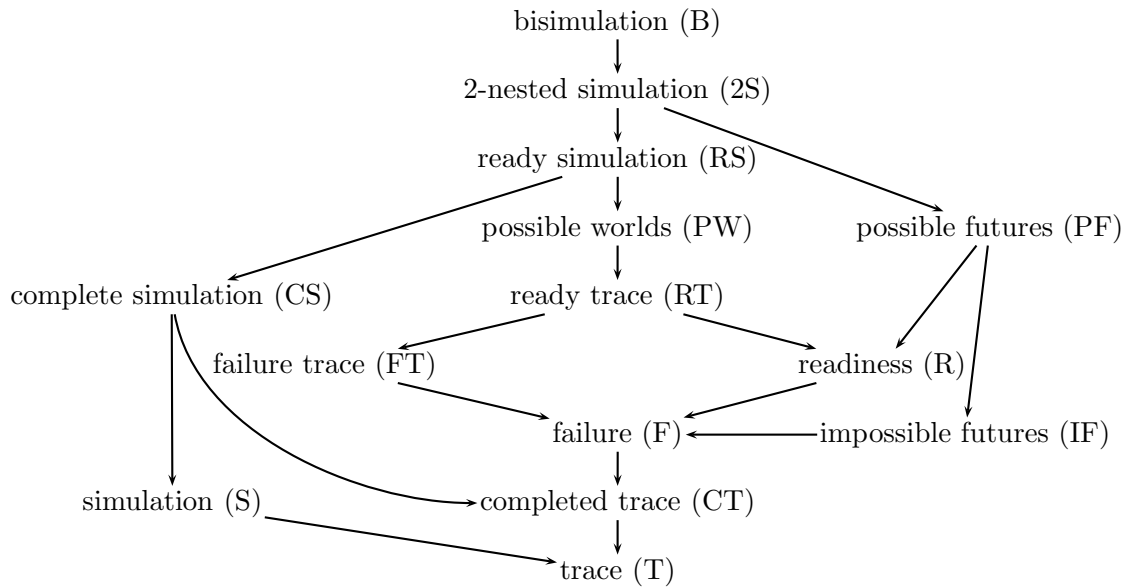

    \centering  
  \psset{nodesep=1pt}
  \begin{tabular}{lcl}
    & \rnode{b}{bisimulation (B)} & \\
    & & \\  & \rnode{ns}{2-nested simulation (2S)} & \\
    & & \\  & \rnode{rs}{ready simulation (RS)} & \\
    & & \\  & \rnode{pw}{possible worlds (PW)} & \hspace*{2em}\rnode{pf}{possible futures (PF)}\\
    & & \\  \rnode{cs}{complete simulation (CS)}&  \rnode{rt}{ready trace (RT)} & \\
    & & \\  \hspace*{7em}\rnode{ft}{failure trace (FT)} & & \rnode{r}{readiness (R)}\\
    & & \\  & \rnode{f}{failure (F)} & \hspace*{.75em}\rnode{if}{impossible futures (IF)}\\
    & & \\ \hspace*{2.5em}\rnode{s}{simulation (S)}& \rnode{ct}{completed trace (CT)}  \\
    & & \\  & \rnode{t}{trace (T)} & \\
    \ncline{->}{b}{ns}
    \ncline{->}{ns}{rs}
    \ncline{->}{ns}{pf}
    \ncline{->}{rs}{pw}
    \ncline{->}{rs}{cs}
    \ncline{->}{cs}{s}
    \nccurve[angleA=-80,angleB=180]{->}{cs}{ct}
    \ncline{->}{pw}{rt}
    \ncline{->}{pf}{r}
    \ncline{->}{pf}{if}
    \ncline{->}{rt}{ft}
    \ncline{->}{rt}{r}
    \ncline{->}{ft}{f}
    \ncline{->}{r}{f}
    \ncline{->}{if}{f}
    \ncline{->}{f}{ct}
    \ncline{->}{ct}{t}
    \ncline{->}{s}{t}
  \end{tabular} 
  \caption{Linear time-branching time spectrum.}
  \label{fig:ltbtsAmpliado}
\end{figure}

\subsection{Overview of results}

This paper contains a consolidated and extended presentation
of the unification of observational, equational and the logic process semantics published in
\cite{FGP09sofsem,FGP09mfps,RF11} for strong behavioral semantics.\footnote{We comment in Section~\ref{sec:conclusions} the work of some of the authors on unification of weak semantics.} We take advantage of the joint and larger presentation of the subject to tighten the
connections between the different views. Besides, we make the paper mostly self-contained providing proofs for 
all the results; we also complete the study with new results not included in \cite{FGP09sofsem,FGP09mfps,RF11}.
We have also completed the revision of the unification of strong process semantics with a section
devoted to the unified presentation of the operational semantics.

Next we describe the main results we have obtained. They can be used 
as a roadmap for reading the paper and to understand 
the technical details in the following sections.
  \begin{iteMize}{$\bullet$}
  \item One of the most generic results we have proved is the existence of
    two essential families of semantics:
    branching semantics and linear semantics. Certainly, 
    this was already hinted by van Glabbeek when he named its spectrum 
    of semantics 
    ``linear time-branching time''.
    
  Our results show that the most representative branching semantics have 
  characterizations as simulation semantics. Moreover, every simulation 
  has a natural family of coarser linear 
  semantics associated to it that inherit some of its properties.

  In Figure~\ref{fig:extended-ltbts} (page~\pageref{fig:extended-ltbts}),
  branching semantics are located to the left and each of them defines a
  layer of ``induced'' linear semantics, to their right.
  For example, ready simulation is the branching semantics from which
  the classic diamond of linear semantics composed of failures, readiness, 
  failure trace, and ready trace semantics is generated.
  These semantics, as we will later see in detail, inherit some of their
  axiomatic characterizations directly from the ready simulation.
  In addition, the same layer also contains the possible worlds semantics,
  which is a deterministic branching semantics.

  Even though the axiomatic characterizations contained in 
  Section~\ref{sec:EquationalSemantics} already show this dependency between
  linear and branching semantics---see Figure~\ref{fig:ReadySlice2} (page~\pageref{fig:ReadySlice2})---it
  is in Section~\ref{observational-sem-sec} where, using techniques from
  denotational semantics, the relations between the original branching
  semantics and the induced linear semantics can be fully appreciated.
  The relationships among the different linear semantics in the same branching
  layer are also completely specified in that section.
    
  \item Equational characterizations reveal in a very concise manner the
    basic properties of the different semantics.
    As a result of our research, we have been able to derive a generic
    axiomatic characterization of the semantics in the spectrum which shows
    clearly the relationships among them: the uniformity in the definitions
    of the branching semantics and the different families of linear
    semantics becomes apparent, as well as the tight relation between each
    branching semantics and its associated linear semantics.
    In Section~\ref{sec:EquationalSemantics} we present all the details
    related to these axiomatic characterizations.

  \item Another result that we consider important is that it is indeed
    possible to establish a clear relationship between the preorder and the
    equivalence associated to a given semantics.
    From an axiomatic point of view, in Section~\ref{sec:EquationalSemantics}
    we show how these characterizations are closely related. 
    In fact, there are algorithms that allow to  easily 
    obtain the axioms for the equivalence from those of the 
    preorder~\cite{AFI07,DeFrutosEtAl08b}, and also the other way around,
    the axioms for the preorder from those of 
    the equivalence~\cite{FG07sos,FG08ifiptcs}.

  \item We also offer a unifying view of the process semantics based on 
    observational (denotational) semantics, according to which we have
    classified the process semantics in four categories:
    \begin{iteMize}{$-$}
    \item  bisimulation semantics, which is the finest semantics in the spectrum and the only one that cannot be
      defined by means of a non-trivial preorder;
    \item the simulation semantics (simulation, complete simulation,
      ready simulation, nested simulation, \dots) which are 
      characterized by means
      of branching observations, that is, labeled trees;
    \item the linear semantics (traces, failures, readiness, \dots),
      characterized by linear observations,
      a degenerated case of branching observations;
    \item the deterministic branching semantics corresponding to an
      intermediate class between branching and linear,
      where observations are deterministic trees. Possible worlds semantics
      is the only semantics in the original van Glabbeek's 
      spectrum in this class.
    \end{iteMize}

Besides their linear or branching nature, semantics are characterized by a local
observation function that generates the local observations at the states.
For the linear case there is also the possibility of observing this
local information in a partial way and this is how for each local
observer, in principle, up to four different semantics can be obtained.
In particular, this gives rise to the classic diamond below the ready simulation
semantics formed by failures, failure trace, readiness, and ready trace
semantics.

The uniform presentation of the process semantics that we 
offer in Section~\ref{observational-sem-sec} 
clarifies the relationships  and hierarchies among all the semantics; 
moreover, it will make possible the development of generic proofs of their 
common properties.

 \item
 We also present a unified view of the logical semantics. Again, the
 bisimulation semantics, which is characterized by the Hennessy-Milner logic
 HML \cite{hm85}, is our starting point, and then we aim for the
 sublogics that characterize each of the semantics in the spectrum. Guided by
 our main unification goal, we have not tried to obtain the smallest possible
 sets of formulas, but have veered for the largest sublogics that characterize
 each of the semantics. 
  
 Hence, the finest semantics are characterized by the largest sublogics and in
 fact we obtained a uniform characterization that informs us about
 the hierarchy of semantics, by proving that a semantics $S_1$ is  finer than
 another $S_2$ if and only if the corresponding logics satisfy
 $\mathcal{L}_{S_2} \subseteq \mathcal{L}_{S_1}$. Moreover, the classification
 into branching and linear time semantics is also reflected in the structural
 definition of each logic. In particular, the branching semantics are
 characterized by the free use of negation over the formulas that define the
 corresponding constraint, while the linear semantics at each layer of the
 spectrum introduce ever more limitations in the subformulas.

\item Finally, we also discuss an operational-like presentation of the
  semantics in the spectrum; more precisely, we consider an evaluation
  semantics to derive the appropriate data which characterize them.
  Those data are quite similar to the ones employed
  for our observational semantics so that it is not them, but the way in which
  they are derived, that enhances our understanding of the features of
  each of the semantics and the relationships between them. These presentations
  somehow generalize the work by Cleaveland and Hennessy \cite{CH93} on the
  characterization of the Testing semantics by means of bisimulation. There is also a clear
  connection with the work by two of the authors of this paper on (Bi)simulation up-to \cite{FG09ic}.

  \item A concomitant, but still important result of our work, is of methodological    nature: the semantics are amenable to a working methodology that allows
    for general results that can be applied to families of semantics as well 
    as to yet to be defined semantics.
    The requirements we impose on these new semantics are relatively mild.
    An example of this is shown seen in Section~\ref{sec:RealDiamond}, 
    where some new process semantics---indeed two new families---smoothly
    integrate into our general theory. In fact, it was nice to discover that
    one of them is just the revivals    
    semantics, which has been recently developed by Bill Roscoe~\cite{Roscoe09}.

    \item Each of the characterization frameworks---equational, observational,
      logical or operational---sheds light on the spectrum in different and
      complementary ways. This has provided us with different ways to study all
      the previously known semantics and the relationships between them. That
      complementary nature also sprang up along our unification work when we
      discovered one by one all the factors that contribute to the structure of
      the extended spectrum. In particular, when considering simple and natural
      combinations of axioms we found out the new meet semantics in
      Section~\ref{sec:RealDiamond}, while their dual join semantics was
      discovered in a natural way when considering the observational
      characterizations. Finally, the semantics of minimal readies (in
      Section~\ref{logica_lineales}) appeared when investigating the logical
      framework. While not too important on its own, our unification work has
      also revealed a mistake in the classic logical characterization of one of
      the semantics in the original spectrum (see
      Section~\ref{logica_ramificada}): it was the general and systematic
      approach that guides our uniform characterization that allowed it.
  \end{iteMize}

\subsection{Some related work} 

Naturally, the goal of defining a global or general theory of process
semantics has been around for a long time and 
several relevant authors in the field have already paved the way that
we now tread.

Despite the methodological differences between Milner's work, based
  on bisimulation, and Hoare's, on denotational semantics, 
  both of them had in common the search for characterizations---logical, 
  axiomatic, observational---that could shed light from different angles on 
  the world of process semantics.

  Hennessy introduced the testing methodology to endow processes with
  semantics, making the notion of equivalence to spring 
  from the application of the interaction principles for processes 
  expressed within the model.
  Perhaps one of the most important contributions of his work was what
  he called the ``trinity'': processes can be seen as syntactic terms
  in an algebra, as operational descriptions in labeled transition systems,
  or as denotational objects in a mathematical model.
  With our work we have somehow extended this trinity in a generic manner to 
  all the semantics in the extended spectrum.

  Van Glabbeek's work, \emph{the linear time-branching time spectrum} 
  aimed at the comparison of most known semantics---at the time he developed 
  his seminal work---by  presenting them within common frameworks that 
  would allow a comparative study of their properties.
  Besides providing uniform definitions over transition systems, 
  van Glabbeek also proposed to characterize the semantics in terms of 
  logical formulas.
  The set of modal formulas whose satisfaction equivalence
  identifies the same processes as the corresponding semantics is 
  defined.
  Because of the compositional definition of the corresponding sets of formulas,
  this characterization can be considered to be denotational semantics.

  Another characterization provided by van Glabbeek is the
  axiomatic one, for which he defines the BCCSP language that is used in this
  work (see Definition~\ref{def:BasicProcessAlgebra}).
  Twelve of the semantics in the spectrum are characterized by means of
  sets of axioms over syntactic terms for this language.
  For most of them---except for bisimulation, that has no associated 
  order---their characterizations are actually twofold: on the one hand,
  the natural order relation that defines each 
  semantics---Table~\ref{tab:AxiomatisationsPreordersLtbts}---and, on the other
  hand, the 
  induced equivalence---Table~\ref{tab:AxiomatisationsEquivalencesLtbts}. 
  Many of these characterizations were previously known but, again, their
  uniform presentation is one of its main merits.

  A deep study---individual as well as comparative---of these
  axiomatizations and the quest for answers to the new questions that 
  arise from this study has been one of the leading forces behind 
  our research.
  Actually, some of our most relevant results can be combined into a new
  way of presenting the spectrum---Figure~\ref{fig:extended-ltbts} 
  (page~\pageref{fig:extended-ltbts})---that allows for a better
  comprehension of the semantics since it clarifies their relative 
  positions within it and shows the existence of ``gaps'' that correspond
  to new semantics whose addition to the graph reflects a desirable regularity
  that makes it clearer.
  Hoare's work on the unification of the study of process algebras
\cite{Hoare05} was also an important influence.
Specially, the relationship between similarity and trace refinement, which we
have generalized by establishing the connection between branching time and
linear time semantics, and the connection between the denotational, the
algebraic, and the operational styles proposed by him and He Jifeng in
\cite{HJ98UTP}. 
  
As already mentioned, Roscoe has contributed in an independent research
effort in parallel with ours to the study of new process semantics by proposing
his stable revivals model \cite{Roscoe09}.  
He relates his new semantics with other well known linear time semantics
and the rediscovery of that semantics in our expanded spectrum gave us
the opportunity to present those relationships with a unified and generic
light. 
  
  There is other relevant work in the area of process theory that
  has inspired us. The number of contributions is too large to cite all of them here.
  Anyone interested on finding a more exhaustive list of relevant references may collect them, for instance, from~\cite{OriginsConcurrentProgramming02,HistoryProcessAlgebra-Baeten-2005,Aceto03,Gla97,S09}.
  There the historic evolution of the area and many of the most important contributions to it are reviewed. 
  To them we can add four recent books on process algebras and related
  subjects \cite{BBR09PAETCP,R10UCS,AILS07RSMSV,S12IBC}, presenting different
  points of view and some of the semantics studied in this paper. 
  Finally, in our Conclusions, we will discuss a bit the work on the generic study and classification of the weak semantics.


\subsection{Paper structure}
We have structured this paper as follows.
Section~\ref{sec:Preliminaries} introduces all the basic definitions and 
notation to properly follow the developments in the following sections.

In Section~\ref{sec:EquationalSemantics}
we propose alternative characterizations for the axiomatizations of 
the semantics in the spectrum, both for orders and equivalences. 
All these axiomatizations are based on just 
two parametrical axiom skeletons that clearly highlight 
the relations among the different semantics. 

Section~\ref{observational-sem-sec} presents a unified observational
characterization for process semantics.  One of the key ideas is that
constrained simulations are uniformly characterized by a branching observation
plus a local observation function. From the observations of a given constrained
simulation, the linear semantics in its layer are uniformly derived.

 In Section~\ref{rnoef:sec} we prove that the equations we presented in Section~\ref{sec:EquationalSemantics}
are deduced from the observations defined in Section~\ref{observational-sem-sec} in a general way, 
without using at all the already known axiomatizations for the semantics.

Section~\ref{sec:LogicalCharacterization} follows the trails of 
Sections~\ref{sec:EquationalSemantics} and \ref{observational-sem-sec} by
introducing a unified logical characterization of process semantics.

In Section~\ref{logic_observational_framework} we prove that the observational
characterizations developed in Section~\ref{observational-sem-sec} allow for
generic proofs for the logical characterizations presented in
Section~\ref{sec:LogicalCharacterization}. Therefore the ``trinity'' of
equations, observations, and logical formulas is established in a generic way
for large families of process semantics.

Section~\ref{sec:RealDiamond} is a practical proof of the applicability of our unification proposal. Some new process
semantics, that were not listed in the original linear time-branching time spectrum, are easily accommodated in our
framework thus getting the corresponding semantic characterizations that we
have presented in previous sections. 

In Section~\ref{sec:OperationalSemantics} we conclude the unified presentation
of the semantics in the spectrum by developing an operational characterization
which mainly produces the information provided by the observational semantics,
but inferred in an operational way, using a (unified) set of SOS-like rules.

Finally, in Section~\ref{sec:conclusions} we offer some conclusions and lines 
for future work.

\paragraph{\textbf{Acknowledgments}.}
We gratefully acknowledge three anonymous referees for their very thoughtful
and detailed comments on a previous version of this work, that have greatly
helped us to improve the presentation of this material.

\section{Preliminaries}
\label{sec:Preliminaries}

Although the main results in this paper are also valid for infinite 
processes---as we showed in~\cite{FG05concur,FG09ic}---in order to 
simplify the presentation of the concepts,
we will mainly consider finite processes generated by the basic process 
algebra BCCSP which contains only the basic process algebraic operators from 
CCS \cite{Mil89CC}, and CSP \cite{Hoa85CSP}, but is sufficiently
powerful to express all finite synchronization trees \cite{Mil80CCS}.
This language has repeatedly been used in unification work, e.g. \cite{AFI07,Gla01}.


\begin{defi}
  \label{def:BasicProcessAlgebra}
  Given a set of actions $\act$,
  the set BCCSP$(\act)$ of processes is
  defined by the following BNF-grammar:
  \[
    p::=\; \cero\barra ap \barra p + q
  \]
  where $a\in \act$; $\cero$ represents the process that performs no action;
  for every action in $\act$, there is a prefix operator; and  $+$ is
  a choice operator.
\end{defi}

\begin{figure}[t]%
\[%
\begin{array}{ccccc}
 ap\tran{a}p &\mbox{\qquad} & \frac{\displaystyle p\tran{a}p'}{\displaystyle
 p+q\tran{a}p'} &\mbox{\qquad} & \frac{\displaystyle q\tran{a}q'}{\displaystyle
  p+q\tran{a}q'}\\
\end{array}
\]
\caption{Operational semantics for BCCSP terms.}
\label{fig:OperationalSemanticsBCCSP}
\end{figure}

The operational semantics for BCCSP terms is defined in 
Figure~\ref{fig:OperationalSemanticsBCCSP}. 
As usual, we write $p\tran{a}\:$ if
there exists a process $q$ such that $p\tran{a}q$, and $p\Tran{\alpha} q$
if $\alpha = a_1\dots a_n$ and $p\tran{a_1}p'\tran{a_2}\dots\tran{a_n}q$.
The initial offer of a process is the set  
$\init(p)=\{a\barra a \in \act\;\mbox{and}\; p\tran{a}\}$. This is 
a simple, but quite important observation function that plays a central
role in the definition of the most popular semantics in the 
linear time-branching time (ltbt) spectrum. We will also denote by $I$
the relation expressing the fact that two processes have the same
initial offers: $pIq\Leftrightarrow \init(p)=\init(q)$.

One way to capture semantics is by means of the
equivalence relation induced by it: given a formal semantics
$\lsem\cdot\rsem_Z$, we say that processes $p$ and $q$ are
equivalent iff they have the same semantics, that is, $p\equiv_Z q
\Leftrightarrow \lsem p\rsem_Z=\lsem q\rsem_Z$.
These semantics can be defined by means of adequate observational 
scenarios, or by logical characterizations that induce natural 
preorders $\sqsubseteq_Z$ whose kernels are the semantic equivalences.
We refer to \cite{Gla01} for the original definition and usual 
notation for all the semantics in the ltbt spectrum that will be discussed
throughout the paper.

To properly express equations or inequations within the process language, we introduce variables from any adequate set $V$\!\!, and consider the extended set BCCCSP($\act,V$) of terms including variables in $V$.

\begin{figure}[tbp]
\[
  \begin{array}{lll@{\qquad}lll}
    (B_1) & & x+y \simeq y+x   &       (B_3) & & x + x \simeq x \\
    (B_2) & & (x+y)+z \simeq x+(y+z) & (B_4) & & x+\cero \simeq  x 
  \end{array}
\]
  \caption{The axiomatization for the (strong) bisimulation equivalence.}
  \label{fig:axioms4bisimulation}
\end{figure}

Some of the semantics in the spectrum are concrete examples of the 
general notion of constrained simulation semantics
that can be defined in a parameterized way.

\begin{defi}
 \label{def:constrainedsimulation}
 Given a relation $N$ over BCCSP processes, a relation $S_N$ is
 an \emph{$N$-constrained simulation} if $pS_N q$ implies:
 \begin{iteMize}{$\bullet$}
 \item for every $a\in\act$, if $p\tran{a}p'$ then there exists some $q'$ such that
  $q\tran{a} q'$ and $p' S_N q'$, and
 \item $pNq$.
 \end{iteMize}
  We say that process $p$ is $N$-simulated by process $q$, or
   that $q$ $N$-simulates $p$, written
   $p\sqsubseteq_\mathit{NS} q$, whenever there exists an $N$-constrained
   simulation $S_N$ such that $p S_N q$.
\end{defi}

We have already studied the constrained simulation semantics in detail in
\cite{FG08ifiptcs}, stressing their general properties.
In particular, the following constraints were considered: 
\begin{iteMize}{$\bullet$}
\item the universal relation $U$ relating all processes, which gives rise to 
 the simulation semantics;
\item the relation $C$, which holds for processes $p$ and $q$ when both,
 or none, are isomorphic to $\cero$, and that gives rise to the
 complete simulation semantics;
\item the relation $I$ relating processes with the same initial offer, which is 
 the constraint for ready simulation;
\item the relation $T$, that holds for processes having the same set of 
  traces and gives rise to the trace simulation semantics;
\item the relation $S$, the inverse of the simulation equivalence relation, whose associated
 constrained simulation is the 2-nested simulation.
\end{iteMize}

\noindent Throughout this paper there appear different order relations.
We use $\sqsubseteq$ to denote semantic preorders and,
for the sake of simplicity, we use the symbol $\sqsupseteq$ to represent the 
preorder relation $\sqsubseteq^{-1}$. With $\equiv$ we 
denote the induced equivalence (that is, $\sqsubseteq\cap\sqsupseteq$).
To refer to a specific preorder 
we shall append the initials of its name as
subscripts to the symbol $\sqsubseteq$ ($\sqsubseteq_{RS}$ for ready
simulation, $\sqsubseteq_{F}$ for failures, and so on). A similar
convention applies to the kernels of the preorders ($\equiv_{RS}$,
$\equiv_{F},$ \ldots) and to the bisimulation equivalence
$\equiv_B$.
 

\newcommand{\ale}{\preceq}
\newcommand{\aeq}{\simeq} %

An {\em inequation} (respectively, an {\em equation}) over the
language BCCSP is a formula of the form $t \ale u$ (respectively, $t
\aeq u$), where $t,u\in$ BCCSP($\act,V$). An {\em (in)equational
axiom system} is a set of (in)equations over the language BCCSP.  An
equation $t \aeq u$ is derivable from an equational axiom system $E$,
written $E\vdash t \aeq u$, if it can be proven from the axioms in $E$
using the rules of equational logic (viz.~reflexivity, symmetry,
transitivity, substitution and closure under BCCSP contexts):
\[
t \aeq t \quad \frac{t \aeq u}{u \aeq t} \quad \frac{{t \aeq u \quad u \aeq 
v}}{{t \aeq v}} \quad
\frac{{t \aeq u}}{{\sigma(t) \aeq \sigma(u)}} \quad \frac{t \aeq
  u}{at \aeq au}
\quad
\frac{t \aeq
  u\quad t' \aeq u'}{t+t' \aeq u+u'} 
\] 
where substitutions $\sigma$ are defined and applied as usual.

For the derivation of an inequation $t \ale u$ from an inequational
axiom system $E$, the rule for symmetry---that is, the second rule
above---is omitted. We write $E\vdash t \ale u$ if the inequation $t
\ale u$ can be derived from $E$.

It is well-known that, without loss of generality, one may assume that
substitutions happen first in (in)equational proofs, i.e., that the
fourth rule may only be used when its premise is one of the (in)equations in
$E$.  Moreover, by postulating that for each equation in $E$ its
symmetric counterpart is also present in $E$, one may assume that
applications of symmetry happen first in equational proofs, i.e., that
the second rule is never used.
In the remainder of this paper, we
shall always tacitly assume that equational axiom systems are closed
with respect to symmetry.  Note that, with this assumption, there is
no difference between the rules of inference of equational and
inequational logic. In what follows, we shall consider an equation
$t\aeq u$ as a shorthand for the pair of inequations $t\ale u$ and $u\ale
t$. 

An inequation $t \ale u$ is {\em sound} with respect to a given
preorder relation $\sqsubseteq$, if $t\sqsubseteq u$ holds true.
An (in)equational axiom system $E$ is sound
with respect to $\sqsubseteq$ if so is each (in)equation in $E$.
An (in)equational axiomatization is
called ground-complete if it can prove all the valid (in)equivalences relating
terms with no occurrences of variables. As in~\cite{Gla01},
we abbreviate ground-completeness for completeness because this is the only
kind we use along the paper.

Bisimilarity, the strongest of the semantics in the 
spectrum, can be axiomatized by means of the four simple 
axioms in Figure~\ref{fig:axioms4bisimulation}. These axioms state that
the choice operator is commutative, associative and idempotent, having the
empty process as identity element. These axioms also justify the use of
the notation $\sum_a\sum_i ap^i_a$ for processes, where the
commutativity and associativity of the choice operator is used to
group together the summands whose initial action is $a$.
We will also write $p|_a$ for the (sub)process we get by projecting all the
$a$-summands of $p$; that is, if $p = \sum_a\sum_i ap^i_a$, then 
$p|_a = \sum_i ap^i_a$.

Besides the semantics in the spectrum, we are interested in a general study
that can be applied to any ``reasonable'' semantics coarser than 
bisimilarity. Since we will use preorders to characterize these semantics
we introduce the following definitions that state the desired properties of 
those reasonable preorders.

\begin{defi}
  \label{def:ProcessPreorder}
  A preorder relation $\ltp$ over processes is a
  \emph{behavior preorder} if
  \begin{iteMize}{$\bullet$}
  \item
  it is weaker than bisimilarity, i.e., $p\eip_B q\Rightarrow
  p\ltp q$, and
  \item it is
  a precongruence with respect to the prefix and choice operators,
  i.e., if $p\ltp q$ then $ap\ltp aq$ and $p+r\ltp q+r$.
\end{iteMize}
If $\sqsubseteq$ is actually an equivalence, it is said to be a \emph{behavior
equivalence}.
\end{defi}

Another way of presenting a semantics is by means of a logical characterization. The Hennessy-Milner logic \cite{hm85}, characterizing the bisimulation semantics is the most popular one. 
\begin{defi}[\textbf{Hennessy-Milner logic, HML}]\label{HM_logic}
The set $\mathcal{L}_{HM}$ of Hennessy-Milner logical formulas is defined by:
if  $\varphi$, $\varphi_i \in \mathcal{L}_{HM}$ for all $i \in I$ and $a \in Act$, then $\bigwedge_{i \in I} \varphi_i$, $a\varphi$, $\neg \varphi$ $\in \mathcal{L}_{HM}$.

The satisfaction relation $\models$ is defined by: 
\begin{iteMize}{$\bullet$}
\item $p\models a\varphi$ if there exists $q$ such that
  $p\stackrel{a}{\rightarrow}q$ and $q \models \varphi$;  
\item $p\models \bigwedge_{i\in I} \varphi_i$ if for all $i\in I: p \models \varphi_i$.
\item $p\models \neg \varphi$ if $p\not\models \varphi$.
\end{iteMize}
\end{defi}

Note that $\bigwedge_{i\in\emptyset}\varphi_i \in \mathcal{L}_{HM}$, and we have $p\models \bigwedge_{i\in\emptyset}\varphi_i$ for all \textit{p}. Therefore, in the following we will consider that $\top \in \mathcal{L}_{HM}$, where $\top$ is syntactic sugar for $\bigwedge_{i\in\emptyset}\varphi_i$. The finite version of this logic, $\mathcal{L}^{f}_{HM}$, uses binary conjunction $\wedge$ instead of the general conjunction $\bigwedge_{i \in I}$. It is well-known that $\mathcal{L}^{f}_{HM}$ characterizes the bisimulation semantics between image-finite processes, that are those that do not allow infinite branching for any action $a \in Act$ at any state.
Van Glabbeek uses $\mathcal{L}_{B}$ to refer to $\mathcal{L}_{HM}$ in \cite{Gla01}.

\section{Equational semantics}
\label{sec:EquationalSemantics}

On Tables~\ref{tab:AxiomatisationsPreordersLtbts} and~\ref{tab:AxiomatisationsEquivalencesLtbts} appear the axiomatic characterizations
for the preorders and equivalences in van Glabbeek's spectrum~\cite{Gla01}. For each column, the set of axioms marked with ``+'' are sound and complete with respect to the preorder or equivalence in the head of that column; axioms marked with ``v'' are valid but not needed.
When studying these tables there are several questions that naturally arise: for 
every semantics, is there any connection between the axioms defining the preorder 
and those for the equivalence? 
Can the axiomatizations of some of these semantics be jointly tackled? 

In this section we will develop new axiomatizations for
all the semantics in the ltbt spectrum that offer a clear 
answer to the previous questions: even if there was not a systematic procedure that led to
produce the axiomatizations of those tables, we can obtain equivalent axiomatizations that do follow a given procedure.

These new axiomatizations are obtained after noticing that 
every process semantics can be understood as the product of
two ``design decisions'', decisions that define what we have called the ``dynamic'' and the 
``static'' basis of the semantics.
    We will show that, besides $B_1$--$B_4$, we only need a
    generic simulation axiom $(NS)$---Proposition~\ref{prop:ns}---which characterizes the family of constrained simulation semantics, 
    to axiomatize the whole class of pure branching semantics. Moreover, to characterize 
    the linear time semantics, we only need to add
    to the corresponding simulation axiom the adequate instantiation of 
    a generic axiom $(ND)$---see page~\pageref{eq:nd}---for reducing the observability of non-determinism 
    in   processes, by means of which we introduce the additional identifications 
    induced by each of the linear semantics.

    Also the axiomatizations between orders and equivalences are closely related;
    in fact, in the case of the linear semantics we could just use an equivalence $(ND_\equiv)$ axiom,
    leaving the order or equivalence aspect to be determined by the use of the order or equivalence
    axiom of the corresponding branching semantics, see Figure~\ref{fig:ReadySlice2}.
    
    In order to justify the form of our axiomatizations without leaving the 
    axiomatic framework, in this section we prove our results with 
    separate and ad-hoc proofs for each semantics just comparing the
    new characterizations with those previously known. This allows us
    to quickly get the taste of the underlying relations 
    of the process semantics. Once the unified observational characterization 
    of semantics is presented in Section~\ref{observational-sem-sec},
    we will provide generic proofs
    for these results in Section~\ref{rnoef:sec} that show the
    suitability of the new axiomatizations with respect to the
    observational characterizations of the semantics.
    
\begin{table}
\centering
\begin{tabular}{|l
      |>{\centering}p{.20cm}
      |>{\centering}p{.25cm}
      |>{\centering}p{.33cm}
      |>{\centering}p{.3cm}
      |>{\centering}p{.3cm}
      |>{\centering}p{.24cm}
      |>{\centering}p{.24cm}
      |>{\centering}p{.27cm}
      |>{\centering}p{.27cm}
      |>{\centering}p{.20cm}
      |c|}
    \cline{2-12}
\multicolumn{1}{l|}{}&B&$\!\!\mbox{RS}$&$\!\!\mbox{PW}$&
  $\!\mbox{RT}$&$\!\mbox{FT}$&R&F&$\!\!\mbox{CS}$&$\!\!\mbox{CT}$&S&T\\ \hline
$(x+y)+z\simeq x+(y+z)$ &+&+&+&+&+&+&+&+&+&+&+\\
$x+y\simeq y+x$ &+&+&+&+&+&+&+&+&+&+&+\\
$x+0\simeq x$ &+&+&+&+&+&+&+&+&+&+&+\\
$x+x\simeq x$ &+&+&+&+&+&+&+&+&+&+&+\\
$ax\preceq ax+ay$ & &+&+&+&+&+&+&v&v&v&v\\
$a(bx+by+z)\simeq a(bx+z)+a(by+z)$ & & &+&v&v&v&v& &v& &v\\
$I(x)=I(y)\Rightarrow ax+ay \simeq  a(x+y)$ & & & &+&v&v&v& &v& &v\\
$ax+ay\succeq a(x+y)$ & & & & &+& &v& &v& &v\\
$a(bx+u)+a(by+v)\succeq a(bx+by+u)$ & & & & & &+&v& &v& &v\\
$ax+a(y+z)\succeq a(x+y)$ & & & & & & &+& &v& &v\\
$ax\preceq ax+y$ & & & & & & & &+&+&v&v\\
$a(bx+u)+a(cy+v)\simeq a(bx+cy+u+v)$ & & & & & & & & &+& &v\\
$x\preceq x+y$ & & & & & & & & & &+&+\\
$ax+ay\simeq a(x+y)$ & & & & & & & & & & &+\\ \hline
  \end{tabular}
  \caption{Axiomatization for the
    preorders in the linear time-branching  time spectrum.}
  \label{tab:AxiomatisationsPreordersLtbts}
\end{table}

\subsection{A new axiomatization of the most popular semantics}
\label{anamps:sec}

We start our study with a very representative and well-known group of semantics 
in the spectrum, each of which 
has been developed and used in important work in the area:
ready simulation~\cite{LS89,BIM95}, 
failures~\cite{BHR84,Hoa85CSP}, 
readiness \cite{OH86}, 
ready trace~\cite{BBK87} and failure traces \cite{Phi87}.

\subsubsection{Semantic preorders}
\label{subsub:preorders}
As already hinted above, the dynamic part of the 
semantics is inherited from a simulation preorder.
As stated in our Introduction, bisimilarity can be axiomatized by the set of axioms $B1-B4$. All the other semantics in the spectrum are coarser than it, and therefore also satisfy these axioms. But due to the fact that bisimulations define equivalence relations and not just preorders, we cannot base on them the characterization of any other interesting semantics. But, plain simulations are somehow defined as half-bisimulations, and can indeed be used as support for the characterizations of some interesting semantics, such as trace semantics. Nevertheless, plain similarity becomes too weak, and some other finer class of simulations is needed to support the characterization of the interesting semantics listed above. Next we recall the axiomatizations of plain, ready and general constrained similarity.

\begin{prop}[\cite{Gla01,FG08ifiptcs}]
 \label{prop:ns}
 \hfill
\begin{enumerate}[\em(1)]
\item 
Plain similarity can be axiomatically defined by means of the axiom
$(S)\;x \preceq x+y$, together with the axioms $B_1$--$B_4$ that define
bisimilarity.
\item 
Ready similarity can be axiomatically defined by means of the
conditional axiom $(RS)\; xIy\Rightarrow x\preceq x+y$, together with
$B_1$--$B_4$.
It can also be axiomatized by means of the axiom scheme
$ax \preceq ax + ay$, where $a$ represents any arbitrary action.
\item Whenever $N$ is a behavior preorder, $N$-similarity can be 
axiomatically defined by means of the 
conditional axiom $(\textit{NS})\; N(x,y)\Rightarrow x\preceq x + y$, together 
with $B_1$--$B_4$.
\end{enumerate}
\end{prop}

\noindent Let us now consider the diamond of semantics coarser than ready similarity 
in the ltbt spectrum.
It consists of the failures, readiness, failure trace, and ready trace
semantics.
None of them is a simulation semantics, so their classic axiomatizations 
(see Table~\ref{tab:AxiomatisationsPreordersLtbts}) contain an additional axiom:
\[
\begin{array}{l@{\ }l@{\quad}l}
\textrm{Failures:} & (F) & a(x+y) \preceq ax + a(y+w)\\
\textrm{Readiness:}& (R)& a(bx+by+u) \preceq a(bx+u) + a(by+v)\\
\textrm{failure trace:} &(FT)& a(x+y) \preceq ax + ay\\
\textrm{ready trace:}& (RT)& I(x)=I(y)\Rightarrow ax+ay\simeq a(x+y)
\end{array}
\]

Since we are interested in capturing the reduction of observability of non-determinism, our first
candidate for a general axiom covering all cases was $(FT)$, which captures
the fact that by delaying the choices we get ``smaller'' processes.
However, since this axiom characterizes the failure trace semantics and this
is finer than failure semantics, a more general axiom is needed: axiom $(F)$ 
became our next proposal because failure semantics is the coarsest of the four 
semantics.
More precisely, we expected to achieve the axiomatization of the four
semantics in the diamond by adding the adequate instance of 
the generic constrained conditional axiom
\[\label{eq:nd}
(\textit{ND\/})\quad M(x,y,w)\Rightarrow a(x+y) \preceq ax + a(y+w)\,.  
\]
This seemed reasonable since the other semantics in the group are finer than failures
and by adding a constraint to $(F)$ we certainly obtain a more restricted axiom
that produces a finer preorder. The conjecture turned out to be correct and we found 
that the semantics in the diamond can
be characterized by the following instances:
\[
\begin{array}{l@{\quad}l}
(\textit{ND\/}^F) & M_F(x,y,w) \iff \textrm{true}\\
(\textit{ND\/}^R)& M_R(x,y,w) \iff I(x)\supseteq I(y)\\
(\textit{ND\/}^\mathit{FT})& M_{FT}(x,y,w) \iff I(w) \subseteq I(y) \\
(\textit{ND\/}^\mathit{RT})& M_{RT}(x,y,w) \iff  I(x) = I(y) \ \textrm{and}\ 
I(w)\subseteq I(y) 
\end{array}
\]

Since $M_F$ is the universal relation containing all triples of processes,
the corresponding instance of the conditional axiom $(\textit{ND\/})$ is
clearly equivalent to $(F)$, and thus adding it to the set 
$\{\textrm{$B_1$--$B_4$}, (RS)\}$ we
obtain a ground-complete axiomatization of $\sqsubseteq_F$.
Let us now prove that the remaining three semantics are also axiomatized by
the corresponding instances of the axiom $(\textit{ND\/})$ together with
$(RS)$.

\begin{prop}\label{static-1:prop}
\hfill
\begin{enumerate}[\em(1)]
\item The readiness preorder $\sqsubseteq_R$ is axiomatized by 
 $\{\textrm{$B_1$--$B_4$}, (RS),$ $(\textit{ND\/}^R)\}$.
\item The failure trace preorder $\sqsubseteq_{FT}$ is axiomatized by
 $\{\textrm{$B_1$--$B_4$}, (RS), (\textit{ND\/}^\mathit{FT})\}$.
\item The ready trace preorder is axiomatized by the set
 $\{\textrm{$B_1$--$B_4$}, (RS), (\textit{ND\/}^\mathit{RT})\}$.
\end{enumerate}
\end{prop}
\begin{proof}
\hfill
\begin{enumerate}[(1)]
\item Let us show that the set 
$\{\textrm{$B_1$--$B_4$}, (RS), (\textit{ND\/}^R)\}$ is logically equivalent to
$\{\textrm{$B_1$--$B_4$},$ $(RS), (R)\}$.
By taking $x = bx' + u$, $y= by'$, and $w=v$ we have that 
$(\textit{ND\/}^R)$ implies $(R)$.
In the other direction, let $x$ and $y$ be arbitrary closed BCCSP terms with
$I(y) \subseteq I(x)$: we will prove, by structural induction on $y$, that
$\{\textrm{$B_1$--$B_4$}, (RS), (R)\} \vdash a(x+y) \preceq ax + a(y+w)$, for
any term $w$.
\begin{iteMize}{$\bullet$}
\item For $y=\cero$, we have $a(x+y) \simeq ax\preceq ax+a(y+w)$, by application
of $(RS)$.
\item For $y= by'+ y''$, it must be $x=bx'+x''$ and taking $v= y''+w$ in 
$(R)$ we obtain $a(x+y) = a(bx'+by'+x''+y'') \preceq a(x+y'') + a(y+w)$.
Then we have $I(y'')\subseteq I(x)$ and we can apply the induction hypothesis
to get $\{\textrm{$B_1$--$B_4$}, (RS), (R)\} \vdash a(x+y)\preceq ax + a(y+w)$.
\end{iteMize}

\item Let us show that the sets 
$\{\textrm{$B_1$--$B_4$}, (RS),(\textit{ND\/}^\mathit{FT})\}$ and
$\{\textrm{$B_1$--$B_4$}, (RS),(FT)\}$
are logically equivalent.
The implication from left to right follows by taking $w=\cero$.
In the other direction, let $w$ and $y$ with $I(w)\subseteq I(y)$, so that
$a(x+y) \preceq ax+ay$ using $(FT)$ and, since $I(y) = I(y+w)$, we have 
$y\preceq y+w$ using $(RS)$: hence we conclude, $a(x+y) \preceq ax+a(y+w)$.

\item Let us show that the set 
$\{\textrm{$B_1$--$B_4$}, (RS), (\textit{ND\/}^\mathit{RT})\}$ 
is logically equivalent to
$\{\textrm{$B_1$--$B_4$}, (RS),$ $(RT)\}$.
We first note that $\{\textrm{$B_1$--$B_4$}, (RS), (RT)\}$ is equivalent
to $\{\textrm{$B_1$--$B_4$}, (RS), (RT_{\succeq})\}$, where $(RT_\succeq)$ is
the axiom $M_{RT}(x,y,w)\Rightarrow ax+ay \succeq a(x+y)$.
This follows from the fact that, whenever $I(x) = I(y)$, we can use $(RS)$
to get $x\preceq x+y$ and $y\preceq x+y$, and then $ax+ay\preceq a(x+y)$.
Now, the implication from left to right follows by taking $w=\cero$.
From right to left, as above, whenever $I(w)\subseteq I(y)$ we have 
$y\preceq y+w$ and then, if $I(x)=I(y)$ we have $a(x+y)\preceq ax+ay$, and
therefore $a(x+y)\preceq ax+a(y+z)$.%
\end{enumerate}%
\end{proof}

\begin{figure}[tbp]
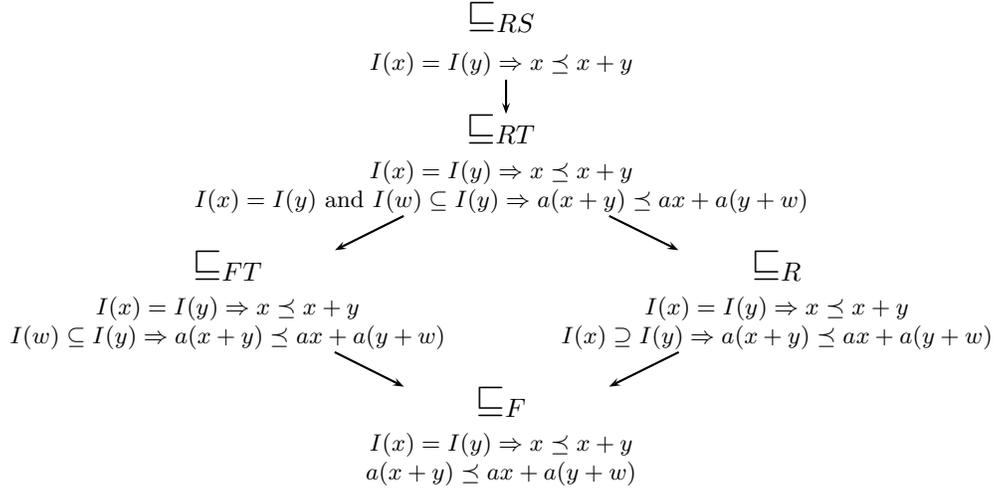

  \begin{center}
    \begin{tabular}{ccc} 
      \multicolumn{3}{c}{\rnode{ns}{%
          \begin{tabular}{c}    
            \Large{$\sqsubseteq_{RS}$}\\[.25em]
            \footnotesize{$\begin{array}{c}
                I(x)=I(y) \Rightarrow x \preceq x + y \end{array}$}
            \end{tabular}
          }} 
      \\ \\
      \multicolumn{3}{c}{\rnode{nl}{%
          \begin{tabular}{c}    
            \Large{$\sqsubseteq_{RT}$}\\[.25em]
            \footnotesize{$\begin{array}{c}
                I(x)=I(y) \Rightarrow x \preceq x + y\\ 
                I(x)=I(y)\mbox{ and } I(w)\subseteq I(y)\Rightarrow a(x+y) \preceq ax + a(y+w)\end{array}$} 
          \end{tabular}
        }}
      \\ \\
      \rnode{nlf}{%
        \begin{tabular}{c}    
          \Large{$\sqsubseteq_{FT}$}\\[.25em]  
          \footnotesize{$\begin{array}{c}
              I(x)=I(y) \Rightarrow x \preceq x + y\\ 
              I(w)\subseteq I(y)\Rightarrow a(x+y) \preceq ax + a(y+w)
            \end{array}$}
        \end{tabular}
      }
      &  &
      \rnode{nlr}{%
        \begin{tabular}{c}    
          \Large{$\sqsubseteq_{R}$}\\[.25em] 
          \footnotesize{$\begin{array}{c} 
              I(x)=I(y) \Rightarrow x \preceq x + y\\ 
              I(x) \supseteq I(y)\Rightarrow a(x+y) \preceq ax + a(y+w)
            \end{array}$}
          \end{tabular}
          }
        \\ \\
        \multicolumn{3}{c}{\rnode{nlfr}{%
          \begin{tabular}{c}    
            \Large{$\sqsubseteq_{F}$}\\[.25em]       
            \footnotesize{$\begin{array}{c}
                I(x)=I(y) \Rightarrow x \preceq x + y\\ 
                a(x+y) \preceq ax + a(y+w)
              \end{array}$}
          \end{tabular}
        }}\\
      \ncline{->}{ns}{nl}
      \ncline{->}{nl}{nlr}
      \ncline{->}{nl}{nlf}
      \ncline{->}{nlr}{nlfr}
      \ncline{->}{nlf}{nlfr}
      \\
    \end{tabular}

    \caption{Inclusion relation for the ready simulation preorder and its associated linear semantics.}
    \label{fig:ReadySlice}
  \end{center}
\end{figure}

Figure~\ref{fig:ReadySlice} shows the already known relations between the
semantics of the spectrum in the ready simulation layer. However, we want to
stress the fact that once the new axiomatizations are proved to be correct,
those relations became obvious since the four constraints defined above
trivially satisfy $M_{RT}(x,y,w) \Rightarrow M_{FT}(x,y,w) \wedge M_R(x,y,w) $
and $M_{FT}(x,y,w) \vee M_R(x,y,w) \Rightarrow M_F(x,y,w) $. It is even more
important that the tight relations and the subtle differences between these
semantics clearly stand out by just looking at their axiomatizations.

Certainly, if we compare our new axiomatizations and those in
Table~\ref{tab:AxiomatisationsPreordersLtbts}, the use of conditions in our
axioms could be on the grounds that complex conditions could be used to hide
the complexity of the semantics.
However, the conditions that we have introduced for the alternative
axiomatizations of the semantics in the spectrum are very simple.
In any case, our main interest was to obtain a uniform presentation of
the axiomatizations that could be used to simplify their generic algebraic
study.

\begin{cor}
\hfill
\begin{enumerate}[\em(1)]
\item $\sqsubseteq_{FT}$ is axiomatized by the set 
$\{\textrm{$B_1$--$B_4$}, (RS), (\textit{ND\/}^\mathit{FT}_0)\}$, where 
$(\textit{ND\/}^\mathit{FT}_0)$ is the instance of $(\textit{ND\/}^\mathit{FT})$
where $w$ is $\cero$.

\item $\sqsubseteq_{RT}$ is axiomatized by 
$\{\textrm{$B_1$--$B_4$}, (RS), (\textit{ND\/}^\mathit{RT}_0)\}$, where 
$(\textit{ND\/}^\mathit{RT}_0)$ is the instance of $(\textit{ND\/}^\mathit{RT})$ 
where $w$ is $\cero$.
\end{enumerate}
\end{cor}
\begin{proof}
Note that for the proof of Proposition~\ref{static-1:prop} only the 
case $w=\cero$ is needed.
\end{proof}

Even if the simplifications above are possible, we prefer to maintain the 
general forms of the axioms $(\textit{ND\/}^\mathit{FT})$ and
$(\textit{ND\/}^\mathit{RT})$ to keep all axiomatizations as similar as 
possible, which will come in handy when proving general properties of the 
semantics.

\begin{cor}
\hfill
\begin{enumerate}[\em(1)]
\item $\sqsubseteq_F$ can be axiomatized by the axioms 
$\{\textrm{$B_1$--$B_4$}, (\textit{ND\/}^F)\}$.
\item $\sqsubseteq_R$ can be axiomatized by the axioms 
$\{\textrm{$B_1$--$B_4$}, (\textit{ND\/}^R)\}$.
\end{enumerate}
\end{cor}
\begin{proof}
Note that $(\textit{ND\/}^F)$ implies $(RS)$ and therefore 
$(\textit{ND\/}^R)$ implies $(RS)$, by taking $y=\cero$ and $w=y$.
\end{proof}

\subsubsection{Equivalences and their preorders}
\label{subsub:equivalences}

Let us now study the equivalences and first of all note that the axiom 
(\textit{ND}) controlling the reduction of non-determinism has been 
presented as an inequational axiom.
Certainly, it cannot simply be replaced by the corresponding equation since, in 
general, it is not true that $ax + ay\simeq a(x+y)$.
However, the two dimensions corresponding to $(RS)$ and $(\textit{ND}^Z)$ 
that control the ``growth'' of a process 
with respect to a preorder $\preceq$ are not orthogonal; for example, 
$a(x+y)\preceq a(x+y) + ax$ can be derived either by an application of 
$(\textit{ND}^\mathit{FT})$ or by one of $(RS)$.
As a consequence of the relation between these two axioms, once $(RS)$ is assumed
then the inequational axiom $(\textit{ND})$ can be substituted by its
(stronger) equational form 
\[
(\textit{ND}_\equiv)\quad M(x,y,w)\Rightarrow ax + a(y+w) + a(x+y)\simeq
                                            ax + a(y+w)\,.
\]
As above, we write $(\textit{ND}^Z_\equiv)$ for the concrete instances of this
axiom for $Z\in \{F, R,$ $\textit{FT},\textit{RT}\}$.

\begin{table}[t]
  \centering
  \begin{tabular}{|l
      |>{\centering}p{.20cm}
      |>{\centering}p{.25cm}
      |>{\centering}p{.25cm}
      |>{\centering}p{.25cm}
      |>{\centering}p{.25cm}
      |>{\centering}p{.25cm}
      |>{\centering}p{.25cm}
      |>{\centering}p{.25cm}
      |>{\centering}p{.25cm}
      |>{\centering}p{.20cm}
      |c|}
    \cline{2-12}
\multicolumn{1}{l|}{}&B&$\!\!\mbox{RS}$&$\!\!\!\mbox{PW}$&
  $\!\!\mbox{RT}$&$\!\!\mbox{FT}$&R&F&$\!\!\mbox{CS}$&$\!\!\mbox{CT}$&S&T\\
\hline
$(x+y)+z\simeq x+(y+z)$ &+&+&+&+&+&+&+&+&+&+&+\\
$x+y\simeq y+x$ &+&+&+&+&+&+&+&+&+&+&+\\
$x+0\simeq x$ &+&+&+&+&+&+&+&+&+&+&+\\
$x+x\simeq x$ &+&+&+&+&+&+&+&+&+&+&+\\
$I(x)=I(y)\Rightarrow a(x+y)\simeq a(x+y)+ay$ & &+&v&v&v&v&v&v&v&v&v\\
$a(bx+by+z) \simeq a(bx+z)+a(by+z)$ & & &+&v&v&v&v& &v& &v\\
$I(x)=I(y)\Rightarrow ax+ay \simeq a(x+y)$ & & & &+&+&v&v& &v& &v\\
$ax+ay\simeq ax+ay+a(x+y)$ & & & & &+& &v& &v& &v\\
$a(bx+u)+a(by+v)\simeq a(bx+by+u)+a(by+v)$ & & & & & &+&+& &v& &v\\
$ax+a(y+z)\simeq ax+a(x+y)+a(y+z)$ & & & & & & &+& &v& &v\\
$a(x+by+z)\simeq a(x+by+z)+a(by+z)$ & & & & & & & &+&v&v&v\\
$a(bx+u)+a(cy+v) \simeq a(bx+cy+u+v)$ & & & & & & & & &+& &v\\
$a(x+y)\simeq a(x+y)+ay$ & & & & & & & & & &+&v\\
$ax+ay \simeq a(x+y)$ & & & & & & & & & & &+\\ \hline
  \end{tabular}
  \caption{Axiomatization for the equivalences in the linear time-branching time spectrum.}
  \label{tab:AxiomatisationsEquivalencesLtbts}
\end{table}

\begin{prop}\label{rnd-eq:prop}
\hfill
\begin{enumerate}[\em(1)]
\item 
The set
$\{\textrm{$B_1$--$B_4$}, (RS), (\textit{ND})\}$ is logically equivalent to
$\{\textrm{$B_1$--$B_4$},(RS), (\textit{ND}_+)\}$, where 
$(\textit{ND}_+)$ is the axiom
\[
M(x,y,w)\Rightarrow ax + a(y+w) + a(x+y)\preceq ax + a(y+w)\,.
\]
\item $\{\textrm{$B_1$--$B_4$}, (RS), (\textit{ND}_+)\}$ is logically equivalent
to $\{\textrm{$B_1$--$B_4$}, (RS), (\textit{ND}_\equiv)\}$.
\end{enumerate}
\end{prop}
\begin{proof}
\hfill
\begin{enumerate}[(1)]
\item We only need to prove the implication from right to left, since the
other follows from $\preceq$ being a precongruence.
For that, from $(RS)$ we get $a(x+y)\preceq a(x+y) + ax+ a(y+w)$ whence,
using $(\textit{ND}_+)$, $a(x+y)\preceq ax + a(y+w)$.
\item We only need to prove that, if $M(x,y,w)$, then
\[
\{\textrm{$B_1$--$B_4$}, (RS), (\textit{ND}_+)\}\vdash 
ax+ a(y+w) \preceq ax + a(y+w) + a(x+y)\,,
\]
which follows from $(RS)$.
\end{enumerate}
\end{proof}

\noindent This result can be interpreted as saying that the only way to
``enlarge'' a  
process is by extending its possible behaviors by means of the ``dynamic''
simulation axioms; the static rules, $(\textit{ND})$ and its variants, instead
generate new identifications among processes.

Actually, any complete axiomatization of a preorder that contains the axiom
$(RS)$ can be turned into an equivalent axiomatization by replacing every
inequality $u\preceq v$ by $u + v\simeq v$.

\begin{prop}
Let $Q = \{\textrm{$B_1$--$B_4$},(RS)\} \cup Q'$ be an axiomatization of 
an order $\sqsubseteq$ such that ${\sqsubseteq}\subseteq I$.
Then, the equational variant of $Q$, $Q^==\{\textrm{$B_1$--$B_4$},(RS)\} 
\cup \{M\Rightarrow u+v\simeq v\mid M \Rightarrow u\preceq v \in Q'\}$ is also
an axiomatization of $\sqsubseteq$.
\end{prop}
\begin{proof}
Analogous to the particular case considered in Proposition~\ref{rnd-eq:prop} 
above. For the sake of clarity we have preferred to present the particular case 
before, because it is easily stated and it corresponds 
to the most important instance of the general result. 
\end{proof}

Finally, to conclude this section we gather in
Table~\ref{tab:AxiomatisationsEquivalencesLtbts} 
axiomatic characterizations for the semantic equivalences that are an
alternative to the classic axioms appearing in~\cite{Gla01}.

\begin{figure}[ht]
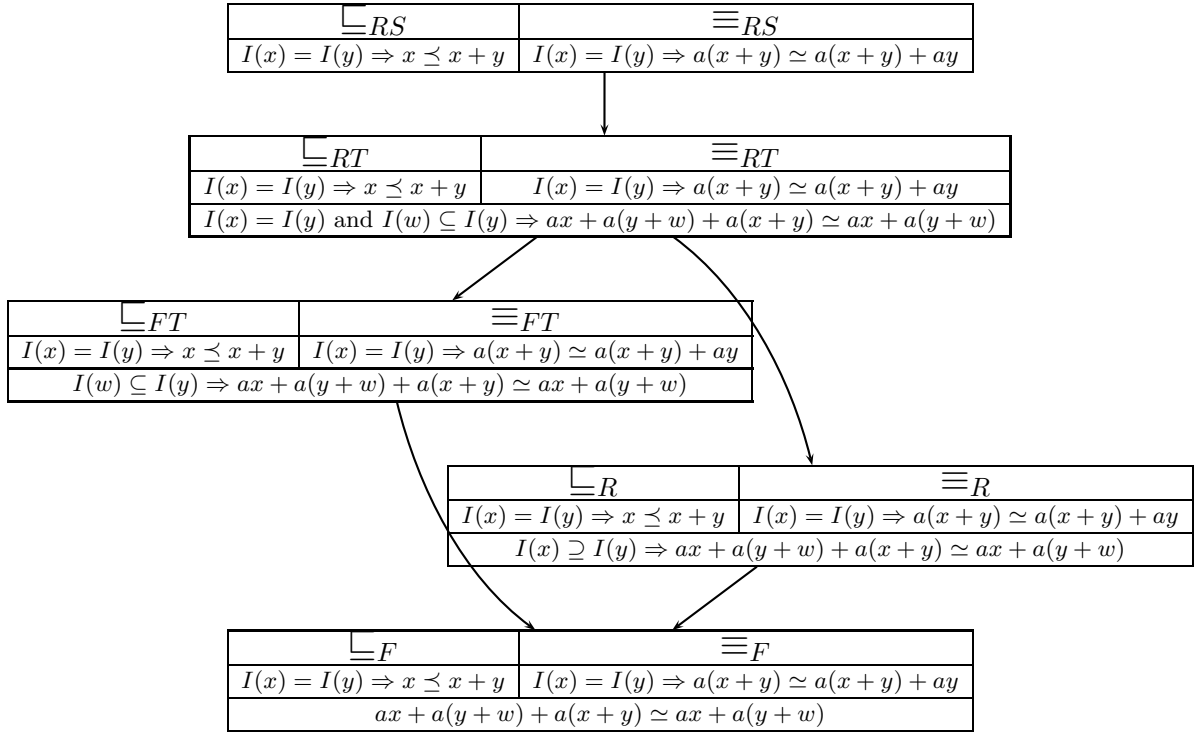
\small
  \begin{center}
    \begin{tabular}{ccc} 
      \multicolumn{3}{c}{\rnode{ns}{%
          \begin{tabular}{|c|c|}    
            \hline
            \Large{$\sqsubseteq_{RS}$}& \Large{$\equiv_{RS}$}\\      
            \hline
            \footnotesize{$I(x)=I(y) \Rightarrow x \preceq x + y$}
            & 
            \footnotesize{$I(x)=I(y) \Rightarrow a(x+y) \simeq a(x + y)+ay$}\\ \hline            
          \end{tabular}   
        }}
      \\ \\ \\ 
      \multicolumn{3}{c}{\rnode{nl}{%
          \begin{tabular}{|c|c|}    
            \hline
            \Large{$\sqsubseteq_{RT}$}& \Large{$\equiv_{RT}$}\\      
            \hline
            \footnotesize{$I(x)=I(y) \Rightarrow x \preceq x + y$}
            & 
            \footnotesize{$I(x)=I(y) \Rightarrow a(x+y) \simeq a(x + y)+ay$}\\ \hline            
            \multicolumn{2}{|c|}{\footnotesize{$  I(x)=I(y)\mbox{ and } I(w)\subseteq I(y)\Rightarrow ax + a(y+w) + a(x+y)\simeq ax + a(y+w)$}}\\             
            \hline
          \end{tabular}               
        }}
      \\ \\ \\
      \multicolumn{2}{l}{
      \rnode{nlf}{%
        \begin{tabular}{|c|c|}    
          \hline
          \Large{$\sqsubseteq_{FT}$}& \Large{$\equiv_{FT}$}\\      
          \hline
            \footnotesize{$I(x)=I(y) \Rightarrow x \preceq x + y$}
            & 
            \footnotesize{$I(x)=I(y) \Rightarrow a(x+y) \simeq a(x + y)+ay$}\\ \hline            
            \multicolumn{2}{|c|}{\footnotesize{$I(w)\subseteq I(y)\Rightarrow ax + a(y+w) + a(x+y)\simeq ax + a(y+w)$}}\\             
            \hline
          \end{tabular}
      }} & \hspace*{5.5cm}
      \\ \\ \\
      \hspace*{5.5cm} & 
      \multicolumn{2}{r}{\rnode{nlr}{%
        \begin{tabular}{|c|c|}    
            \hline
            \Large{$\sqsubseteq_{R}$}& \Large{$\equiv_{R}$}\\      
            \hline
            \footnotesize{$I(x)=I(y) \Rightarrow x \preceq x + y$}
            & 
            \footnotesize{$I(x)=I(y) \Rightarrow a(x+y) \simeq a(x + y)+ay$}\\ \hline            
            \multicolumn{2}{|c|}{\footnotesize{$I(x)\supseteq I(y)\Rightarrow ax + a(y+w) + a(x+y)\simeq ax + a(y+w)$}}\\             
            \hline
          \end{tabular}
          }}
        \\ \\ \\ 
        \multicolumn{3}{c}{\rnode{nlfr}{%
          \begin{tabular}{|c|c|}    
            \hline
            \Large{$\sqsubseteq_{F}$}& \Large{$\equiv_{F}$}\\      
            \hline
            \footnotesize{$I(x)=I(y) \Rightarrow x \preceq x + y$}
            & 
            \footnotesize{$I(x)=I(y) \Rightarrow a(x+y) \simeq a(x + y)+ay$}\\ \hline            
            \multicolumn{2}{|c|}{\footnotesize{$ax + a(y+w) + a(x+y)\simeq ax + a(y+w)$}}\\             
            \hline
          \end{tabular}
        }}\\
      \ncline{->}{ns}{nl}
      \ncarc[arcangle=20]{->}{nl}{nlr}
      \ncline{->}{nl}{nlf}
      \ncline{->}{nlr}{nlfr}
      \ncarc[arcangle=20]{<-}{nlfr}{nlf}
    \end{tabular}

    \caption{Axioms for the ready simulation layer of semantics.}
    \label{fig:ReadySlice2}
  \end{center}
\end{figure}

Following the same ideas that we have already discussed for the preorders,
a key point is to find the equations that characterize the simulation
equivalence that governs each layer.
As showed in~\cite{FG08ifiptcs}, there 
is a generic axiom that we can use:
$$(NS_\equiv)\quad N(x,y)\Rightarrow a(x+y) \simeq a(x+y) + ay.$$

We consider the instantiated equation that 
characterizes the ready simulation equivalence:
\[
(\textit{RS}_\equiv)\quad I(x)=I(y) \Rightarrow a(x + y ) \simeq a(x + y)+ay,
\]
and the rest of the characterization follows by using the equation 
$(ND_\equiv)$ presented above.
\begin{prop}
\label{prop:AltAxEquiv}
\hfill
\begin{enumerate}[\em(1)]
\item The failure equivalence  $\equiv_F$ is axiomatized by 
 $\{\textrm{$B_1$--$B_4$}, (RS_\equiv),$ $(\textit{ND\/}_\equiv^F)\}$.
\item The readiness equivalence $\equiv_R$ is axiomatized by 
 $\{\textrm{$B_1$--$B_4$}, (RS_\equiv),$ $(\textit{ND\/}_\equiv^R)\}$.
\item The failure trace equivalence $\equiv_{FT}$ is axiomatized by
 $\{\textrm{$B_1$--$B_4$}, (RS_\equiv), (\textit{ND\/}_\equiv^\mathit{FT})\}$.
\item The ready trace equivalence $\equiv_{RT}$ is axiomatized by the set
 $\{\textrm{$B_1$--$B_4$}, (RS_\equiv), (\textit{ND\/}_\equiv^\mathit{RT})\}$.
\end{enumerate}
\end{prop}
\begin{proof}
  To prove these results we can compare the new and old axiomatizations
  similarly as we did in the proof of Proposition~\ref{rnd-eq:prop} or
  alternatively  
make use of the ``ready to preorder'' algorithm
thoroughly studied in~\cite{AFI07,FGP08fics,DeFrutosEtAl08b}.
\end{proof}

The results in this section clarify the entanglement
between axiomatizations for preorders and equivalences. For example: for 
the ready simulation and its associated linear semantics, we just need 
three axioms $(RS)$, $(RS_\equiv)$ and 
$(ND_\equiv)$---conveniently instantiated---to characterize the 10 
relations (orders and equivalences) involved, as summarized in
 Figure~\ref{fig:ReadySlice2}.

\subsection{The coarsest semantics in the spectrum}
\label{tcss:sec}

The results in Section~\ref{anamps:sec} show the relations between
the ready simulation and the linear semantics naturally associated to
it. The same phenomenon occurs for other simulations. In this
section we  focus on the bottom part of the spectrum where lie the simulation semantics coarser than ready simulation: 
plain and complete simulation, and the 
semantics coarser than these.
For the simulation semantics we obtain the corresponding axiomatizations
simply by considering the universal
constraint for the case of plain simulations and the complete constraint
for complete simulations:
\[
\begin{array}{l@{\quad}l}
\textrm{Simulation}& U(x,y) \iff \textit{true}\\
\textrm{Complete simulations} & C(x,y) \iff (x=\cero \textrm{ iff } y=\cero)
\end{array}
\]
Trace and completed trace semantics can be
defined by simply adding our axiom $(\textit{ND\/}^F)$ to the appropriate 
instance of
\[
(\textit{NS})\quad N(x,y)\Rightarrow x\preceq x + y.
\]

\begin{prop}
\hfill
\begin{enumerate}[\em(1)]
\item $\sqsubseteq_T$ is axiomatized by the axioms\footnote{Note that
$(S)$ is equivalent to $(\textit{US\/})$, the instantation of
$(\textit{NS\/})$ with $U$ as $N$.}
$\{\textrm{$B_1$--$B_4$}, (S), (\textit{ND\/}^F)\}$.
\item $\sqsubseteq_{CT}$ is axiomatized by the axioms 
$\{\textrm{$B_1$--$B_4$}, (CS), (\textit{ND\/}^F)\}$,
where $(CS)$ is the instantiation of $(\textit{NS})$ taking $C(x,y)$ as $N(x,y)$.
\end{enumerate}
\end{prop}
\begin{proof}
\hfill
\begin{enumerate}[(1)]
\item The classic axiomatization of trace semantics is given by
$\{\textrm{$B_1$--$B_4$}, (S),$ $(T)\}$, where $(T)$ is the axiom 
$ax + ay \simeq a(x+y)$.
Note that $\{\textrm{$B_1$--$B_4$}, (S), (T)\}$ is logically equivalent
to $\{\textrm{$B_1$--$B_4$}, (S), (T_\sqsubseteq)\}$, where $(T_\sqsubseteq)$ is
the axiom $a(x+y)\preceq ax+ay$, because $(S)$ can be used to obtain
$ax \preceq a(x+y)$ and $ay\preceq a(x+y)$.
And it is immediate that $(\textit{ND\/}^F)$ implies $(T_\sqsubseteq)$.
Also, $\{(S), (T_\sqsubseteq)\}\vdash a(x+y)\preceq ax+ a(y+w)$, since 
$a(x+y)\preceq ax + ay$ by $(T_\sqsubseteq)$ and $ax + ay \preceq ax+a(y+w)$ 
by $(S)$.
\item Analogous to the previous case once we realize that the 
classic axiom for completed trace, $(CT)\ a(bx+u)+ a(cy+v)\simeq
a(bx+cy+u+v)$, is equivalent to the conditional axiom 
$C(x,y) \Rightarrow ax+ay \simeq a(x+y)$.
This follows because
$bx+u$ and $cy+v$ are two independent patterns describing non-null processes
and when the condition is instantiated with $x$ and $y$ equal
to $\cero$ the identity is trivial: $a\cero+a\cero \simeq a\cero$.\qedhere
\end{enumerate}
\end{proof}

\noindent By an argument analogous to that in Proposition~\ref{rnd-eq:prop}, we
can obtain for $\sqsubseteq_T$ the axiomatization 
$\{\textrm{$B_1$--$B_4$}, (S), (\textit{ND}^F_\equiv)\}$.
Note that although $(\textit{ND}^F_\equiv)$ is an equation, this axiomatization is
not the classic one; obviously, $(T)\ ax +ay = a(x+y)$ 
implies $(\textit{ND}^F_\equiv)$ but the converse is false.

It is easy to check that in the case of trace semantics, the particular
instance $(\textit{ND}_0)$ of the axiom $(\textit{ND})$ with $w$ equal to
$\cero$ is powerful enough to generate the trace preorder.
This was certainly not the case when we were under ready simulation, where
$(\textit{ND}_0)$ just generates the failure trace preorder instead of
the coarser failures preorder.

It is also interesting to note that for the trace semantics the symmetric
version of $(\textit{ND})$,
\[
(\textit{ND}_{vw})\quad a(x+y) \preceq a(x+v) + a(y+w),
\]
is also valid, so we can take both 
$\{\textrm{$B_1$--$B_4$}, (S), (\textit{ND}_{vw})\}$ and
$\{\textrm{$B_1$--$B_4$}, (S),$ $(\textit{ND}^\equiv_{vw})\}$, where
\[
(\textit{ND}^\equiv_{vw})\quad a(x+v) + a(y+w)+a(x+y)\simeq a(x+v) + a(y+w),
\]
as alternative axiomatizations of the trace preorder.

Should we expect another diamond of ``reasonable'' semantics under plain
simulation in the spectrum?
Were that to be the case, why have we only found the trace semantics?

In order to answer these questions, note that the diamond of semantics
under ready simulation was completely governed by the function $I$, which
appears in the constraints of the different instantiations of the axiom
$(\textit{ND})$.
For plain simulations, however, the trivially true predicate $U(x,y)$ 
corresponds to the observation function that can see nothing. 
As a consequence, if we substitute $U$ for $I$ in each of the four constraints
of the diamond they all collapse into a single one: trace semantics.
Nevertheless, an alternative path can be explored to obtain new semantics: 
let us keep the different axioms $(\textit{ND\/}^Z)$ the way they stand and
simply replace $(RS)$ by $(S)$. Then we obtain the following results:

\begin{prop}
$\{\textrm{$B_1$--$B_4$}, (S), (\textit{ND\/}^\mathit{FT})\}$ is another 
axiomatization of trace semantics.
Hence, under $(S)$ the failures and the failure trace axioms generate the
same preorder, namely the trace preorder.
\end{prop}
\begin{proof}
$\{\textrm{$B_1$--$B_4$}, (S), (\textit{ND}_0)\}$ is a complete axiomatization
of trace preorder, and $(\textit{ND}_0)$ is a particular case of
$(\textit{ND\/}^\mathit{FT})$.
\end{proof}

The axioms corresponding to readiness and ready trace, however, give rise
to two new semantics that we shall name \emph{extended ready} and 
\emph{extended ready trace} semantics.
They are defined by the order obtained by inclusion of
the offers of the processes, either just at the end of a trace, or
after each action within it: in order to have $p\sqsubseteq_\mathit{ER} q$, for 
each $p\Tran{\alpha}p'$ with $I(p') = R$ we need some
$q\Tran{\alpha} q'$ with $I(q')\supseteq R$; the extended ready trace 
preorder $\sqsubseteq_{ERT}$ is defined analogously, but using ready traces.

\begin{prop}
\hfill
\begin{enumerate}[\em(1)]
\item The set $\{\textrm{$B_1$--$B_4$}, (S), (\textit{ND\/}^R)\}$ is an 
axiomatization of $\sqsubseteq_{ER}$.
\item The set
$\{\textrm{$B_1$--$B_4$}, (S), (\textit{ND\/}^\mathit{RT})\}$ is an 
axiomatization of $\sqsubseteq_{ERT}$.
\end{enumerate}
\end{prop}

\noindent Let us now consider the versions of the axioms
$(\textit{ND\/}^R)$, $(\textit{ND\/}^\mathit{FT})$, $(\textit{ND\/}^\mathit{RT})$
where the constraint $I$ has been replaced by the completeness condition $C$
defined by $C(x) \Longleftrightarrow x = \cero$:
\[
\begin{array}{l@{\quad}l}
(\textit{C-ND\/}^R)& M_\mathit{CR}(x,y,w) \iff (\textrm{$C(x)$ implies $C(y)$})\\
(\textit{C-ND\/}^\mathit{FT}\/)& M_\mathit{CFT}(x,y,w) \iff 
 (\textrm{$C(y)$ implies $C(w)$})\\
(\textit{C-ND\/}^\mathit{RT})& M_\mathit{CRT}(x,y,w) \iff 
 \big(\textrm{($C(x)$ iff $C(y)$) and ($C(y)$ implies $C(w)$)}\big)
\end{array}
\]
Once again, we simply obtain three alternative axiomatizations of the completed
trace semantics.

\begin{prop}
The following axiomatizations are equivalent:
\begin{enumerate}[\em(1)]
\item $\{\textrm{$B_1$--$B_4$}, (CS),$ $(\textit{ND\/}^F)\}$.
\item $\{\textrm{$B_1$--$B_4$}, (CS), (\textit{C-ND\/}^R)\}$.
\item $\{\textrm{$B_1$--$B_4$}, (CS), (\textit{C-ND\/}^\mathit{FT})\}$.
\item $\{\textrm{$B_1$--$B_4$},$ $(CS), (\textit{C-ND\/}^\mathit{RT})\}$.
\end{enumerate}
\end{prop}
\begin{proof}
Clearly, $(1) \Rightarrow (2) \Rightarrow (3) \Rightarrow (4)$ and therefore
it is enough to prove that $(4) \Rightarrow (1)$.
If $x$ and $y$ are not $\cero$ we can apply $(\textit{C-ND\/}^\mathit{RT})$ 
to obtain the inequality in $(\textit{ND\/}^F)$.
If $x$ is $\cero$ but $y$ is not, we need to obtain $ay\preceq a\cero + a(y+w)$.
By $(CS)$ we have $y\preceq y+w$ and then $ay\preceq a(y+w)$; applying $(CS)$
again, $a(y+w)\preceq a(y+w) + a\cero$ and thus $ay\preceq a\cero+a(y+w)$.
If $y$ is $\cero$ but $x$ is not, we need to obtain $ax\preceq ax+ aw$, which
results from an immediate application of $(CS)$.
Finally, if both $x$ and $y$ are $\cero$, $a\cero\preceq a\cero+aw$.
\end{proof}

As before, if we consider the original axioms 
$(\textit{ND\/}^R)$, $(\textit{ND\/}^\mathit{FT})$, and 
$(\textit{ND\/}^\mathit{RT})$ we obtain,
together with an alternative axiomatization of the completed trace semantics,
two new semantics.

\begin{prop}
\label{pro:cs_failures}
The set
$\{\textrm{$B_1$--$B_4$}, (CS), (\textit{ND\/}^\mathit{FT})\}$ is logically 
equivalent to $\{\textrm{$B_1$--$B_4$},$ $(CS), (\textit{ND\/}^F)\}$.
Hence, under $(CS)$, the failures and the failure trace axioms generate the
same semantics.
\end{prop}
\begin{proof}
It is enough to prove that $(\textit{C-ND\/}^\mathit{FT})$
can be derived from $\{\textrm{$B_1$--$B_4$}, (CS),$ $(\textit{ND\/}^F)\}$.
\begin{iteMize}{$\bullet$}
\item If $y$ is $\cero$ we then have $w$ equal to $\cero$ and can apply
$(\textit{ND\/}^\mathit{FT})$.
\item If $y$ is not $\cero$ we can apply $(\textit{ND\/}^\mathit{FT}_0)$ to 
obtain 
$a(x+y)\preceq ax+ay$ and then $(CS)$ to conclude that 
$a(x+y)\preceq ax+a(y+w)$.\qedhere
\end{iteMize}
\end{proof}

\noindent By contrast, as happened for plain simulations, under $(CS)$
the axioms of the ready semantics generate two slightly different versions of
the extended ready and extended ready trace semantics introduced before, that 
we call \emph{extended complete ready} and \emph{extended complete ready trace}
semantics.
In order to have $p\sqsubseteq_\mathit{ECR}
q$, whenever $p\Tran{\alpha} p'$ with $I(p')\neq \emptyset$
we require some $q\Tran{\alpha} q'$ with $I(q')\supseteq I(p')$, but if
$I(p') = \emptyset$ then the corresponding $q'$ also has to 
satisfy $I(q') = \emptyset$.
The extended complete ready trace preorder $\sqsubseteq_\mathit{ECRT}$ is 
defined in an analogous way, starting from the ready traces of the
processes.

As we did in Section~\ref{subsub:equivalences}, we can prove that the axioms that characterize
trace and completed trace preorders reflect the fact that the order relation is
inherited from simulation and complete simulation, respectively, and that the role of
the static rules is to introduce identifications. As stated in Proposition~\ref{pro:cs_failures} above, the only 
inequation that we use to axiomatize the trace and completed trace orders is $(S)$, the remaining axioms being equational axioms.

\begin{prop}
\hfill
  \begin{enumerate}[\em(1)]
  \item $\{\textrm{$B_1$--$B_4$}, (S), (\textit{ND}^F)\}$ is logically equivalent
to $\{\textrm{$B_1$--$B_4$}, (S), (\textit{ND}^F_\equiv)\}$.
\item $\{\textrm{$B_1$--$B_4$}, (CS), (\textit{ND}^F)\}$ is logically equivalent
to $\{\textrm{$B_1$--$B_4$}, (CS), (\textit{ND}^F_\equiv)\}$.
\end{enumerate}
\end{prop}

\noindent A similar discussion could have been carried out for trace and
completed trace equivalences, and  
indeed a very natural axiomatization for these relations can 
be obtained  based on the corresponding instantiation of the $(NS_\equiv)$
equation: 
$$
\begin{array}{l}
(S_\equiv)\quad a(x + y ) \simeq a(x + y)+ay\\
(CS_\equiv)\quad C(x,y) \Rightarrow a(x + y ) \simeq a(x + y)+ay\, .
\end{array}
$$

\begin{prop}
\hfill
\begin{enumerate}[\em(1)]
\item The trace equivalence  $\equiv_T$ is axiomatized by 
 $\{\textrm{$B_1$--$B_4$}, (S_\equiv),$ $(\textit{ND\/}_\equiv^F)\}$.
\item The completed trace equivalence $\equiv_{CT}$ is axiomatized by 
 $\{\textrm{$B_1$--$B_4$}, (CS_\equiv),$ $(\textit{ND\/}_\equiv^F)\}$.
\end{enumerate}
\end{prop}

To conclude this section devoted to the unification of the equational
characterizations of process semantics, we present in
Figure~\ref{fig:new-ltbts} a condensed view of our new spectrum.  This
presentation exploits in an expressive way the two
dimensions of the picture, which in fact reflects a tridimensional structure.
On the lefthand side the constrained simulations and bisimulations appear,
totally ordered from top to bottom. Each constrained simulation generates a
layer of semantics. Here, we have only detailed the layers corresponding to
ready simulation and that of plain simulation. As a matter of fact, the latter
degenerates to a single point due to the simplicity of the constraint $U$
governing plain simulations. The naturality of the semantics appearing in this
part of the spectrum is illustrated by our generic axiomatization, where a
single (constrained) simulation axiom governs all the constrained simulation
semantics, whereas adding a single axiom we complete the axiomatizations of
each of the linear semantics at the righthand side of the picture.

\begin{figure}[th]
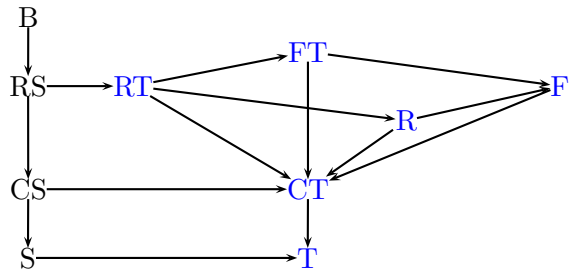

 \begin{center}
   \begin{tabular*}{.5\textwidth}[h]{@{\extracolsep{\fill}}ccccccc}

     \rnode{bs}{B} \\
      & & &  \rnode{ft}{\textcolor{blue}{FT}} &  \\
     \rnode{rs}{RS}&
       \rnode{rt}{\textcolor{blue}{RT}} & & & &&\rnode{f}{\textcolor{blue}{F}}\\
     & & & &\rnode{r}{\textcolor{blue}{R}} \\
     \ncline[linestyle=solid]{->}{bs}{rs}
     \ncline{->}{rs}{rt}
     \ncline{->}{rt}{ft}
     \ncline{->}{rt}{r}
     \ncline{->}{ft}{f}
     \ncline{->}{r}{f}

     \ncline[linestyle=solid]{->}{ts}{rs}
     \ncline[linestyle=solid]{->}{pw2}{pw}
     \ncline[linestyle=solid]{->}{n1}{rt}
     \ncline[linestyle=solid]{->}{pf}{ft}
     \ncline[linestyle=solid]{->}{n3}{r}
     \ncline[linestyle=solid]{->}{n2}{f}
     \\

     \rnode{cs}{CS}& &&
       \rnode{ct}{\textcolor{blue}{CT}} & \\
     \ncline{->}{cs}{ct}

     \ncline[linestyle=solid]{->}{rs}{cs}
     \ncline[linestyle=solid]{->}{f}{ct}
     \ncline[nodesep=1pt,linestyle=solid]{->}{r}{ct}
     \ncline[linestyle=solid]{->}{ft}{ct}
     \ncline[linestyle=solid]{->}{rt}{ct}
     \\
     \rnode{s}{S}&  &&
       \rnode{t}{\textcolor{blue}{T}} & 
     \ncline{->}{s}{t}

     \ncline[linestyle=solid]{->}{cs}{s}

     \ncline[linestyle=solid]{->}{ct}{t}
   \end{tabular*}
 \end{center}
 \caption{New view of the linear time-branching time spectrum.}
 \label{fig:new-ltbts}
\end{figure}

\section{Observational semantics}\label{observational-sem-sec}

Along Section~\ref{sec:EquationalSemantics}
we have presented some views of the axiomatizations for process 
semantics that highlight the common properties and the subtle differences
between them; likewise these views of the axiomatic characterizations 
point out the similarities between the preorder and the equivalence
of a given semantics.

In this section we focus on the characterizations of process semantics 
based on observations. Indeed, this idea of determining the semantics by 
means of observations lies deep inside the foundations of process theory.
 
\begin{quote}
  \emph{Our calculus is founded in two central ideas. The first is 
  observation; [\ldots] two systems are indistinguishable if we cannot tell
them apart without pulling them apart. We therefore give a formal definition
of observation equivalence and investigate its properties.}~\cite{Mil80CCS}
\end{quote}

\begin{quote}
  \emph{Imagine there is an observer with a notebook who watches
  the process and writes down the name of each event as it occurs.}~\cite{Hoa85CSP}
\end{quote}

Besides the classical references to Milner and Hoare, this idea of observation 
pervades the  Hennessy's testing methodology~\cite{Hen88ATP} and most of 
the work on linear semantics. Observations, in spite of the variations in
different proposals, constitute a denotational space closely related to 
the classical developments of semantics based on  denotations for programming
languages~\cite{SemanticsForComputer-Scott-Strachey-1971}.

In this section we will show how most of the semantics can be characterized
with one of the two main families of observations:

\begin{iteMize}{$\bullet$}
\item  Branching general observations, Section~\ref{subsec:bgo}, 
  that are essentially labeled trees, that characterize 
  the simulation semantics: simulation, complete simulation,
  ready simulation, nested simulation, \dots
\item Linear observations, a simplified case of branching observations,
  Section~\ref{linear:sec}, that characterize  the 
  linear semantics: traces, failures, readiness, ready trace, \dots
\end{iteMize}

\noindent We consider also in Section~\ref{deterministic:sec} a more
exotic kind of observations, deterministic branching observations,
which are essentially deterministic trees. Possible worlds semantics
is the only semantics appearing in the classical spectrum in this
class, although, our general approach will show how this kind of
observations define new full families of process semantics.
 
To develop this observational characterization for process semantics allows
us to deepen into the ultimate nature of the similarities and differences 
between them. Along this section we present a thorough study of the local
observation functions that generate the local observations of the states,
Figure~\ref{fig:Slice}. For the linear case, there is also the possibility of
observing this 
local information in a partial way and this is how for each local
observer, in principle, up to four different semantics can be obtained.
This fact explains the classic diamond below the ready simulation
semantics formed by the failures, failure trace, readiness, and ready trace
semantics. Again, the generality of our study makes it exportable to 
other simulation layers enriching and completing the spectrum of semantics,
Figure~\ref{fig:extended-ltbts}.

Finally, from a methodological point of view, the unification of observational
semantics that we present in this section 
introduces all the technical machinery needed to rewrite the 
proofs of Section~\ref{sec:EquationalSemantics} in a generic way,
proving that the two unification procedures produce characterizations of 
the same semantics.
We will address this topic in Section~\ref{rnoef:sec}. Let us now 
concentrate on the observational semantics.
  
\subsection{Branching general observations}
\label{subsec:bgo}

In order to characterize the simulation semantics in an extensional way
we need local and branching general observations.

\begin{defi}\label{local-observations:def}
The sets $L_N$ of \emph{local observations} corresponding to each of the
constrained simulations in the spectrum, and $L_N(p)$ of observations associated
to a process $p$, are defined as follows:
\begin{iteMize}{$\bullet$}
\item Universal (or Plain) simulation: $L_U = \{\cdot\}$; $L_U(p) = \cdot$.
\item Ready simulation: $L_I = \calP(\textit{Act})$; $L_I(p) = I(p)$.
\item Complete simulation: $L_C = \textit{Bool}$; $L_C(p)$ is
 \textit{true} if $I(p)=\emptyset$ and \textit{false} otherwise.
\item Trace simulation\footnote{Trace simulations are the only ones in this
    list that do not appear in~\cite{Gla01}. They can be defined as 
 $T$-simulations, with $T(x,y) ::= T(x) = T(y)$, and the general theory about
 constrained simulations in~\cite{FG08ifiptcs} applies to them. In particular,
 they can be axiomatized as stated in Proposition~\ref{prop:ns}(3),
 page~\pageref{prop:ns}.}: $L_T = \calP(\textit{Act}^*)$;
 $L_T(p)=T(p)$, the set of traces of $p$.
\item 2-nested simulation: $L_S = \{\lsem p \rsem_S\mid 
 p\in\textit{BCCSP}\}$; $L_{S}(p) =\lsem p\rsem_S$, where $\lsem p\rsem_S$ 
 represents the equivalence class of $p$ with respect to the simulation equivalence.
\end{iteMize}
\end{defi}

\begin{defi}\label{bgo:def}
\hfill
\begin{enumerate}[(1)]
\item
A \emph{branching general observation} (bgo for short)
of a process is a finite, non-empty tree 
whose arcs are labeled with actions in \textit{Act} and whose nodes are labeled
with local observations from $L_N$, for $N$ a constraint; the corresponding set
$\textit{BGO}_N$ is recursively defined as:
 \begin{iteMize}{$\bullet$}
 \item
 $\langle l,\emptyset\rangle \in \textit{BGO}_N$ for $l\in L_N$.
 \item
 $\langle l, \{(a_i, \textit{bgo}_i)\mid i\in 1..n\}\rangle \in 
 \textit{BGO}_N$ for every $n\in\nat$, 
 $a_i\in\textit{Act}$ and $\textit{bgo}_i\in
 \textit{BGO}_N$.
 \end{iteMize}

\item
The set $\textit{BGO}_N(p)$ of branching general observations of $p$ 
corresponding to the constraint $N$ is
\[
\textit{BGO}_N(p) = 
 \{\langle L_N(p), S\rangle \mid S\subseteq\{(a,\textit{bgo})\mid
  \textit{bgo}\in \textit{BGO}_N(p'), p\tran{a}p'\}\}\,.
\]

\item We write $p\leq^b_N q$ if $\textit{BGO}_N(p)\subseteq\textit{BGO}_N(q)$.
\end{enumerate}\smallskip
\end{defi}

\noindent In Figure~\ref{bgos2:fig} some simple examples of bgo's for
$N=I$ are shown.  We represent $\textit{bgo}_1$ as
\[
\langle \{a\}, \{(a,\langle \{b\}, \{(b, \langle \{c\}, \emptyset\rangle)\}
\rangle), (a, \langle \{b\}, \{(b, \langle \{d\}, \emptyset\rangle)\}\rangle)\}
\rangle
\]
and
$\textit{bgo}_2$ as  
\[
\langle \{a\}, \{(a,\langle \{b\}, \{(b, \langle \{c\}, \emptyset\rangle), 
(b, \langle\{d\},\emptyset\rangle)\}\rangle)\}\rangle.
\]
We use braces for the set of children of a node, parentheses to represent a 
branch of the tree as
a pair (initial arc, subtree below), and angular brackets to represent each tree
as a pair $\langle\textrm{root},\textrm{children}\rangle$.

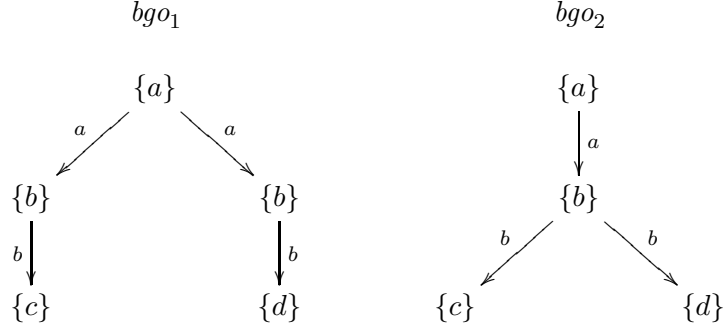
\begin{figure}[t]
\[
\begin{array}{c@{\qquad\qquad}c}
\textit{bgo}_1&\textit{bgo}_2\\ \\
\xymatrix{
&\{a\}\ar[dl]_a\ar[dr]^a\\
\{b\}\ar[d]_b&& \{b\}\ar[d]^b\\
\{c\}&& \{d\}}
&
\xymatrix{
&\{a\}\ar[d]^a\\
&\{b\}\ar[dl]_b\ar[dr]^b\\
\{c\}&& \{d\}
}
\end{array}
\]
\caption{Two branching observations.}
\label{bgos2:fig}
\end{figure}

Note that the bgo's of a process $p$ described by its transition system 
can be generated by inductively applying the clauses defining the set 
$\textit{BGO}_N(p)$, even when $p$ is infinite.
For instance, if $N=I$ and we consider the process $p :: = c . p$ defining a 
clock, since $\emptyset \subseteq\{(c,\textit{bgo})\mid
  \textit{bgo}\in \textit{BGO}_I(p), p\tran{c}p\}$,
it follows that $\langle \{c\}, \emptyset\rangle\in\textit{BGO}_I(p)$.
But now $\{(c,\langle \{c\}, \emptyset\rangle)\} \subseteq\{(c,\textit{bgo})\mid
  \textit{bgo}\in \textit{BGO}_I(p), p\tran{c}p\}$
and therefore $\langle\{c\},\{(c,\langle \{c\}, \emptyset\rangle)\}\rangle\in
\textit{BGO}_I(p)$, and so on.

It is clear that the bgo's of a process have an operational flavor.
The nodes of the observations correspond to its states and the arcs to
its transitions; this is why we will be able to define the orders associated 
to the different simulation semantics simply by set inclusion over the sets of
bgo's.

Let us also comment on the fact that in all five cases that we have considered
in Definition~\ref{local-observations:def}, 
which correspond to the five constrained simulation semantics
in the spectrum, the local observation functions $L_N$ 
define a representation of the equivalence relation $N$
used to define the constrained simulation relations. 
This means that we have $L_N(p) = L_N(q) \Longleftrightarrow p N q$.

\begin{exa}
For $N = I$, if $x = b(c+d)$ and $y= bc + bd$, then for $p = a(x+y)$
we have $\textit{bgo}_k \in \textit{BGO}_I(p)$ for $k\in \{1,2,3\}$, where the
bgo's are depicted in Figure~\ref{bgos:fig}.
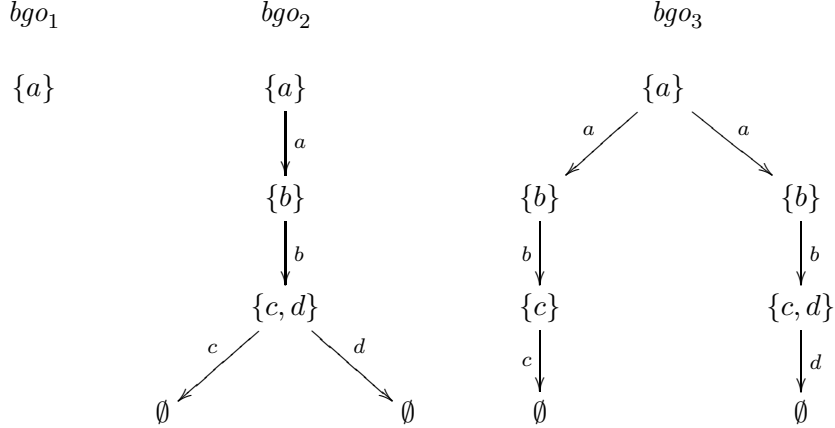
\begin{figure}[t]
\[
\begin{array}{c@{\qquad\quad}c@{\qquad\quad}c}
\textit{bgo}_1&\textit{bgo}_2&\textit{bgo}_3\\ \\
\{a\}
&
\xymatrix{
&\{a\}\ar[d]^a\\
&\{b\}\ar[d]^b\\
&\{c,d\}\ar[dr]^d\ar[dl]_c\\
\emptyset&&\emptyset}
&
\xymatrix{
&\{a\}\ar[dl]_a\ar[dr]^a\\
\{b\}\ar[d]_b&&\{b\}\ar[d]^b\\
\{c\}\ar[d]_c&&\{c,d\}\ar[d]^d\\
\emptyset&&\emptyset
}
\end{array}
\]
\caption{Three branching observations.}
\label{bgos:fig}
\end{figure}
It is easy to check that all of them are also branching observations of
$q = a(x+y)+ ax$.
As a matter of fact, we have $\textit{BGO}_I(p) = \textit{BGO}_I(q)$.
Note that in order to obtain $\textit{bgo}_3\in \textit{BGO}_I(p)$ we need to
combine two different observations of the process $x+y$, which is the only
$p'$ such that $a(x+y)\tran{a} p'$.

In contrast, for $p= a(bc+bd)$ and $q=abc + abd$, 
$\textit{BGO}_I(q) \not\subseteq \textit{BGO}_I(p)$, since for the branching
observation $\textit{bgo}_1$ in Figure~\ref{bgos2:fig} we have 
$\textit{bgo}_1\in\textit{BGO}_I(q)$ and
$\textit{bgo}_1\not\in \textit{BGO}_I(p)$.
And also, we have $\textit{BGO}_I(p)\not\subseteq\textit{BGO}_I(q)$, since
for $\textit{bgo}_2$ as in Figure~\ref{bgos2:fig} we have
$\textit{bgo}_2\in \textit{BGO}_I(p)$, but
$\textit{bgo}_2\notin \textit{BGO}_I(q)$.
The key idea is that we can indeed include in a 
single bgo two separated computations but we cannot ``mix'' two different ones, 
even if the labels both in their initial transitions and in the local 
observations of the reached nodes were the same.
This is why $\textit{bgo}_2\notin\textit{BGO}_I(q)$.
\end{exa}

The following simple properties will be immediate consequences of
Theorem~\ref{denotational-main:thm} below; we use them here to illustrate the 
expressive power of each kind of bgo.

\begin{defi}
An axiom $t\preceq u$, respectively $t\simeq u$, is satisfied in a model
$\textit{BGO}_N$ if $\textit{BGO}_N(t')\subseteq\textit{BGO}_N(u')$, 
respectively $\textit{BGO}_N(t')=\textit{BGO}_N(u')$, for every possible 
ground instantiation $t'\preceq u'$ or $t'\simeq u'$ of the axiom.
\end{defi}

\begin{prop}
\hfill
\begin{enumerate}[\em(1)]
\item \label{auxiliar-1:prop}
 The axiom $(S)\ x\preceq x + y$ is satisfied in the model $\textit{BGO}_U$.
\item The axiom $(S_\equiv)\ a(x+y)\simeq a(x+y) + ax$ 
 is satisfied in the model $\textit{BGO}_U$.
\end{enumerate}
\end{prop}
\begin{proof}
\hfill
\begin{enumerate}[(1)]
\item It is an immediate consequence of the fact that if $p\tran{a}p'$ then 
 $p+q\tran{a}p'$, and therefore 
 $\{a\mid p\tran{a}\}\subseteq\{a\mid p+q\tran{a}\}$.
\item Again, it is a simple exercise to check that $\textit{BGO}_U(p)\subseteq
 \textit{BGO}_U(q)$ implies $\textit{BGO}_U(ap)\subseteq \textit{BGO}_U(aq)$,
 and that if $\textit{BGO}_U(p), \textit{BGO}_U(q)\subseteq \textit{BGO}_U(r)$,
 then $\textit{BGO}_U(p + q) \subseteq \textit{BGO}_U(r)$;
 in combination with (\ref{auxiliar-1:prop}), this produces the result.\qedhere
\end{enumerate}
\end{proof}

\begin{prop}
$\textit{BGO}_I(p)\subseteq \textit{BGO}_I(p+q)$ iff $I(q) \subseteq I(p)$.
\end{prop}
\begin{proof}
$(\Leftarrow)$ Since $I(p+q) = I(p)$, the root of the bgo's is the same for
both processes and obviously $p+q$ has all the observations of $p$.

$(\Rightarrow)$ If $I(q) \not\subseteq I(p)$, then $I(p)\neq I(p+q)$ and
then no bgo of $p$ is a bgo of $p+q$ because the roots of the
observations of both processes are different.
\end{proof}

The fact, that we now prove, that the observational semantics 
$\textit{BGO}_N(p)$ can be defined in a compositional way, is an important
property that will simplify the proofs of many of their properties.

\begin{thm}\label{denotational-aux:thm}
Let $L$ be a function used as a local observation function and let us also
denote by $L$ the range of $L$, as done in
Definition~\ref{local-observations:def}. 
If there exist semantic functions $+^L: L\times L\to L$ and $a^L: L \to L$
satisfying
$L(a p) = a^L L(p)$ and $L(p+q) = L(p) +^L L(q)$, then:
\begin{iteMize}{$\bullet$}
\item $\textit{BGO}_N(ap) = \{\langle a^LL(p), \{(a,\textit{bgo})\mid
 \textit{bgo}\in B\}\rangle \mid B\subseteq \textit{BGO}_N(p)\}$.
\item $\textit{BGO}_N(p+q) = \{\langle L(p) +^L L(q), S_1\cup S_2\rangle\mid
 \langle L(p), S_1\rangle\in \textit{BGO}_N(p), 
 \langle L(p), S_2\rangle\in \textit{BGO}_N(q)\}$.
\end{iteMize}
\end{thm}
\begin{proof}
The first equality  is immediate by definition of $\textit{BGO}_N(ap)$.
As for the second, we only need to realize that $p+q\tran{a}r$ iff
$p\tran{a} r$ or $q\tran{a}r$: then, the set of children of the root
labeled by $L_N(p+q)$ at any bgo $\in\textit{BGO}_N(ap)$ correspond to
the union of the two sets of children that contain some bgo's of processes
$p_i$ such as $p\tran{a}p_i$ (and then  $p+q\tran{a}p_i$) or $q_i$
such that  $q\tran{a}q_i$ (and then  $p+q\tran{a}q_i$).
Note that from the equalities above it follows that $\textit{BGO}_N(p)$ can
be computed compositionally.
\end{proof}

In particular, $\textit{BGO}_N(p)$ is compositional for any of the constraints
considered in Definition~\ref{local-observations:def}.

\begin{prop}
For $N\in\{U, I, C, T, S\}$, $L_N$ can be defined in a compositional
way over the terms in BCCSP.
\end{prop}
\begin{proof}
The result for $U$ is obvious since it is a degenerate semantics that 
identifies all processes.
By Theorem~\ref{denotational-aux:thm} and 
Theorem~\ref{denotational-main:thm} below we can conclude that the simulation 
semantics can indeed be denotationally defined.
The result for traces is well-known, while $I$ and $C$ can be easily
defined denotationally since $I(ap) = \{a\}$ and $I(p+q) = I(p)\cup I(q)$.
\end{proof}

Now we show that bgo's characterize $N$-simulation semantics in all cases.

\begin{thm}\label{denotational-main:thm}
For all $N\in\{U, I, C, T, S\}$ and any two processes $p$ and $q$,
$p\sqsubseteq_\mathit{NS} q$ iff $p\leq^b_N q$.
\end{thm}
\begin{proof}
$(\Rightarrow)$
Let $p=\sum\sum a p_a^i$ and $q= \sum\sum a q_a^j$; if 
$p\sqsubseteq_\mathit{NS} q$, then $pNq$ and therefore $L_N(p) = L_N(q)$.
Now we proceed by induction on $p$.
If $p\equiv \cero$ the result is trivial.
Otherwise, for every $a\in I(p)$ such that $p\tran{a}p'$ there exists 
$q\tran{a}q'$ such that $p'\sqsubseteq_\mathit{NS} q'$.
By induction hypothesis $\textit{BGO}_N(p')\subseteq \textit{BGO}_N(q')$
from where, by the definition of $\textit{BGO}_N(p)$, it follows that
$\textit{BGO}_N(p)\subseteq\textit{BGO}_N(q)$.

$(\Leftarrow)$
Let us show that the relation 
$R= \{(p,q)\mid \textit{BGO}_N(p)\subseteq \textit{BGO}_N(q)\}$ is an 
$N$-simulation.
If $(p,q)\in R$, then 
$L_N(p)=L_N(q)$ because $\langle L_N(p),\emptyset\rangle\in \textit{BGO}_N(q)$
and thus $pNq$.
Now, for each $p\tran{a} p'$ we have
$\{\langle L_N(p), \{(a,\textit{bgo})\}\rangle\mid 
 \textit{bgo}\in\textit{BGO}_N(p')\}\subseteq\textit{BGO}(q)$
and therefore there must exist some $q\tran{a} q'$ such that
$\textit{BGO}_N(p')\subseteq \textit{BGO}_N(q')$, so that $(p',q')\in R$.
\end{proof}

Note that for this result to hold it is only required that the local observation 
function $L_N$ satisfies $pNq$ iff $L_N(p) = L_N(q)$.
That is, $L_N$ must compute a concrete representative of the equivalence 
class defined by $N$ and this stresses again the interest of using behavior 
equivalences $N$ as constraints for the definition of constrained simulations.
Let us recall that, in principle, any behavior preorder could be used as
such a constraint.
For instance, the predicate $I_\supseteq$ defined by $pI_\supseteq q$ iff
$I(q)\subseteq I(p)$ could be used to define $I_\supseteq$-simulations (which in fact 
coincide with $I$-simulations).
But from $I(q)\subseteq I(p)$ we cannot conclude that $L_N(p)=L_N(q)$ and,
hence, either a more complicated characterization of $\sqsubseteq_\mathit{NS}$ 
in terms of bgo's or an additional argument to show that
$p\sqsubseteq_{I_\subseteq} q$ implies $I(p)\subseteq I(q)$ would be needed.
And although this is obvious for a constraint as simple as $I$, or even
$T$ or $S$, it could be far from trivial for other, more complex constraints: 
therefore, it is always advisable to consider equivalence behaviors as 
constraints.

\begin{cor}
For any constraint $N$ that is a behavior equivalence, whenever we have as
local observation function $L_N$ the quotient function 
$L_N(p) = \lsem p\rsem_N$ or any concrete representation of it satisfying
$L_N(p) = L_N(q)\iff pNq$, then 
$p\sqsubseteq_\mathit{NS} q$ iff $\textit{BGO}_N(p)\subseteq \textit{BGO}_N(q)$.
\end{cor}

The results above bring forward the fact that despite the resemblance between
the bgo's of a process and its computation tree, the possibility of
mixing several computations in a single branching observation makes it
possible to identify non-bisimilar processes by their sets of branching
observations.

\subsection{Linear observations and linear time semantics}\label{linear:sec}

We introduce the linear observations of a process as a particular (degenerate)
case of branching observations: those with a linear structure.

\begin{defi}\label{lbo:def}
\hfill
\begin{enumerate}[(1)]
\item
The set $\textit{LGO}_N$ of \emph{linear general observations} 
(lgo for short)
for a local observer $L_N$ is the subset of $\textit{BGO}_N$ defined as:
\begin{iteMize}{$\bullet$}
\item $\langle l,\emptyset\rangle\in \textit{LGO}_N$ for each $l\in L_N$.
\item $\langle l, \{(a, \textit{lgo})\}\rangle$, whenever 
 $a\in A$ and $\textit{lgo}\in\textit{LGO}_N$.
\end{iteMize}

\item
The set of \emph{linear general observations} 
of a process $p$ with respect to the
local observer $L_N$ is $\textit{LGO}_N(p) = \textit{BGO}_N(p)\cap 
\textit{LGO}_N$.
\end{enumerate}
\end{defi}

\noindent Since lgo's are linear they can be presented as traces, avoiding the sets
of descendants in the bgo's.
Therefore, we will consider them as elements of the set
$L_N\times (\textit{Act}\times L_N)^*$.

It is also clear that the set of linear observations can be defined 
recursively without resorting to branching observations.

\begin{defi}
The set $\textit{LGO}_N(p)$ of linear general observations of a 
process $p$ is recursively defined by
\[
\textit{LGO}_N(p) = \{\langle L_N(p)\rangle\}\cup 
               \{\langle L_N(p),a\rangle\circ\textit{lgo}\mid p\tran{a} p',\,
                 \textit{lgo}\in\textit{LGO}_N(p')\}.
\]
\end{defi}

We can also compute $\textit{LGO}_N(p)$ in a compositional way.

\begin{prop}
Let $L$ be a local observation function such that there exist
semantic functions $+^L: L_N\times L_N\to L_N$ and $a^L: L_N \to L_N$
satisfying
$L(a p) = a^L L(p)$ and $L(p+q) = L(p) +^L L(q)$.
Then:
\begin{iteMize}{$\bullet$}
\item $\textit{LGO}_N(ap) = \{\langle a^L L(p) \rangle\} \cup 
 \{\langle a^LL(p), a\rangle \circ \textit{LGO}_N(p)\}$.
\item $\textit{LGO}_N(p+q) = \{\langle L(p) +^L L(q)\rangle\circ t\mid
       \langle L(p)\rangle\circ t\in \textit{LGO}_N(p) \textrm{ or } \\
       \hphantom{\textit{LGO}_N(p+q) = \{\langle L(p)+^LL(q)\rangle\circ t\mid{}}
 \langle L(p)\rangle\circ t\in \textit{LGO}_N(q)\}$.
\end{iteMize}
\end{prop}
\begin{proof}
Just like that of Theorem~\ref{denotational-aux:thm}.
\end{proof}

Obviously, for $N= U$ we have that $\textit{LGO}_U$ is isomorphic to $\textit{Act}^*$ 
and thus $\textit{LGO}_U(p) = \textit{Traces}(p)$.
By contrast, for $N = I$, $\textit{LGO}_I(p)$ is the set of ready traces of
$p$, $\textit{ReadyTraces}(p)$.

Set inclusion of linear observations with respect to a local
observer $L_N$ gives us the preorder defining the corresponding semantics.

\begin{defi} \label{orden_lN}
A process $p$ is less than or equal to $q$ with respect to the linear
observations generated by $L_N$, denoted $p\leq^l_N q$, if
$\textit{LGO}_N(p)\subseteq \textit{LGO}_N(q)$.
We will denote the corresponding equivalence by $=^l_N$.
\end{defi}

\begin{prop}\label{traces-rtraces:prop}
(1) ${\leq^l_U} = {\sqsubseteq_T}$; (2) ${\leq^l_I} = {\sqsubseteq_{RT}}$;
(3) ${\leq^l_C} = {\sqsubseteq_{CT}}$.
\end{prop}
\begin{proof}
It is trivial, since $\textit{LGO}_U(p) = \textit{Traces}(p)$,
$\textit{LGO}_I(p) = \textit{ReadyTraces}(p)$, and
$\textit{LGO}_C(p) =
\{
(\textit{false},a_1)\circ{}\dots{}\circ (\textit{false},a_n,\textit{true}),
(\textit{false},a_1)\circ{}\dots{}\circ (\textit{false},a_i,\textit{false})\mid$
$a_1\dots a_n\in \textit{CompleteTraces}(p), i< n
\}$.
\end{proof}

\begin{prop}
For $N\in \{U, C, I, T, S\}$, if $p\sqsubseteq_\mathit{NS} q$ then 
$p\leq_N^l q$, but the converse is false in general.
\end{prop}
\begin{proof}
The implication follows from Theorem~\ref{denotational-main:thm}
and the fact that lgo's are just a particular case of bgo's.
To see that the converse is false in general consider $N=U$; we have ${\sqsubseteq_{US}} = 
{\sqsubseteq_S}$ and ${\leq_U^l} = {\sqsubseteq_T}$, and it is well-known that 
${\sqsubseteq_S} \not\subseteq{\sqsubseteq_T}$ since,
for instance, $a(b+c) \not\sqsubseteq_S ab + ac$, but $a(b+c) =_T ab+ac$.
\end{proof}

Therefore, by means of linear observations and set inclusion we can
characterize the orders that define some of the semantics in the spectrum which
are not simulation semantics.
However, there are still some other semantics for which a different way of
treating the linear observations is needed. We need to introduce some identifications
in the corresponding domain $\textit{LGO}_N$ to obtain their characterizations.

\begin{defi}
For $\calT, \calT'\subseteq\textit{LGO}_I$ we define the orders
$\leq^{l\supseteq}_I$, $\leq_I^{lf}$, and $\leq_I^{lf\supseteq}$ by:
\begin{iteMize}{$\bullet$}
\item $\calT\leq_I^{l\supseteq}\calT' \iff 
       \begin{array}[t]{l}
       \textrm{for all $X_0a_1X_1\dots X_n\in\calT$}\\
       \textrm{there is some $Y_0a_1Y_1\dots Y_n\in\calT'$
               with $X_i\supseteq Y_i$, for all $i\in 0..n$.}
       \end{array}$
\item $\calT\leq_I^{lf} \calT'\iff 
       \begin{array}[t]{l}
       \textrm{for all $X_0a_1X_1\dots X_n\in\calT$}\\
       \textrm{there is some $Y_0a_1Y_1\dots Y_n\in\calT'$
               with $X_n= Y_n$.}
       \end{array}$
\item $\calT\leq_I^{lf\supseteq}\calT' \iff
       \begin{array}[t]{l}
       \textrm{for all $X_0a_1X_1\dots X_n\in\calT$}\\
       \textrm{there is some $Y_0a_1Y_1\dots Y_n\in\calT'$ with $X_n\supseteq
         Y_n$.}
       \end{array}$
\end{iteMize}
Then, for each $\delta\in\{\supseteq, f, f{\supseteq}\}$ we write $p\leq_I^{l\delta} q$ if $\textit{LGO}_I(p)\leq_I^{l\delta}\textit{LGO}_I(q)$.
\end{defi}

Since the definition of $\leq_I^{lf}$ ignores all the intermediate ready sets
$X_i$ with $i<n$ and requires the final ready sets to coincide, it is obvious
that it defines the readiness preorder.
Let us now prove that the two semantics based on failures are also characterized
by our preorders $\leq_I^{lf\supseteq}$ and $\leq_I^{l\supseteq}$.

\begin{prop}\label{failures-ftraces:prop}
The preorder $\leq_I^{lf\supseteq}$ generates the failures preorder and 
$\leq_I^{l\supseteq}$ generates the failure trace preorder.
\end{prop}
\begin{proof}
The proof is based on the definition of initial failures of a process:
we say that $p$ rejects $X$ if and only if $X\cap I(p)=\emptyset$.
Then, $\langle \alpha, X\rangle$ is a failure of $p$ if  and only
if $p\Tran{\alpha} p'$ and $p'$ rejects $X$.
Using lgo's, for $\alpha = a_1\dots a_n$, $\langle \alpha, X\rangle$ is a 
failure of $p$ iff there exists $X_0a_1\dots X_n\in \calT$ such that
$X_n\cap X=\emptyset$.
Thus: 
\[
\begin{array}{rcl}
p\sqsubseteq_F p'
&\iff&\textit{Failures}(p)\subseteq\textit{Failures}(p') \\
&\iff&\langle \alpha, X\rangle \in \textit{Failures}(p')
    \ \textrm{for all}\ \langle \alpha, X\rangle \in \textit{Failures}(p)\\
&\iff&X_0a_1\dots X_n\in \textit{LGO}_I(p) \ \textrm{with}\  
    X_n\cap X =\emptyset \ \textrm{implies that there exists}\\
&&\hphantom{\textrm{iff}}\ Y_0a_1\dots Y_n\in \textit{LGO}_I(p') 
    \ \textrm{with}\ Y_n\cap X=\emptyset\,,
\end{array}
\]
and then $p\leq_I^{lf\supseteq} p'$ implies $p\sqsubseteq_F p'$.

Conversely, assume that $p\sqsubseteq_F p'$ and recall that
$p\leq_I^{lf\supseteq} p'$ iff for all $t= X_0a_1\dots X_n\in \textit{LGO}_I(p)$
there exists $Y_0a_1\dots Y_n\in \textit{LGO}_I(p')$ such that 
$X_n\supseteq Y_n$. For each set $X$, let us denote by $X^c$ its complement.
If $t\in\textit{LGO}_I(p)$, we have
$\langle \alpha, X_n^c\rangle\in\textit{Failures}(p)$ and therefore
$\langle\alpha, X_n^c\rangle\in \textit{Failures}(p')$, which implies that
there exists $p'\Tran{\alpha}p''$ such that $I(p'')\cap X_n^c=\emptyset$.
This means that there is some
$t'=Y_0a_1\dots a_nI(p'')\in\textit{LGO}_I(p')$ with $I(p'')\subseteq X_n$, 
and therefore we can conclude that $p\leq_I^{lf\supseteq} p'$.

The proof for failure trace is very similar and we omit it.
\end{proof}

As a matter of fact, the characterization of failures by means of the reverse
inclusion of offerings is not a great discovery at all: for instance,
the same idea
can be found in the definition of acceptance trees \cite{Hennessy85}.
However, our sets of linear observations produce quite a nice 
characterization and allow us to forget about the notion of
failures and consider instead reverse inclusion of offerings.
But the most important property of our characterizations in terms of
different orders on the set $\textit{LGO}_I$ is that they can be generalized
to other local observation functions.

\begin{defi}\label{obs:def}
For $\calT, \calT'\subseteq\textit{LGO}_N$ we define the 
orders $\leq_N^{l\supseteq}$, $\leq_N^{lf}$, and $\leq_N^{lf\supseteq}$ by:
\begin{iteMize}{$\bullet$}
\item $\calT\leq_N^{l\supseteq}\calT'\iff
       \begin{array}[t]{l}
       \textrm{for all $X_0a_1X_1\dots X_n\in\calT$}\\
       \textrm{there is some $Y_0a_1Y_1\dots Y_n\in\calT'$
               with $X_i\supseteq Y_i$ for all $i\in 0..n$.}
       \end{array}$
\item $\calT\leq_N^{lf} \calT'\iff
       \begin{array}[t]{l}
       \textrm{for all $X_0a_1X_1\dots X_n\in\calT$}\\
       \textrm{there is some $Y_0a_1Y_1\dots Y_n\in\calT'$
               with $X_n= Y_n$.}
       \end{array}$
\item $\calT\leq_N^{lf\supseteq}\calT' \iff
       \begin{array}[t]{l}
       \textrm{for all $X_0a_1X_1\dots X_n\in\calT$}\\
       \textrm{there is some $Y_0a_1Y_1\dots Y_n\in\calT'$ 
               with $X_n\supseteq Y_n$.}
       \end{array}$
\end{iteMize}
Then, for each $\delta\in\{\supseteq, f, f{\supseteq}\}$ we write $p\leq_N^{l\delta} q$ if 
$\textit{LGO}_N(p)\leq_N^{l\delta} \textit{LGO}_N(q)$.
\end{defi}

By abuse of notation, we have used the superset inclusion symbol $\supseteq$
in the definitions above for any $N$.
That is indeed the right interpretation for the cases $N=I, T$; however, for 
$N =U,C$ the superset inclusions degenerate to equalities
while for $N=S$ it should be interpreted as 
$\lsem p\rsem_S \geq_S \lsem q\rsem_S$.
Then, with the right notation we could have used such an inequality 
$\lsem p\rsem_N \geq_N \lsem q\rsem_N$ in all the cases. 

When defining an observational semantics one expects the order between processes
to be plain set inclusion as is the case, for instance, for the classic
definition of failures semantics.
Fortunately, it is easy to obtain such a characterization for the three
semantics considered above by means of some suitable closure operators.

\begin{defi}\label{closures:def}
For $\calT\subseteq\textit{LGO}_N$, the following three closures are defined:
\begin{iteMize}{$\bullet$}
\item $\ol{\calT}^\supseteq = \{X_0a_1X_1\dots a_nX_n\mid 
 \textrm{there is some $Y_0a_1Y_1\dots a_nY_n\in \calT$
         with $X_i\supseteq Y_i$ for $i\in 0..n$}\}$.
\item $\ol{\calT}^f = \{X_0a_1X_1\dots a_nX_n\mid 
 \textrm{there is some $Y_0a_1Y_1\dots a_nX_n\in \calT$}\}$.
\item $\ol{\calT}^{f\supseteq} = \{X_0a_1X_1\dots a_nX_n\mid 
 \textrm{there is some $Y_0a_1Y_1\dots a_nY_n\in \calT$
         with $X_n\supseteq Y_n$}\}$.
\end{iteMize}
\end{defi}

\begin{prop}\label{closures:prop}
All the operators in Definition~\ref{closures:def} are indeed closures:
if $\delta\in\{\supseteq,f,f{\supseteq}\}$ and 
$\calT,\calT'\subseteq \textit{LGO}_N$, then
$\calT\subseteq \ol{\calT}^\delta$ and $\ol{\ol{\calT}^\delta}^\delta = \ol{\calT}^\delta$; also,
if $\calT\subseteq \calT'$ then $\ol{\calT}^\delta\subseteq \ol{\calT'}^\delta$.
\end{prop}
\begin{proof}
The first and third conditions are immediate from the definitions.
As for the second, let $X_0a_1X_1\dots a_nX_n\in \ol{\ol{\calT}^f}^f$.
Then, there exists $Y_0a_1Y_1\dots a_nX_n\in \ol{\calT}^f$ and thus
there exists $Z_0a_1Z_1\dots a_nX_n\in \calT$, which implies
 $X_0a_1X_1\dots a_nX_n\in \ol{\calT}^f$;
the inclusion in the other direction follows from monotonicity.
Analogously for the other two operators.
\end{proof}

\begin{prop}\label{prop_contenido}
For all $\delta\in\{ \supseteq, f, f{\supseteq}\}$,
$\calT\leq_N^{l\delta} \calT'$ iff $\;{\ol{\calT}^\delta}\!\!\subseteq {\ol{\calT'}^\delta}$.
\end{prop}
\begin{proof}
It is easy but tedious, so only the case $\delta = f{\supseteq}$ is 
presented in detail. 
Assume $\calT\leq_N^{lf\supseteq} \calT'$: for all 
$t = X_0a_1X_1\dots a_nX_n\in\calT$ there exists 
$Y_0a_1Y_1\dots a_nY_n\in\calT'$ with $X_n\supseteq Y_n$ and hence
$t\in \ol{\calT'}^{f\supseteq}$ and $\calT\subseteq \ol{\calT'}^{f\supseteq}$; 
$\ol{\calT}^{f\supseteq}\subseteq \ol{\calT'}^{f\supseteq}$ follows because of 
the properties of closures.

Conversely, from $\ol{\calT}^{f\supseteq}\subseteq \ol{\calT'}^{f\supseteq}$ 
it follows that
$\calT\subseteq \ol{\calT'}^{f\supseteq}$ and thus for all
$X_0a_1X_1\dots a_nX_n\in\calT$ there exists 
$Y_0a_1Y_1\dots a_nY_n\in\calT'$ with $X_n\supseteq Y_n$: therefore
$\calT\leq_N^{lf\supseteq}\calT'$.
\end{proof}

\begin{defi}
For each $\delta\in\{\supseteq, f, f{\supseteq}\}$, $p\in\textit{BCCSP}$, and $N$ 
a constraint, we define $$\textit{LGO}_N^\delta(p) = \ol{\textit{LGO}_N(p)}^\delta.$$
\end{defi}

Let us see which of the semantics in the spectrum are characterized by
the orders $\leq_N^{l\delta}$ defined above.

\begin{prop}
For $N=U$ we have ${\leq_U^l} = {\leq_U^{l\supseteq}} = {\leq_U^{lf}} =
{\leq_U^{lf\supseteq}} = {\sqsubseteq_T}$.
As a consequence, the only semantics coarser than plain simulation that can be
characterized by means of linear observations using $L_U$ is the 
trace semantics.
\end{prop}
\begin{proof}
The first three equalities are obvious since $U$ provides useless (empty)
local information $(L_U = \{\cdot\})$.
The last equality was proved in Proposition~\ref{traces-rtraces:prop}(1).
\end{proof}

\begin{prop}
For $N=C$ we have ${\leq_C^l} = {\leq_C^{l\supseteq}} = {\leq_C^{lf}} =
{\leq_C^{lf\supseteq}} = {\sqsubseteq_{CT}}$.
As a consequence, the only semantics coarser than complete simulation that can 
be characterized by means of linear observations using $L_C$ is the 
completed trace semantics.
\end{prop}
\begin{proof}
Note that the local information at the intermediate steps of traces in 
$\textit{LGO}_C$ has to be \textit{false}, since it corresponds to 
non-terminated states; thus, only the final states provide real information.
Since in this case $\supseteq$ corresponds to Boolean equality, the first 
three equalities follow; the fourth was proved in 
Proposition~\ref{traces-rtraces:prop}(3).
\end{proof}

\begin{prop}
For $N=I$, $\leq_I^{lf\supseteq}$ characterizes the failures semantics,
$\leq_I^{lf}$ the readiness semantics, $\leq_I^{l\supseteq}$ the failure
trace semantics, and $\leq_I^l$ the ready trace semantics.
Therefore, the possible worlds semantics is the only semantics in the
ltbt spectrum coarser than ready simulation that cannot be characterized using
$\textrm{lgo}_I$'s.
\end{prop}
\begin{proof}
We have already proved (Propositions~\ref{traces-rtraces:prop} and 
\ref{failures-ftraces:prop}) the four characterizations, while $\sqsubseteq_{PW}$
cannot be characterized using $\textrm{lgo}_I$'s because all the 
information available in our $\textrm{lgo}_I$'s was needed to capture
the ready trace semantic and it is well-known that the possible worlds
semantics is strictly finer.
\end{proof}

As we will see in Section~\ref{deterministic:sec}, the possible worlds semantics 
is the only deterministic branching semantics in the spectrum and will require 
the use of the deterministic branching observations introduced there to be
characterized in an observational way.
This is not the case, however, for the possible futures semantics (already 
discussed in \cite{Gla01}), and the impossible futures semantics
\cite{VM01}.

\begin{defi}
\hfill
\begin{enumerate}[(1)]
\item The impossible futures semantics is defined as:
 $p\sqsubseteq_{IF} q$ if for all $S\subseteq \calP(\textit{Act}^*)$, if 
 $p\Tran{\alpha} p'$ with $T(p')\cap S=\emptyset$ then
 there exists $q\Tran{\alpha} q'$ with $T(q')\cap S = \emptyset$.
\item The possible futures semantics is defined as:
 $p\sqsubseteq_{PF} q$ if 
 $p\Tran{\alpha} p'$ then
 there exists $q\Tran{\alpha} q'$ with $T(q')=T(p')$.
\end{enumerate}
\end{defi}

\begin{prop}
\hfill
\begin{enumerate}[\em(1)]
\item $\leq_T^{lf}$ is the possible futures preorder.
\item $\leq_T^{lf\supseteq}$ is the impossible futures preorder.
\end{enumerate}
\end{prop}
\begin{proof}
\hfill
\begin{enumerate}[(1)]
\item Obvious.
\item Assume that $p\leq_T^{lf\supseteq} q$.
 Then $p\Tran{\alpha}p'$, with $\alpha= a_1\dots a_n$, implies $q\Tran{\alpha}q'$
 with $T(q')\subseteq T(p')$.
 Therefore, if $p\Tran{\alpha} p'$ with $T(p')\cap X=\emptyset$
 then $q\Tran{\alpha} q'$ with $T(q')\cap X=\emptyset$
 which implies $p\sqsubseteq_{IF} p'$.

 Conversely, if $p\sqsubseteq_{IF} q$,
 $t=X_0a_0X_1\dots X_n\in\textit{LGO}_T(p)$ and
 $p\Tran{\alpha}p'$ with $\alpha=a_1\dots a_n$, obviously we have $T(p')\cap T(p')^c=\emptyset$, 
 where $T(p')^c$ just represent the complement of the set $T(p')$.
 Now applying the definition of $\sqsubseteq_{IF}$, we have some $q\Tran{\alpha}q'$ with $T(q')\cap T(p')^c=\emptyset$.
 Hence, there exists $t'=X_0'a_0X_1'\dots X_n'\in\textit{LGO}_T(q)$ with
 $T(q')\subseteq T(p')$, which implies $p\leq_T^{lf\supseteq} q$.\qedhere
\end{enumerate}
\end{proof}

\noindent As a matter of fact, the possible futures semantics is just below the 
2-nested simulation semantics in the spectrum only because
the trace simulation semantics is missing there.

At this point we are ready to present our first two ``missing links'', which
arise through the remaining two orders: $\leq_T^l$ and $\leq_T^{l\supseteq}$.

\begin{defi}\label{df:PF_IF}
The \emph{possible futures trace semantics} is defined by $\textrm{lgo}_T$'s
related by $\leq_T^l$ and the \emph{impossible futures trace semantics} is
defined by $\leq_T^{l\supseteq}$. 
\end{defi}

Let us complete this part of the new extended spectrum by introducing
the diamond generated by $\textrm{lgo}_{S}$'s.
This produces four new semantics coarser than 2-nested semantics.
For instance, for the case of failures we obtain the following
definition.

\begin{defi}
The \emph{extended simulation failures} of a process $p$ are defined as
\[\textit{ExtSimFailures}(p) = \{\langle \alpha, p''\rangle\mid \alpha\in A^*,
p\Tran{\alpha}p', p'\sqsubseteq_S p''\}.\]
The \emph{simulation failures} of a process $p$ are defined as 
$\textit{SimFailures}(p)=\{ \langle \alpha, B\rangle\mid
p\Tran{\alpha}p', B\cap \textit{BGO}_U(p')=\emptyset \}$.
We write $p\sqsubseteq_{SF} q$ iff
$\textit{SimFailures}(p)\subseteq \textit{SimFailures}(q)$.
\end{defi}

It can be proved that the inclusion 
$\textit{SimFailures}(p)\subseteq \textit{SimFailures}(q)$ holds
if and only if $\textit{ExtSimFailures}(p)\subseteq \textit{ExtSimFailures}(q)$.
Thus, simulation failures are essentially defined by translating the 
characterization of ordinary failures with the closure of readiness.

\begin{prop}
${\leq_S^{lf\supseteq}} = {\sqsubseteq_{SF}}$.
\end{prop}
\begin{proof}
Analogous to the characterization of $\sqsubseteq_F$ in terms of
$\leq_I^{lf\supseteq}$.
\end{proof}

\subsection{Deterministic branching observations}\label{deterministic:sec}

\begin{defi}\label{branch_obs_def}
\hfill
\begin{enumerate}[(1)]
\item
We say that a bgo is \emph{deterministic} if the set of children 
$\{(a_i, \textit{bgo}_i)\}$ of every node satisfies $a_i \neq a_j$ whenever
$i\neq j$.
We denote with $\textit{dBGO}_N$ the set of deterministic observations in
$\textit{BGO}_N$.
\item
The set of deterministic branching observations (dbgo for short) of a process 
$p$ is $\textit{dBGO}_N(p) = \textit{BGO}_N(p) \cap \textit{dBGO}_N$.
\item
We write $p\leq_N^\mathit{db} q$ if 
$\textit{dBGO}_N(p) \subseteq \textit{dBGO}_N(q)$. 
\end{enumerate}
\end{defi}

\noindent Like the linear observations, the set $\textit{dBGO}_N(p)$ can be defined
recursively and the corresponding semantics, compositionally.

\begin{exa}
For the two processes $p=a(bc+bd)$ and $q=abc+abd$ we have that both
deterministic observations in Figure~\ref{bopw1:fig} belong to
$\textit{dBGO}_I(p)$ and $\textit{dBGO}_I(q)$.
Indeed, that must be the case since it is easy to check that 
$\textit{dBGO}_I(p)=\textit{dBGO}_I(q)$.
\begin{figure}[t]
\[
\begin{array}{c@{\qquad\qquad\quad}c}
\xymatrix{
\{a\}\ar[d]^a\\
\{b\}\ar[d]^b\\
\{c\}
}
&
\xymatrix{
\{a\}\ar[d]^a\\
\{b\}\ar[d]^b\\
\{d\}
}
\end{array}
\]
\caption{Deterministic branching observations.}
\label{bopw1:fig}
\end{figure}
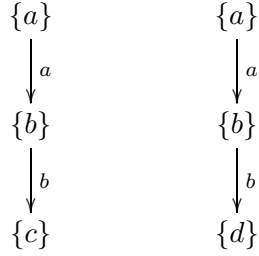
\end{exa}

In order to prove that $\textrm{dbgo}$'s for the constraint $I$ 
characterize the possible worlds
semantics we first recall the definition of that semantics in 
\cite{Gla01}.

\begin{defi}
A deterministic process $p$ is a possible world of a process $q$ if 
$p\sqsubseteq_{RS} q$.
The set of possible worlds of $p$ is denoted by $PW(p)$.
We define the order $p\sqsubseteq_{PW} q$ iff $PW(p)\subseteq PW(q)$.
\end{defi}

When defining the possible worlds of a process we have to solve all the
non-deterministic choices in it, each choice leading to one of its
possible worlds.
The same idea supports the selection of dbgo's to characterize this
semantics: the non-deterministic branching observations in
$\textit{BGO}_N(p)$ are not present in $\textit{dBGO}_N(p)$, where we have 
instead all the possible deterministic subtrees of every branching observation.

In our proof below we will relate the $\textrm{dbgo}$'s in 
$\textit{dBGO}_I(p)$ and the possible worlds in $PW(p)$. 
When necessary, we will consider observations in $\textit{dBGO}_I(p)$ as 
processes in BCCSP by removing the information from their nodes; 
by abuse of notation we will also denote with \textit{dbgo} the process obtained 
after such a removal.
Also, we call \emph{complete} those
observations that, for every node labeled by an offering $A$, have a branch 
labeled by each of the actions in $A$.

\begin{defi}
The set of \emph{complete deterministic} branching observations for the local 
observation function $L_I$ is the set $\textit{cdBGO}_I\subseteq
\textit{dBGO}_I$ recursively defined as:
\begin{iteMize}{$\bullet$}
\item $\langle \emptyset, \emptyset\rangle\in \textit{cdBGO}_I$.
\item $\langle A, \{(a, \textit{cdbgo}_a)\mid a\in A\}\rangle\in
 \textit{cdBGO}_I$ for every $a\in A$ and $\textit{cdbgo}_a\in\textit{cdBGO}_I$.
\end{iteMize}
For each $p\in\textrm{BCCSP}$ we define its set of complete deterministic 
branching observations
$\textit{cdBGO}_I(p) = \textit{dBGO}_I(p)\cap \textit{cdBGO}_I$.
\end{defi}

We also associate to a deterministic process $q$ its universal (complete
deterministic) branching observation.

\begin{defi}
For a deterministic process $p$, its \emph{universal deterministic} branching
observation $\textit{cdbgo}(p)$ is:
\begin{iteMize}{$\bullet$}
\item $\textit{cdbgo}(\cero) = \langle\emptyset, \emptyset\rangle$.
\item $\textit{cdbgo}(\sum_{a\in A} ap_a) = 
  \langle A,\{(a,\textit{cdbgo}(p_a))\mid a\in A\}\rangle$.
\end{iteMize}
\end{defi}

The following result is now immediate.

\begin{prop}
For every $p\in \textrm{BCCSP}$, $\textit{cdbgo}(p)\in \textit{cdBGO}_I(p)$.
\end{prop}

\begin{lem}\label{pw-aux-1:lem}
For every $q\in PW(p)$, $\textit{cdbgo}(q)\in\textit{cdBGO}_I(p)$.
\end{lem}
\begin{proof}
By structural induction on $q$:
\begin{iteMize}{$\bullet$}
\item If $q$ is $\cero$, then $p\equiv \cero$ and 
 $\langle \emptyset, \emptyset\rangle\in\textit{cdBGO}_I(\cero)$.
\item If $q$ is $\sum a q_a$, since $q\in PW(p) $ we have
  $q\sqsubseteq_{RS} p$.  This implies $I(q) = I(p)$ and that, for all
  $a\in A$, there exists $p\tran{a} p_a$, $q_a\sqsubseteq_{RS} p_a$,
  so that $q_a\in PW(p_a)$.  By induction hypothesis,
  $\textit{cdbgo}(q_a)\in \textit{cdBGO}_I(p)$.  Now, by definition,
  $\textit{cdbgo}(q) = \langle A, \{(a,\textit{cdbgo}(q_a))\mid a\in
  A)\}\rangle$ and, from $p\tran{a} p_a$ and $I(p) = I(q)$, we
  conclude $\textit{cdbgo}(q)\in \textit{dBGO}_I(p)$ and therefore
  $\textit{cdbgo}(q)\in \textit{cdBGO}_I(p)$.\qedhere
\end{iteMize}
\end{proof}

\begin{lem}\label{pw-aux-2:lem}
For every process $q$ such that $\textit{cdbgo}(q)\in\textit{cdBGO}_I(p)$
we have $q\sqsubseteq_{RS} p$ and therefore $q\in PW(p)$.
\end{lem}
\begin{proof}
We will prove that the set $S=\{(q,p)\mid \textit{cdbgo}(q)\in 
\textit{cdBGO}_I(p)\}$ is a ready simulation.
Obviously, for $(q,p)\in S$ it is $I(q) = I(p)$ and, if $q\tran{a} q_a$, there
exists $p\tran{a}p_a$ with $\textit{cdbgo}(q_a)\in\textit{cdBGO}_I(p_a)$,
which shows that $(q_a,p_a)\in S$ and that $S$ is a ready simulation.
\end{proof}

\begin{thm}\label{possible-worlds:thm}
For all processes $p_1, p_2\in \textrm{BCCSP}$, $p_1\sqsubseteq_{PW} p_2$ iff
$p_1\leq_I^{db} p_2$.
\end{thm}
\begin{proof}
$(\Leftarrow)$
For $q\in PW(p_1)$, by Lemma~\ref{pw-aux-1:lem} we have 
$\textit{cdbgo}(q)\in\textit{cdBGO}_I(p_1)$ and therefore 
$\textit{cdbgo}(q)\in\textit{cdBGO}_I(p_2)$.
Now, by Lemma~\ref{pw-aux-2:lem}, $q\sqsubseteq_{RS} p_2$ and thus 
$q\in PW(p_2)$.

$(\Rightarrow)$
Let $\textit{dbgo}\in dBGO_I(p_1)$: by definition of $\textit{dBGO}_I(p_1)$ it is
clear that we can extend \textit{dbgo} into some 
$\textit{dbgo}'\in\textit{cdBGO}_I(p_1)$.
Now, by Lemma~\ref{pw-aux-2:lem}, $\textit{dbgo}'\sqsubseteq_{RS} p_1$ 
(taking $\textit{dbgo}'$ as a deterministic process).
Therefore, $\textit{dbgo}'\in PW(p_1)$ and thus $\textit{dbgo}'\in PW(p_2)$ and,
by Lemma~\ref{pw-aux-1:lem}, $\textit{cdbgo}(\textit{dbgo}') = \textit{dbgo}'\in
\textit{cdBGO}_I(p_2)$: hence $\textit{dbgo}\in \textit{dBGO}_I(p_2)$ as 
required.
\end{proof}

\begin{rem}
If we consider infinite processes, then our characterization of $\sqsubseteq_{PW}$ by means of $\leq_I^{db}$ only works if we restrict ourselves to image-finite processes. We will continue the discussion on this part when studying the logical characterization of this semantics at Section \ref{sec:LogicalCharacterization}.
\end{rem}

Let us briefly consider the remaining new semantics definable
by means of deterministic branching observations.
It is clear that in all cases the corresponding orders verify
${\leq_N^b}\subseteq {\leq_N^{db}} \subseteq {\leq_N^l}$, so that the associated
semantics will be situated between the corresponding semantics defined by
branching observations in $\textit{BGO}_N$ and linear
observations in $\textit{LGO}_N$, as is the case for the possible worlds
semantics, located between the ready simulation semantics and the ready
trace semantics.

Admittedly, most of these semantics are rather strange and this is probably
the reason why, as far as we know, they have not been previously considered.
However, the simplest of them all, that corresponding to $N= U$, has 
properties similar to the possible worlds semantics and, in fact, can be
defined by simply removing from its definition the ``$R$''
in the condition $q\sqsubseteq_{RS} p$.
Hence, we can regard as possible worlds those deterministic implementations
where we offer just a part of the action offered by the given process.

\begin{defi}\label{PW_order}
The partial possible worlds of a process $p$ are those deterministic
processes that verify $q\sqsubseteq_S p$.
We denote with $PW_U(p)$ the set of partial possible worlds of a process
$p$ and define $p\sqsubseteq_\mathit{UPW} q$ if $PW_U(p)\subseteq PW_U(q)$.
\end{defi}

\begin{prop}
For all processes $p_1, p_2\in \textrm{BCCSP}$, $p_1\sqsubseteq_\mathit{UPW} p_2$
iff $p_1\leq_U^{db} p_2$.
\end{prop}
\begin{proof}
Similar to Theorem~\ref{possible-worlds:thm}, simplified by the fact that all
\textit{dbgo} in $PW_U(p)$ satisfy $\textit{dbgo}\sqsubseteq_S p$.
\end{proof}

\begin{exa}
We have $a\sqsubseteq_\mathit{UPW} a+b$ since
$\langle \cdot, \{(a,\emptyset)\}\rangle\in\textit{dBGO}_U(a+b)$.
By contrast, for $p=ab+ac$ and $q=a(b+c)$ we have $p\sqsubseteq_\mathit{UPW}q$
but $q\not\sqsubseteq_\mathit{UPW} p$ because 
$\langle \cdot, \{(a,\langle \cdot,\{ (b,\langle \cdot,\emptyset\rangle),
(c,\langle \cdot,\emptyset\rangle) \}\rangle\}\rangle\in 
\textit{dBGO}_U(q) - \textit{dBGO}_U(p)$.
\end{exa}

Analogously, for any other constraint $N$ we could define the $N$-possible
worlds order $\sqsubseteq_\mathit{NPW}$ using $\sqsubseteq_\mathit{NS}$ instead 
of $\sqsubseteq_\mathit{S}$ at Definition \ref{PW_order}.
However, it is easy to see that when $N$ is fine enough, e.g. $N=T$, this order would become
totally wrong. Instead, we can still consider the observations in $\textit{dBGO}_N$ and by means
of them we define the ``reasonable'' deterministic branching semantics, for any layer in the spectrum.

The extended spectrum can now be depicted as in 
Figures~\ref{fig:extended-ltbts} and \ref{fig:Slice}.

\begin{figure}[tbp]
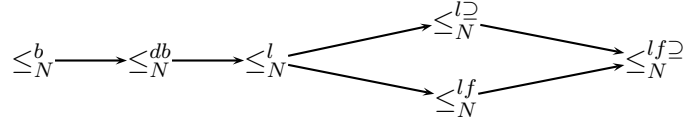

  \begin{center}
    \begin{tabular*}{.6\textwidth}[h]{@{\extracolsep{\fill}}ccccccc}

      &  & & &\rnode{nlr}{$\leq_N^{l\supseteq}$} &  \\
      \rnode{ns}{$\leq^b_N$}& \rnode{dbs}{$\leq^{db}_N$}& 
        \rnode{nl}{$\leq_N^{l}$} & & & &\rnode{nlfr}{$\leq_N^{lf\supseteq}$}\\
      & & & & \rnode{nlf}{$\leq_N^{lf}$} \\
      \ncline{->}{ns}{dbs}
      \ncline{->}{dbs}{nl}
      \ncline{->}{nl}{nlr}
      \ncline{->}{nl}{nlf}
      \ncline{->}{nlr}{nlfr}
      \ncline{->}{nlf}{nlfr}
      \\
    \end{tabular*}
    \caption{Basic layer in the linear time-branching time spectrum.}
    \label{fig:Slice}
  \end{center}
\end{figure}
\begin{figure}[tbp]
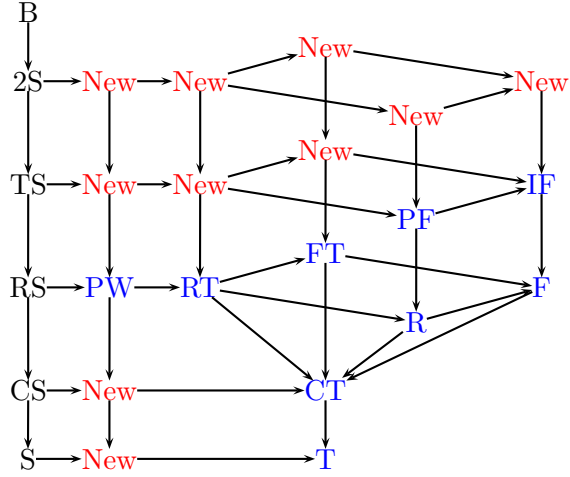

{
 \begin{center}
   \begin{tabular*}{.5\textwidth}[h]{@{\extracolsep{\fill}}cccccccc}

     \rnode{bs}{B} \\

     & & & & \rnode{n7}{\textcolor{red}{New}}  \\
     \rnode{n2s}{2S}& \rnode{pw1}{\textcolor{red}{New}} &
       \rnode{n4}{\textcolor{red}{New}} && &&&\rnode{n5}{\textcolor{red}{New}}\\
     & & & & & \rnode{n6}{\textcolor{red}{New}} \\ 
     \ncline{->}{n2s}{pw1}
     \ncline{->}{pw1}{n4}
     \ncline{->}{n4}{n6}
     \ncline{->}{n4}{n7}
     \ncline{->}{n6}{n5}
     \ncline{->}{n7}{n5}

     \ncline[linestyle=solid]{->}{bs}{n2s}
     & & & & \rnode{pf}{\textcolor{red}{New}} &  \\
     \rnode{ts}{TS}& \rnode{pw2}{\textcolor{red}{New}} &
       \rnode{n1}{\textcolor{red}{New}} && &&&\rnode{n2}{\textcolor{blue}{IF}}\\
     & & & & & \rnode{n3}{\textcolor{blue}{PF}} \\
     \ncline{->}{ts}{pw2}
     \ncline{->}{pw2}{n1}
     \ncline{->}{n1}{pf}
     \ncline{->}{n1}{n3}
     \ncline{->}{pf}{n2}
     \ncline{->}{n3}{n2}

     \ncline[linestyle=solid]{->}{n2s}{ts}
     \ncline[linestyle=solid]{->}{pw1}{pw2}
     \ncline[linestyle=solid]{->}{n4}{n1}
     \ncline[linestyle=solid]{->}{n5}{n2}
     \ncline[linestyle=solid]{->}{n6}{n3}
     \ncline[linestyle=solid]{->}{n7}{pf}
      & & & & \rnode{ft}{\textcolor{blue}{FT}} &  \\
     \rnode{rs}{RS}& \rnode{pw}{\textcolor{blue}{PW}} &
       \rnode{rt}{\textcolor{blue}{RT}} & & & &&\rnode{f}{\textcolor{blue}{F}}\\
     & & & & &\rnode{r}{\textcolor{blue}{R}} \\
     \ncline{->}{rs}{pw}
     \ncline{->}{pw}{rt}
     \ncline{->}{rt}{ft}
     \ncline{->}{rt}{r}
     \ncline{->}{ft}{f}
     \ncline{->}{r}{f}

     \ncline[linestyle=solid]{->}{ts}{rs}
     \ncline[linestyle=solid]{->}{pw2}{pw}
     \ncline[linestyle=solid]{->}{n1}{rt}
     \ncline[linestyle=solid]{->}{pf}{ft}
     \ncline[linestyle=solid]{->}{n3}{r}
     \ncline[linestyle=solid]{->}{n2}{f}
     \\

     \rnode{cs}{CS}& \rnode{pw3}{\textcolor{red}{New}} & &&
       \rnode{ct}{\textcolor{blue}{CT}} & \\
     \ncline{->}{cs}{pw3}
     \ncline{->}{pw3}{ct}

     \ncline[linestyle=solid]{->}{rs}{cs}
     \ncline[linestyle=solid]{->}{pw}{pw3}
     \ncline[linestyle=solid]{->}{f}{ct}
     \ncline[nodesep=1pt,linestyle=solid]{->}{r}{ct}
     \ncline[linestyle=solid]{->}{ft}{ct}
     \ncline[linestyle=solid]{->}{rt}{ct}
     \\
     \rnode{s}{S}& \rnode{pw4}{\textcolor{red}{New}} & &&
       \rnode{t}{\textcolor{blue}{T}} & 
     \ncline{->}{s}{pw4}
     \ncline{->}{pw4}{t}

     \ncline[linestyle=solid]{->}{cs}{s}
     \ncline[linestyle=solid]{->}{pw3}{pw4}
     \ncline[linestyle=solid]{->}{ct}{t}
   \end{tabular*} 
 \end{center}
}
 \caption{Semantics in the new linear time-branching time spectrum.}
 \label{fig:extended-ltbts}
\end{figure}

\subsection{Back to branching observations}

The orders $\leq^{l\delta}_N$ with $\delta\in\{\supseteq,f,f{\supseteq}\}$ that characterize
some of the linear semantics studied in Section~\ref{linear:sec}, restricted in several ways the use
of the local information, when characterizing those semantics.
The same scheme can be generalized to the branching observations.
This way, for each constraint $N$ we would obtain three new branching
semantics based on $\textrm{bgo}$'s in $\textit{BGO}_N$ which, together with
the original $N$-simulation semantics, would constitute a diamond of
branching semantics at a higher layer in our extended ltbt spectrum.
The introduction of these new semantics also offers a clearer
view of the spectrum, with two main levels of branching and linear
semantics and an intermediate one of deterministic branching semantics.
Although this provides the means for obtaining a host of new semantics, it
is also true that most of them are bizarre, 
in sharp contrast with the fact that the corresponding orders gave rise to 
interesting semantics when applied to linear observations.

To illustrate the comments above, next we consider in some
detail the case $N=I$, which corresponds to the most interesting semantics.

\begin{defi}
For $\textit{bgo}, \textit{bgo}'\in \textit{BGO}_I$ we define:
\begin{iteMize}{$\bullet$}
\item $ \textit{bgo}\leq_I^\supseteq\textit{bgo}'\iff
  \begin{array}[t]{l}
  \big(\textit{bgo}=\langle A_1, S_1\rangle \textrm{ and } 
  \textit{bgo}'=\langle A_2, S_2\rangle \textrm{ and } A_1 \supseteq A_2
  \textrm{ and } \\
  \hphantom{\big(}S_1 = \{(a_i,\textit{bgo}_i)\mid i\in I\} \textrm{ and }
  S_2 = \{(a_i,\textit{bgo}'_i)\mid i\in I\} \textrm{ and }\\
  \hphantom{\big(}\textrm{for all }i\in I\ 
     (\textit{bgo}_i\leq_I^\supseteq \textit{bgo}_i')\big)\,.
  \end{array}$
 \smallskip
\item $\textit{bgo}\leq_I^f\textit{bgo}'\iff
 \begin{array}[t]{l}
 \big(\textit{bgo}=\langle A_1, \emptyset\rangle \textrm{ and }
 \textit{bgo}'=\langle A_1, \emptyset\rangle\big)\textrm{ or }\\
 \big(\textit{bgo}=\langle A_1, S_1\rangle \textrm{ and } 
  \textit{bgo}'=\langle A_2, S_2\rangle
  \textrm{ and } \\
  \hphantom{\big(}S_1 = \{(a_i,\textit{bgo}_i)\mid i\in I\} \textrm{ and }
  S_2 = \{(a_i,\textit{bgo}'_i)\mid i\in I\} \textrm{ and }\\
  \hphantom{\big(}\textrm{for all } i\in I\ 
     (\textit{bgo}_i\leq_I^f \textit{bgo}_i')\big)\,.
 \end{array}$
 \smallskip
\item $\textit{bgo}\leq_I^{f\supseteq}\textit{bgo}'\iff
 \begin{array}[t]{l}
  \big(\textit{bgo}=\langle A_1, \emptyset\rangle \textrm{ and }
  \textit{bgo}'=\langle A_2, \emptyset\rangle \textrm{ and }
  A_1\supseteq A_2\big) \textrm{ or }\\
 \big(\textit{bgo}=\langle A_1, S_1\rangle \textrm{ and } 
  \textit{bgo}'=\langle A_2, S_2\rangle
  \textrm{ and } \\
  \hphantom{\big(}S_1 = \{(a_i,\textit{bgo}_i)\mid i\in I\} \textrm{ and }
  S_2 = \{(a_i,\textit{bgo}'_i)\mid i\in I\} \textrm{ and }\\
  \hphantom{\big(}\textrm{for all }i\in I\ 
        (\textit{bgo}_i\leq_I^{f\supseteq} \textit{bgo}_i')\big)\,.
 \end{array}$
\end{iteMize}
\end{defi}

\begin{defi}
For $\calB, \calB'\subseteq \textit{BGO}_I$ and 
$\delta\in\{\supseteq,f,f{\supseteq}\}$, we define the orders $\leq_I^{b\delta}$ by:
\begin{iteMize}{$\bullet$}
\item $\calB\leq_I^{b\delta} \calB' \iff 
 \textrm{for all}\ \textit{bgo}\in\calB
 \textrm{ there exists }\textit{bgo}'\in\calB'
 \textrm{ with } \textit{bgo}\leq_I^{\delta}\textit{bgo}'$.
\end{iteMize}
Then, we write $p\leq_I^{b\delta} q$ if 
$\textit{BGO}_I(p)\leq_I^{b\delta} \textit{BGO}_I(q)$.
\end{defi}

It is somewhat surprising to discover that ${\leq_I^{b\supseteq}} ={\leq_I^b}$,
since this was not the case for their linear ``projections''
$\leq_I^{l\supseteq}$ and $\leq_I^l$.

\begin{prop}
For all processes $p_1, p_2 \in \textrm{BCCSP}$, $p_1\leq_I^{b\supseteq} p_2$ iff
$p_1\leq_I^{b} p_2$.
\end{prop}
\begin{proof}
Assume that $p_1\leq_I^{b\supseteq} p_2$ and let 
$\textit{bgo}\in\textit{BGO}_I(p_1)$: it is clear that it can be 
extended into a complete $\textit{cbgo}\in\textit{BGO}_I(p_1)$.
Then, there exists some $\textit{cbgo}'\in \textit{BGO}_I(p_2)$ with
$\textit{cbgo}\leq_I^{\supseteq} \textit{cbgo}'$ and, since $\textit{cbgo}$ is
complete, $\textit{cbgo}=\textit{cbgo}'$ and hence 
$\textit{bgo}\in\textit{BGO}_I(p_2)$.
The other implication is trivial.
\end{proof}

\begin{exa}
For $p_1 = a(b+c)$ and $p_2= ab+ac$, $p_1\leq_I^{l\supseteq}p_2$
but $p_1\not\leq_I^l p_2$.
However, $p_1\not\leq_I^{b\supseteq} p_2$ since for
$\textit{bgo} = 
\langle \{a\}, (a,\langle\{b,c\}, \{(b,\emptyset),(c,\emptyset)\}\rangle)\rangle
\in \textit{BGO}_I(p_1)$
there is no $\textit{bgo}'\in \textit{BGO}_I(p_2)$ with
$\textit{bgo}\leq_I^{l\supseteq} \textit{bgo}'$.
\end{exa}

By contrast, the branching semantics defined by $\leq_I^{bf}$ and
$\leq_I^{bf\supseteq}$ are indeed new.

\begin{exa}
For the processes $p$ and $q$ in Figure~\ref{weird-sem:fig},
$p\leq_I^{bf} q$ but $p\not\leq_I^b q$.
\begin{figure}[htbp]
\[
\begin{array}{c@{\quad}c@{\quad}c}
p & q & q'
\vspace{0.2cm}
\\
\xymatrix{
\cdot\ar[d]^a\\
\cdot\ar[d]^a\\
\cdot\ar[d]^a\\
\cdot\ar[d]^a\\
\cdot
}
&
\xymatrix{
&&\cdot\ar[dll]_a\ar[d]^a\ar[drr]^a\\
\cdot\ar[d]^a &&\cdot\ar[dl]_a\ar[d]^b &&\cdot\ar[dl]_a\ar[d]^b\\
\cdot&\cdot\ar[d]_a&\cdot&\cdot\ar[d]_a\ar[dr]^b&\cdot\\
&\cdot & &\cdot\ar[d]_a&\cdot\\
 &  &  & \cdot
}
&
\xymatrix{
&&\cdot\ar[dll]_a\ar[dl]^a\ar[d]^a\ar[drr]^a\\
\cdot &\cdot\ar[dl]_a\ar[d]_b&\cdot\ar[d]_a\ar[dr]^b &&\cdot\ar[d]_a\ar[dr]^b\\
\cdot&\cdot&\cdot\ar[d]_a\ar[dr]^b&\cdot&\cdot\ar[d]_a\ar[dr]^b&\cdot\\
& &\cdot &\cdot&\cdot\ar[d]_a\ar[dr]^b&\cdot\\
 &  &  & & \cdot&\cdot
}
\end{array}
\]
\caption{Three processes.}
\label{weird-sem:fig}
\end{figure}
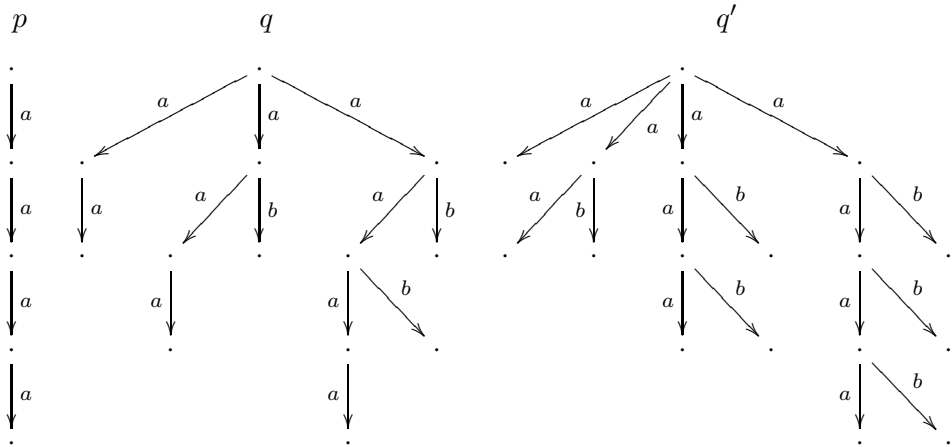
\end{exa}

This example shows that it is quite difficult to characterize this semantics
as a simulation one.
Furthermore, we conjecture that it is not finitely axiomatizable in the classic way (that means using only unconditional axioms).
As a matter of fact, we were also unable to find any conditional axiomatization, what we interpret as a ``proof'' of 
the fact that these new branching semantics are quite strange.

\begin{defi}
We say that $R\subseteq \textit{BGO}_I\times \textrm{BCCSP}$ is a 
\emph{final-ready simulation} when:
\begin{iteMize}{$\bullet$}
\item $(\langle A, \emptyset\rangle, q)\in R$ implies $I(q) = A$.
\item $(\langle A, \{(a_i, \textit{bgo}_i)\}\rangle, q) \in R$ implies that
 for all $i\in I$ there exists $q\tran{a_i} q_i$ such that 
 $(\textit{bgo}_i,q_i)\in R$.
\end{iteMize}
We say that $p$ is final-ready simulated by $q$ when for all 
$\textit{bgo}\in\textit{BGO}_I(p)$ there exists a final-ready simulation
$R$ with $(\textit{bgo},q)\in R$, and write $p \sqsubseteq_\mathit{fRS}q$.
\end{defi}

\begin{thm}
For all $p,q\in\textrm{BCCSP}$, $p\sqsubseteq_\mathit{fRS} q$ iff
$p\leq_I^{bf} q$.
\end{thm}

\begin{exa}
It is easy to check that for $p$ and $q'$ as in Figure~\ref{weird-sem:fig} we
have $p\leq_I^{bf{\supseteq}} q'$ but $p\not\leq_I^{bf} q'$.
\end{exa}

Final failure simulations are defined exactly like final-ready simulations
but substituting $I(q)\subseteq \textit{Act}$ for $I(q)= A$ in the first
clause, giving rise to
the order $\sqsubseteq_\mathit{fFS}$ between processes.

\begin{thm}
For all $p,q\in\textrm{BCCSP}$, $p\sqsubseteq_\mathit{fFS} q$ iff
$p\leq_I^{bf\supseteq} q$.
\end{thm}

As previously noted, these are certainly bizarre semantics but we 
believe it is interesting to indicate their existence because, by analogy to the linear case, their definitions
in terms of branching observations look
quite natural.
However, it also seems that when dealing with branching observations the 
introduction of any kind of asymmetry in the treatment of local observations 
produces quite involved semantics.

\section{Relating the observational and equational frameworks}
\label{rnoef:sec}

In this section we tie up all loose ends and show how our unification
theory is fully self-contained. Namely, we prove the results on
axiomatic characterizations in 
Section~\ref{sec:EquationalSemantics} from the observational 
semantics developed in 
Section~\ref{observational-sem-sec}: we show how the equations are
deduced from the observations in a general way without resorting to 
the already existing axiomatizations.

One of the key points of this section is to illustrate how the 
particular proofs needed in Section~\ref{sec:EquationalSemantics} 
for every one of the semantics can be replaced by a generic 
proof that stands for a whole family of semantics. In fact, we will
show in Section~\ref{sec:RealDiamond} that the same proof is still
valid for the new semantics suggested in Roscoe's work.

\subsection{Semantics coarser than ready simulation}\label{rnoef:sec1}

Let us now see how, from this uniform definition of the linear 
semantics, the proofs of the correctness and completeness of the corresponding 
axiomatizations can be derived in a uniform way avoiding the case 
analyses of Sections~\ref{anamps:sec} and \ref{tcss:sec}.
Although this could be done generically, with $N\in\{U, C, I, T\}$, we prefer
to start with the particular case $N= I$, which corresponds to the most
popular semantics already studied in Section~\ref{anamps:sec}.

To start with, we show how the axiomatizations can be synthetized
from the observational characterizations.
Our general axiom $(\textit{ND})$ for the reduction of non-determinism
specifies the hypothesis $M(x,y,w)$ under which the process 
$ax+ a(y+w)$ can be (syntactically) expanded by adding a new summand $a(x+y)$
without changing its semantics.
Then, let us compare the two sides of our general axiom.
Since $I(ax + a(y+w)) = I(ax) = I(a(y+w)) = \{a\}$, we have
\[
\begin{array}{lcl}
\textit{LGO}_I(ax+a(y+w)) &= &\textit{LGO}_I(ax)\cup \textit{LGO}_I(a(y+w)),\\
\textit{LGO}_I(a(y+w)) &= &\{\langle\{a\}\rangle\}\cup {}\\
    &&         \{\langle\{a\},a, I(y+w)\rangle\circ S \mid
               \langle \{a\},a,I(y)\rangle\circ S\in \textit{LGO}_I(ay) \lor{} \\
    && \hphantom{\{\langle\{a\},a, I(y+w)\rangle\circ S \mid{}}
               \langle \{a\},a,I(w)\rangle\circ S\in \textit{LGO}_I(aw) \}.
\end{array}
\]
Notice then that the observations of $a(y+w)$ are exactly those of $ay+aw$ 
simply replacing $I(y)$ or $I(w)$, respectively, by $I(y+w) = I(y)\cup I(w)$.
Analogously,
\[
\begin{array}{lcl}
\textit{LGO}_I(a(x+y)) &= &\{\langle\{a\}\rangle\}\cup {}\\
    &&         \{\langle\{a\},a, I(x+y)\rangle\circ S \mid
               \langle \{a\},a,I(x)\rangle\circ S\in \textit{LGO}_I(ax) \lor{} \\
    && \hphantom{\{\langle\{a\},a, I(x+y)\rangle\circ S \mid{}}
               \langle \{a\},a,I(y)\rangle\circ S\in \textit{LGO}_I(ay) \}.
\end{array}
\]

Now, in order to get the adequate condition $M_Z(x,y,w)$ for each of the 
semantics, let us examine the formulas that define the preorders $\leq_I^{lY}$:
\begin{iteMize}{$\bullet$}
\item $\leq_I^l$. 
To have
$\textit{LGO}_I(a(x+y))\subseteq \textit{LGO}_I(ax)\cup 
\textit{LGO}_I(a(y+w))$ it is enough to require
$\{\langle\{a\},a, I(x)\cup I(y)\rangle \circ S
  \mid \langle I(x) \rangle\circ S\in \textit{LGO}_I(x)\}
\subseteq
\{\langle\{a\},a, I(x)\rangle \circ S
  \mid \langle I(x) \rangle\circ S\in \textit{LGO}_I(x)\}$
and
$\{\langle\{a\},a, I(x)\cup I(y)\rangle \circ S
  \mid \langle I(y) \rangle\circ S\in \textit{LGO}_I(y)\}
\subseteq
\{\langle\{a\},a, I(y)\cup I(w)\rangle \circ S
  \mid \langle I(y) \rangle\circ S\in \textit{LGO}_I(y)\}$.
Thus, a first proposal for $M_{RT}$ would be
\[
I(y)\subseteq I(x) \land I(x) = I(y) \cup I(w).
\]
However, due to the fact that this axiom will be used in combination with
$(RS)$, the following, more restrictive but simpler form, can be used instead:
\[
M_{RT}(x,y,w) \iff I(x) = I(y) \land I(w) \subseteq I(y).
\]
Clearly, this form is stronger than the condition synthetized above.
Reciprocally, $a(x+y)\preceq ax + a(y+w)$ can be proved from the assumptions  
$I(y)\subseteq I(x) $ and $I(x) = I(y) \cup I(w)$ using $(RS)$ first to get
$a(x+y)\preceq a(x+y+w)$, and then $(\textit{ND\/})$ instantiated with
$M_{RT}$ to obtain $a(x + y + w)\preceq ax + a(y+w)$.
\item $\leq_I^{l{\supseteq}}$.\label{stronger-versions:pag}
We need the inclusion
$\ol{\textit{LGO}_I(a(x+y))}^{{\supseteq}} \subseteq 
\ol{\textit{LGO}_I(ax+a(y+w))}^{{\supseteq}}$
to hold.
Since $I(x)\cup I(y)\supseteq I(x)$, the general observations in 
$a(x+y)$ that arise from $x$ will be also in 
$\ol{\textit{LGO}_I(ax)}^{l{\supseteq}}$.
For those that arise from $y$, it is required that
$I(x)\cup I(y)\supseteq I(y)\cup I(w)$.
Once again, $(RS)$ can be used to simplify this condition into the simpler
\[
M_{FT}(x,y,w) \iff I(w) \subseteq I(y).
\]
The less restrictive variant of the axiom can be derived from the stronger
one and $(RS)$ as follows.
Taking $w = \cero$, since $I(\cero)\subseteq I(y)$ we obtain
$a(x+y)\preceq ax + ay$ from $(\textit{ND\/}^\mathit{FT})$; in particular,
$a(x+y+w) \preceq ax + a(y+w)$.
Also, by $(RS)$, $x+y\preceq (x+y) + (x+y+w)$, from where it follows
$a(x+y) \preceq a(x+y+w)$. 

\item $\leq_I^{lf}$.
We consider the inclusion
$\ol{\textit{LGO}_I(a(x+y))}^{f} \subseteq \ol{\textit{LGO}_I(ax+a(y+w))}^{f}$.
We only have to consider the lgo
$\langle \{a\},a, I(x)\cup I(y)\rangle$ in $\ol{\textit{LGO}_I(a(x+y))}^f$ and
show that it also belongs to $\ol{\textit{LGO}_I(ax+a(y+w))}^f$, since all lgo's
of length greater than 1 start with the prefix $\langle\{a\}, a\rangle$.
For that, either $I(x)\cup I(y) = I(x)$ or $I(x)\cup I(y) = I(y)\cup I(w)$, that
is, $I(y)\subseteq I(x)$ or $I(x)\cup I(y)= I(y)\cup I(w)$.
Again, we can remove the second condition and define
\[
M_R(x,y,w) \iff I(y)\subseteq I(x)
\]
since, whenever $I(x)\cup I(y)= I(y)\cup I(w)$,
$a(x + y + w)\preceq ax + a(y+w)$ can be obtained by
taking $x = y+w$, $y=x$, and $w=\cero$, and then by applying 
$(RS)$ we conclude $a(x+y)\preceq ax+a(y+w)$.

\item $\leq_I^{lf{\supseteq}}$.
An argument analogous to the previous one leads us to check that
$I(x)\cup I(y)\supseteq I(x)$ or $I(x)\cup I(y)\supseteq I(y)\cup I(w)$,
and the first is certainly true.
\end{iteMize}

In order to prove the completeness of our axiomatizations we introduce the
following notions of head normal forms.

\begin{defi}
For $p=\sum_{a\in X_0}\sum_{i\in I_a} ap_a^i$ and $Z\in\{F, R, \textit{FT},
\textit{RT}\}$, its $Z$-head normal form $\hnfx{p}$ is:
\begin{iteMize}{$\bullet$}
\item For $a\in X_0$, $i\in I_a$, and $X_1\subseteq \bigcup_{i\in I_a} 
I(p_a^i)$ such that $I(p_a^i)\subseteq X_1$, we define
$\textit{hnf\/}^Z(p,a,i,X_1) = a(p_a^i+ \sum\{ p_a^j|_{X_1}\mid j\neq i,
M_Z(p_a^i,p_a^j|_{X_1},p_a^j|_{\ol{X_1}})\})$.
\item $\hnfx{p} = p+\sum_{a\in X_0}\sum_{i\in I_a}\sum_{X_1\subseteq \bigcup I(p_a^i)}
 \textit{hnf\/}^Z(p,a,i,X_1)$.
\end{iteMize}
\end{defi}

It is clear that several redundancies arise in this definition: for example, 
if $Z= RT$ then 
$\textit{hnf\/}^Z(p,a,i,X_1)=\textit{hnf\/}^Z(p,a,i,I(p_a^i))$, so that the
argument $X_1$ would not be needed in this case.
We prefer to maintain the generic definition in order to allow a homogeneous 
treatment of all the semantics.

\begin{prop}\label{hnfx-aux:prop}
For $Z\in \{RT, FT, R, F\}$,
$\{\textrm{$B_1$--$B_4$}, (RS), (\textit{ND\/}^Z)\}\vdash \hnfx{p}\preceq p$.
\end{prop}
\begin{proof}
Let $p=\sum_{a\in X_0}\sum_{i\in I_a} a p_a^i$.

Considering the definition of $\textit{hnf\/}^Z(p,a,i,X_1)$, let us consider an enumeration
of the set of j's contributing to it: If $J_i=\{j\neq i \mid M_Z(p_a^i,p_a^{j_1}|_{X_1},p_a^{j_1}|_{\ol{X_1}})\}$, we take
$\{j_k\mid k=1\dots |J_i|\}=J_i$.

Then we can prove by induction on $l$ that for all $0\leq l \leq |J_i|$ we have 
$\{\textrm{$B_1$--$B_4$}, (RS), (\textit{ND\/}^Z)\}\vdash 
 a(p_a^i+\sum_{h=1}^l p_a^{j_h}|_{X_1})\preceq ap_a^i+ \sum_{h=1}^l ap_a^{j_h}$.
 
 The case of $l=0$ is trivial. Assuming the result for $l$, we prove the result for $l+1$.
From $M_Z(p_a^i, p_a^{j_{l+1}}|_{X_1},p_a^{j_{l+1}}|_{\ol{X_1}})$ we can infer $M_Z(p_a^i +\sum_{h=1}^lp_a^{j_h}|_{X_1}, p_a^{j_{l+1}}|_{X_1},p_a^{j_{l+1}}|_{\ol{X_1}})$
so that we can derive $\vdash 
 a(p_a^i+\sum_{h=1}^{l+1} p_a^{j_h}|_{X_1}) \preceq 
a(p_a^i+\sum_{h=1}^l p_a^{j_h})+ ap_a^{j_{l+1}}$;
and applying the i.h. we conclude
$\{\textrm{$B_1$--$B_4$}, (RS), (\textit{ND\/}^Z)\}\vdash 
 a(p_a^i+\sum_{h=1}^{l+1}p_a^{j_h}|_{X_1}|_{X_1}) \preceq ap_a^i+\sum_{h=1}^{l+1}ap_a^{j_h}$. 

From this we immediately obtain
$\{\textrm{$B_1$--$B_4$}, (RS), (\textit{ND\/}^Z)\}\vdash 
 \textit{hnf\/}^Z(p,a,i,X_1)\preceq p|_a$.
Finally, adding all these inequalities we conclude
$\{\textrm{$B_1$--$B_4$}, (RS), (\textit{ND\/}^Z)\}\vdash \hnfx{p}\preceq p$.
\end{proof}

Let us define
$l(F) = lf{\supseteq}$, $l(FT) = l{\supseteq}$, $l(R) = lf$, and $l(RT) = l$.
In order to apply structural induction to prove the completeness of the 
axiomatizations we need to show that, whenever 
$p=\sum_{a\in X_0}\sum_{i\in I_a} a p_a^i$
and $p\leq_I^{l(Z)} q$, there is a summand $ah_a^k$ of $\hnfx{q}$ such that
$p_a^i\leq_I^{l(Z)} h_a^k$ for each $a\in \textit{Act}$, $i\in I_a$.

\begin{prop}\label{hnfx-aux2:prop}
Let $Z\in \{F, FT, R, RT\}$, and let
$p=\sum_{a\in X_0}\sum_{i\in I_a} a p_a^i$, and
$q=\sum_{a\in X_0}\sum_{j\in J_a} a q_a^j$.
If $p\leq_I^{l(Z)} q$ then there exists a summand $ah_a^k$ of
$\hnfx{q}$ such that $p_a^i\leq_I^{l(Z)} h_a^k$.
\end{prop}
\begin{proof}
Using Definition \ref{obs:def} and Proposition  \ref{prop_contenido}, we need to show that there exists $\ol{\textit{LGO}_I(p_a^i)}^{l(Z)} \subseteq 
\ol{\textit{LGO}_I(h_a^k)}^{l(Z)}$ but, due to the fact that
 $\ol{(\_)}^{l(Z)}$ is a closure
operator (Proposition~\ref{closures:prop}), it is enough to prove that
$\textit{LGO}_I(p_a^i)\subseteq \ol{\textit{LGO}_I(h_a^k)}^{l(Z)}$.
For $\langle I(p_a^i)\rangle\in \textit{LGO}_I(p_a^i)$, since $p\leq_I^{l(Z)} q$ 
there is some $q_a^j$ such that $\langle I(p_a^i)\rangle\in 
{\textit{LGO}_I(q_a^j)}^{l(Z)}$; we then consider 
$\textit{hnf\/}^Z(q,a,j,I(p_a^i)) = ah^k_a$.

If $t\in \textit{LGO}_I(p_a^i)$ then $\langle I(p),a\rangle\circ t\in
\textit{LGO}_I(p)\subseteq\ol{\textit{LGO}_I(q)}^{l(Z)}$ and
there exists $j_t$ such that $t\in\ol{\textit{LGO}_I(q_a^{j_t})}^{l(Z)}$.
In addition, $M_Z(q_a^j, q_a^{j_t}|_{I(p_a^i)}, q_a^{j_t}|_{\ol{I(p_a^i)}})$:
\begin{iteMize}{$\bullet$}
\item If $Z=RT$, then $t\in \textit{LGO}_I(q_a^{j_t})$ and therefore 
$I(q_a^{j_t}) = I(p_a^i) = I(q_a^j)$.
Hence, condition $M_{RT}(q_a^j, q_a^{j_t}|_{I(p_a^i)}, \cero)$ holds and therefore 
$M_{RT}(q_a^j, q_a^{j_t}|_{I(p_a^i)}, q_a^{j_t}|_{\ol{I(p_a^i)}})$.
\item If $Z = FT$, from $t\in \ol{\textit{LGO}_I(q_a^{j_t})}^{\supseteq}$
it follows that $I(q_a^{j_t})\subseteq I(p_a^{j_t})$ and therefore
$I(q_a^{j_t}|_{\ol{I(p_a^i)}}) =\emptyset \subseteq I(q_a^{j_t}|_{I(p_a^i)})$. 
Hence, $M_{FT}(q_a^j, q_a^{j_t}|_{I(p_a^i)}, q_a^{j_t}|_{\ol{I(p_a^i)}})$.
\item If $Z=R$, from $\langle I(p_a^i)\rangle\in\ol{\textit{LGO}_I(q_a^{j})}^{f}$ 
we have that $I(p_a^i) = I(q_a^{j})$ and thus 
$I(q_a^{j_t}|_{I(p_a^i)}) \subseteq I(q_a^j)$ and
$M_{R}(q_a^j, q_a^{j_t}|_{I(p_a^i)}, q_a^{j_t}|_{\ol{I(p_a^i)}})$.
\item For $Z= F$ it is trivial since $M_F(x,y,w)$ is always true.
\end{iteMize}
Therefore $q_a^{j_t}$ is one of the summands of $h_a^k$ and, since
$t\in\ol{\textit{LGO}_I(q_a^{j_t})}^{l(Z)}$, we have $p_a^i\leq_I^{l(Z)} h_a^k$.
\end{proof}

\begin{thm}[Soundness and completeness]\label{sandc:thm}
For $Z\in \{RT, FT, R, F\}$:
\[
p\leq_I^{l(Z)} q \textrm{ iff } 
\{\textrm{$B_1$--$B_4$}, (RS), (\textit{ND\/}^\mathit{Z})\}\vdash p\preceq q.
\]
\end{thm}
\begin{proof}
(Soundness)
The axiomatizations are sound because of the way they have been derived.

(Completeness)
By structural induction on $p$.
\begin{iteMize}{$\bullet$}
\item Let $p$ be $\cero$. 
  As usual, we can consider terms up to bisimulation since
  \textrm{$B_1$--$B_4$} are equations needed for all the semantics.
  If $p\leq_I^{l(Z)} q$, then $q$ must be $\cero$ (or bisimilar to $\cero$)
  because the set of local observations of $\cero$ is empty and cannot contain
  any observations (see Definition~\ref{obs:def}.)
\item If $p=\sum_{a\in X_0} \sum_{i\in I_a} a p_a^i$ then, by 
Proposition~\ref{hnfx-aux2:prop}, $p\leq_I^{l(Z)} q$
implies that there exists a summand $ah_a^k$ of $\hnfx{q}$ such that 
$p_a^i\leq_I^{l(Z)} h_a^k$.
By induction hypothesis, 
$\{\textrm{$B_1$--$B_4$}, (RS), (\textit{ND\/}^Z)\}\vdash p_a^i\preceq h_a^k$
and therefore
$\{\textrm{$B_1$--$B_4$}, (RS),$ 
$(\textit{ND\/}^Z)\}\vdash ap_a^i\preceq ah_a^k$;
adding all these inequalities and using $(RS)$, which is allowed because
$I(p) = I(q)$, it follows that
$\{\textrm{$B_1$--$B_4$}, (RS),$ 
$(\textit{ND\/}^Z)\}\vdash p\preceq \hnfx{q}$ and,
by Proposition~\ref{hnfx-aux:prop},
$\{\textrm{$B_1$--$B_4$}, (RS), (\textit{ND\/}^Z)\}\vdash p\preceq q$.\qedhere
\end{iteMize}
\end{proof}

\subsection{The semantics that are not coarser than ready simulation}

Once we have a clear picture of the semantics that are coarser than ready
simulation, it is time to consider the rest of the semantics in the spectrum.
Let us start with the possible futures and the impossible futures semantics.
Recall that we have shown that they can be described by $\textit{LGO}_T$
observations so that they are defined by $\leq_T^{lf{\supseteq}}$ and 
$\leq_T^{lf}$, respectively.

We introduce the $T$-versions of our $(\textit{ND\/}^Z)$ axioms: all of
them are instances of our general axiom for the reduction of non-determinism
and therefore are defined by the adequate constraint $M^T_Z(x,y,w)$.
As expected, they are obtained by substituting every occurrence of $I$
in $M_Z(x,y,w)$ by the observer $T$ defining the traces of processes.

\begin{defi}
The constraints $M^T_Z$ that characterize the semantics coarser than 
$T$-simulation semantics are:
\[
\begin{array}{l@{\quad}l}
(\textit{T-ND\/}^F) & M^T_F(x,y,w) \iff \textrm{true}\\
(\textit{T-ND\/}^R)& M^T_R(x,y,w) \iff T(x)\supseteq T(y)\\
(\textit{T-ND\/}^\mathit{FT})& M^T_{FT}(x,y,w) \iff T(w) \subseteq T(y) \\
(\textit{T-ND}^\mathit{RT})& M^T_{RT}(x,y,w) \iff  T(x) = T(y) \ \textrm{and}\ 
 T(w)\subseteq T(y) 
\end{array}
\]
\end{defi}

As indicated in Section~\ref{linear:sec} (Definition~\ref{df:PF_IF}), the semantics associated to the last two conditions do not appear in the ltbt 
spectrum and, as far as we know, they have not been previously studied nor  even
defined.
\comment{We ignore whether they will be of practical interest in the future but,
nonetheless, we have decided to include them for the sake of completeness.}

Using the same arguments as in Section~\ref{rnoef:sec1}, we can prove that 
$\leq^{l(Z)}_T$ satisfies the axiom $(\textit{T-ND\/}^Z)$ for $Z\in\{RT, FT\}$.
\begin{prop}
$M_Z^T(x,y,w)$ implies $T(a(x+y)) = T(ax + a(y+w))$ for 
$Z\in \{RT, FT\}$. However, this is not the case for $Z\in \{R, F\}$
\end{prop}
\begin{proof}
$M^T_{RT}$ implies $M^T_{FT}$, and therefore $T(y+w) = T(y)$, which leads to
$T(ax+a(y+w)) = T(a(x+y))$.
Neither $M_R^T$ nor $M_F^T$ refer to $w$ and therefore, in general,
$T(ax + a(y+w)) \neq T(a(x+y))$ in those cases.
\end{proof}

\vspace{0.2cm}
Note that when proving the correctness of the corresponding axiom 
$(\textit{ND\/}^Z)$ for $\leq_I^{l(Z)}$ we had $I(a(x+y)) = \{a\} =
I(ax+a(y+w))$ in all cases.
Now, $T(a(x+y)) = T(ax + a(y+w))$ only under the constraints 
corresponding to the finer semantics $\sqsubseteq_{FT}$ and $\sqsubseteq_{RT}$.
The properties of the prefixes appearing in all the terms in both
sides of the axiom $(\textit{ND})$ are not used anymore in the proofs in 
Section~\ref{rnoef:sec1}, so they can be transferred to the $T$-semantics, thus
proving the correctness of $(\textit{T-ND\/}^Z)$ for both
$\leq_T^{l(RT)}$ and $\leq_T^{l(FT)}$.

The introduction of the equational version $(\textit{ND}_\equiv)$
of the axiom $(\textit{ND})$ now becomes crucial in order to preserve
the generality of our unifying work.
We saw that under $(RS)$ these two axioms were equivalent.
However, when observing the set of traces $T(x)$ of any process,
instead of just the initial offer $I(x)$ we need to consider $T$-simulations,
that are 
constrained by the condition $T(x,y) \Longleftrightarrow T(x) = T(y)$;
under the corresponding axiom $(TS)$, things turn out to be different.

\begin{prop}\label{aux-traces:prop}
$T(a (x+y) + ax + a(y+w)) = T(ax)\cup T(ay)\cup T(aw) = 
T(ax+a(y+w))$.
\end{prop}

As a consequence, for $(\textit{T-ND}^Z_+)$ and $(\textit{T-ND}_\equiv)$
we can apply the same arguments used in Section~\ref{rnoef:sec1}
to show that $(\textit{ND\/}^Z)$ was satisfied by $\leq_I^{l(Z)}$.

\begin{prop}
For $Z\in\{RT, FT, R, F\}$, the preorder $\leq_T^{l(Z)}$ satisfies the axiom
$(\textit{T-ND\/}^Z_+)$ and also $(\textit{T-ND\/}^Z_\equiv)$.
\end{prop}
\begin{proof}
To show that $\leq_T^{l(Z)}$ satisfies $(\textit{T-ND\/}^Z_+)$ we just need to
apply Proposition~\ref{aux-traces:prop} and follow the line of thought in the
second bullet on page~\pageref{stronger-versions:pag}, substituting
the observer $T$ for $I$.
For the other axiom, from $T(a(x+y))\subseteq T(ax+ a(y+w))$ it follows that 
$\{(TS)\} \vdash ax + a(y+w) \preceq (ax+a(y+w)) + a(x+y)$.
\end{proof}

Notice that for $Z\in \{RT, FT\}$ we can also obtain the correctness of 
$(\textit{T-ND}^Z_\equiv)$ from that of $(\textit{T-ND\/}^Z)$ and vice versa,
as a consequence of the following fact.

\begin{prop}
The axiomatization
$\{\textrm{$B_1$--$B_4$}, (TS), (\textit{T-ND\/}^Z)\}$ is equivalent to the
axiomatization
$\{\textrm{$B_1$--$B_4$}, (TS), (\textit{T-ND}^Z_\equiv)\}$ for $Z\in \{RT,FT\}$.
\end{prop}
\begin{proof}
Let us first show that
$\{\textrm{$B_1$--$B_4$}, (TS), (\textit{T-ND\/}^Z)\}$ is equivalent to
$\{\textrm{$B_1$--$B_4$},$ $(TS), (\textit{T-ND\/}^Z_+)\}$.
This holds because $(\textit{T-ND\/}^Z)$ implies $(\textit{T-ND\/}^Z_+)$ and,
since $T(w)\subseteq T(y)$ implies $T(a(x+y)) = T(ax+a(y+w))$ and then
we have $\{\textrm{$B_1$--$B_4$}, (TS)\}\vdash a(x+y)\preceq
a(x+y) + (ax + a(y+w))$.

To prove 
$\{\textrm{$B_1$--$B_4$}, (TS), (\textit{T-ND\/}^Z_+)\}$ equivalent
to $\{\textrm{$B_1$--$B_4$},(TS),(\textit{T-ND\/}^Z_\equiv)\}$ we only need to
show that $\{\textrm{$B_1$--$B_4$}, (TS), (\textit{T-ND\/}^Z_+)\}\vdash
(M_Z^T(x,y,w) \Rightarrow ax + a(y+w)\preceq ax+ a(y+w)+ a(x+y))$, but we
have that for $Z\in \{RT, FT\}$, $M_Z^T(x,y,w)$ implies $T(w)\subseteq T(y)$,
so that $T(a(x+y)) =  T(ax+a(y+w))$ and therefore
$\{(TS)\}\vdash ax + a(y+w) \preceq (ax+a(y+w)) + a(x+y)$.
\end{proof}

The important fact about the obtained sets of correct axioms for the 
semantics $\leq_T^{l(Z)}$ is that, although our proofs of completeness for the
axiomatizations $\{\textrm{$B_1$--$B_4$}, (RS), (\textit{ND\/}^Z)\}$ considered
the inequational axioms $(\textit{ND\/}^Z)$, they are also valid for the
axiomatizations $\{\textrm{$B_1$--$B_4$}, (RS), (\textit{ND\/}^Z_\equiv)\}$.

The steps in the procedure that leads to the completeness of 
$\{\textrm{$B_1$--$B_4$}, (RS),$ $(\textit{ND\/}^Z)\}$ can be adapted by
substituting each reference to the observer $I$ by $T$, thus obtaining a
proof of the completeness of
$\{\textrm{$B_1$--$B_4$}, (RS), (\textit{T-ND\/}^Z_\equiv)\}$ for $\leq_T^{l(Z)}$.
However, the notion of head normal form for $N=I$ uses the fact that the 
summands $\textit{hnf\/}^Z(q,a,i,X_1)$ can be defined in 
terms of the offers $X_1\subseteq \calP(\textit{Act})$, which correspond to the 
values produced by the observer $I$.
For an arbitrary $N$, a more general definition of hnf's, valid for every 
observer, is needed.

\begin{defi}
For $p=\sum_{a\in X_0}\sum_{i\in I_a} ap_a^i$, its totally expanded $Z$-head normal
form $\tehnfx{p}$ is that given by:
\begin{iteMize}{$\bullet$}
\item For $a\in X_0$, $i\in I_a$, and $K_a\subseteq I_a$ we consider a
decomposition $p_a^k = p_a^{k_1} + p_a^{k_2}$ such that
$M_Z^N(p_a^i, p_a^{k_1}, p_a^{k_2})$. 
Then, 
$\textit{tehnf\/}^Z_N(p,a,i,\langle (p_a^{k_1},p_a^{k_2})\rangle_{k\in K_a})=
a(p_a^i +\sum_{k\in K_a}p_a^{k_1})$.
\item $\tehnfx{p} = \sum \textit{tehnf\/}^Z_N(p,a,i,\langle (p_a^{k_1},p_a^{k_2})
\rangle_{k\in K_a})$.
\end{iteMize}
\end{defi}

It is clear that for $K_a'\subseteq K_a$, or any decomposition 
$p_a^k= p_a^{k_3} + (p_a^{k_4} + p_a^{k_2})$ with $p_a^{k_1}=p_a^{k_3} +p_a^{k_4}$,
the corresponding $\textit{tehnf\/}^Z_N(\ldots)$ is a subterm of
$\textit{tehnf\/}^Z_N(p,a,i,$ $\langle p_a^{k_1},p_a^{k_2}\rangle_{k\in K_a})$ and
thus contributes nothing to the expanded normal form.
This is the reason why we preferred the more compact definition of 
$\hnfx{p}$ for semantics coarser than ready simulation.

\begin{thm}
For $Z\in\{RT, FT, R, F\}$,  
$\{\textrm{$B_1$--$B_4$}, (TS), (\textit{T-ND\/}^Z_\equiv)\}\vdash p\preceq q$
if and only if $p\leq_T^{l(Z)} q$.
\end{thm}

\section{Logical characterization of semantics}
\label{sec:LogicalCharacterization}
The third and a very natural alternative to associate a semantics to processes lies in the logical
framework.  
This is indeed quite a natural way to do it.
We have a language to express properties of processes and a way to check
whether a process satisfies a formula of the language:
then, two processes are equivalent with respect to this semantics if and only
if they satisfy the same set of formulas.  
In fact, the semantics can also be defined in terms of the induced preorder
that indicates whether a process satisfies more formulas than another one.

Each subset $\mathcal{L}$ of $\mathcal{L}_{HM}$ induces a semantics as stated in the following definition.
\begin{defi}
Any subset $\mathcal{L}$ of $\mathcal{L}_{HM}$ induces a logical semantics for processes, given by the preorder $\sqsubseteq_{\mathcal{L}}$: $p\sqsubseteq_{\mathcal{L}} q$ whenever for all $\varphi \in \mathcal{L}$, if $p\models \varphi$ then $q\models \varphi$. We say that $\mathcal{L}$ and $\mathcal{L}^{\prime}$ are equivalent, and we write $\mathcal{L} \sim \mathcal{L}^{\prime}$, if they induce the same semantics, that is ${\sqsubseteq_{\mathcal{L}}}={\sqsubseteq_{\mathcal{L}^{\prime}}}$.
\end{defi}
Let us start with a look at Table~\ref{logic_table}, which contains the logical
characterization of each of the semantics in van Glabbeek's spectrum. 
$\mathcal{L}_{Z}$ with $Z\in \{T, CT, F, FT, R, RT, PF, S, CS, RS,$ $2S, PW,
B\}$ denotes each of the logics; the dots indicate the clauses needed to obtain the corresponding languages; and the boxes marked with $\mathcal{\nu}$ correspond to rules that could be added to $\mathcal{L}_{Z}$, but would only introduce redundant formulas. 
The following constructs, which appear in the table but are not in
$\mathcal{L}_{HM}$, can be obtained as syntactic sugar: 
$$\widetilde{X}:=\bigwedge_{a\in X} \neg a \top \hspace{1.2cm} \widetilde{X}\varphi^{\prime}:=\widetilde{X} \wedge \varphi^{\prime} \hspace{1.2cm} 0:=\widetilde{Act}$$
$$\varphi_1 \wedge \varphi_2:=\bigwedge_{i\in \{1,2\}}\varphi_i \hspace{0.65cm} X:=\bigwedge_{a\in X}a\top \wedge \bigwedge_{a\not \in X} \neg a \top \hspace{0.7cm}X\varphi^{\prime}:=X \wedge \varphi^{\prime} \hspace{0.65cm} \widetilde{a}:=\neg a \top$$

\begin{table}[t]
\begin{center}
\scalebox{0.85}{\footnotesize
\begin{tabular}{|c|c|c|c|c|c|c|c|c|c|c|c|c||c|}
\hline
\backslashbox{Formulas}{Semantics ($\mathcal{Z}$)}& T & S & CT & CS & F & FT & R & RT & PW & RS & PF & 2S & B\\
\hline
$\top  \in \mathcal{L}_{\mathcal{Z}}$ & $\bullet$ & $\nu$ & $\bullet$ & $\nu$ & $\bullet$ & $\bullet$ & $\bullet$ & $\bullet$ & $\nu$ & $\nu$ & $\nu$ & $\nu$ & $\nu$\\
\hline
$\textbf{0}\in \mathcal{L}_{\mathcal{Z}}$ & & & $\bullet$ & $\bullet$ & $\nu$ & $\nu$ & $\nu$ & $\nu$ &$\nu$ & $\nu$ & $\nu$ & $\nu$ & $\nu$ \\
\hline
$\varphi \in \mathcal{L}_{\mathcal{Z}}, \hspace{0.075cm} a \in Act \Rightarrow$ &\multirow{2}{*}{$\bullet$} &\multirow{2}{*}{$\bullet$} &\multirow{2}{*}{$\bullet$} &\multirow{2}{*}{$\bullet$} &\multirow{2}{*}{$\bullet$} &\multirow{2}{*}{$\bullet$} &\multirow{2}{*}{$\bullet$} &\multirow{2}{*}{$\bullet$} &\multirow{2}{*}{$\nu$} &\multirow{2}{*}{$\bullet$} &\multirow{2}{*}{$\bullet$} & \multirow{2}{*}{$\bullet$} & \multirow{2}{*}{$\bullet$}\\
$a \varphi \in \mathcal{L}_{\mathcal{Z}}$& & & & & & & & & & & & &\\
\hline
$X \subseteq Act \Rightarrow$ & & & & & \multirow{2}{*}{$\bullet$} & \multirow{2}{*}{$\nu$} & \multirow{2}{*}{$\nu$} & \multirow{2}{*}{$\nu$} & \multirow{2}{*}{$\nu$} &\multirow{2}{*}{$\nu$} &\multirow{2}{*}{$\nu$} &\multirow{2}{*}{$\nu$} &\multirow{2}{*}{$\nu$} \\
$\widetilde{X} \in \mathcal{L}_{\mathcal{Z}}$ & & & & & & & & & & & & & \\
\hline
$X \subseteq Act \Rightarrow$ & & & & & & & \multirow{2}{*}{$\bullet$} &  \multirow{2}{*}{$\nu$} & \multirow{2}{*}{$\bullet$} &\multirow{2}{*}{$\bullet$}& \multirow{2}{*}{$\nu$} &\multirow{2}{*}{$\nu$} & \multirow{2}{*}{$\nu$}\\
$X \in \mathcal{L}_{\mathcal{Z}}$ & & & & & & & & & & & & &\\
\hline
$\varphi \in \mathcal{L}_{\mathcal{Z}}, \hspace{0.075cm} X \subseteq Act \Rightarrow$ & & & & & & \multirow{2}{*}{$\bullet$} & & \multirow{2}{*}{$\nu$} & \multirow{2}{*}{$\nu$}& \multirow{2}{*}{$\nu$} & & \multirow{2}{*}{$\nu$} & \multirow{2}{*}{$\nu$} \\
$\widetilde{X}\varphi \in \mathcal{L}_{\mathcal{Z}}$ & & & & & & & & & & & & &\\
\hline
$\varphi \in \mathcal{L}_{\mathcal{Z}}, \hspace{0.075cm} X \subseteq Act \Rightarrow$ & & & & & & & & \multirow{2}{*}{$\bullet$}  &\multirow{2}{*}{$\nu$} &\multirow{2}{*}{$\nu$} & &\multirow{2}{*}{$\nu$} & \multirow{2}{*}{$\nu$} \\
$X\varphi \in \mathcal{L}_{\mathcal{Z}}$ & & & & & & & & & & & & &\\
\hline
$\varphi_i \in \mathcal{L}_{\mathcal{Z}} \hspace{0.075cm} \forall i \in I \Rightarrow$ & & \multirow{2}{*}{$\bullet$} & &\multirow{2}{*}{$\bullet$} & & & & & &\multirow{2}{*}{$\bullet$} & & \multirow{2}{*}{$\bullet$} & \multirow{2}{*}{$\bullet$}\\
$\bigwedge_{i \in I} \varphi_i \in \mathcal{L}_{\mathcal{Z}}$ & & & & & & & & & & & & &\\
\hline
$X\subseteq Act, \hspace{0.075cm} \varphi_a \in \mathcal{L}_{PW} \hspace{0.075cm} \forall a \in X \Rightarrow$ & & & & & & & & & \multirow{2}{*}{$\bullet$}  &\multirow{2}{*}{$\nu$} & & \multirow{2}{*}{$\nu$} & \multirow{2}{*}{$\nu$}\\
$\bigwedge_{a \in X} a\varphi_a \in \mathcal{L}_{\mathcal{Z}}$ & & & & & & & & & & & & &\\
\hline
$\varphi_i, \varphi_j \in \mathcal{L}_{T} \hspace{0.075cm} \forall i \in I \hspace{0.05cm} \forall j \in J \Rightarrow$ & & & & & & & & & & &\multirow{2}{*}{$\bullet$} & \multirow{2}{*}{$\nu$}  & \multirow{2}{*}{$\nu$}\\
$\bigwedge_{i \in I} \varphi_i \wedge \bigwedge_{j \in J} \neg \varphi_j \in \mathcal{L}_{\mathcal{Z}}$ & & & & & & & & & & & & &\\
\hline
$\varphi \in \mathcal{L}_{S} \Rightarrow$ & & & & & & & & & & & & \multirow{2}{*}{$\bullet$} & \multirow{2}{*}{$\nu$}\\
$\neg \varphi \in \mathcal{L}_{\mathcal{Z}}$ & & & & & & & & & & & & &\\
\hline
$\varphi \in \mathcal{L}_{\mathcal{Z}} \Rightarrow$ & & & & & & & & & & & & &\multirow{2}{*}{$\bullet$}\\
$\neg \varphi \in \mathcal{L}_{\mathcal{Z}}$ & & & & & & & & & & & & &\\
\hline
\end{tabular}}

\vspace{1ex}
\caption{Van Glabbeek's logical characterizations for the semantics in the ltbt spectrum.} \label{logic_table}
\end{center}
\end{table}

Disjunction does not appear in $\mathcal{L}_{HM}$ and therefore neither in any of the logics $\mathcal{L}_{Z}$ characterizing the semantics in the ltbt spectrum. It is probably folklore that it can be added in all cases without affecting the expressive power of each of these logics, but since we have not found a clear statement in this direction in any of our references, next we establish the result and comment on its proof. 

\begin{prop}\label{disyuncion_teo}
Let us define $\mathcal{L}_Z^{\vee}$, with $Z\in \{T, CT, F, FT, R, RT, PF, S, CS,$ $RS, 2S, PW, B\}$, by adding the clause $\bigvee_{i\in I}\varphi_i \in \mathcal{L}_Z^{\vee}$ if $\varphi_i \in \mathcal{L}_Z^{\vee}$ for all $i \in I$ to the clauses that define $\mathcal{L}_Z$, replacing $\mathcal{L}_Z$ by $\mathcal{L}_Z^{\vee}$ in the other clauses, and making $p \models \bigvee \varphi_i$ iff there exists $i \in I$ with $p\models \sigma_i$. 
Then, $\mathcal{L}_Z^{\vee}\sim\mathcal{L}_Z$.
\end{prop}
\begin{proof}
It is interesting to observe that even if the result is valid for all the semantics, the reason behind is not the same as for bisimulation. In that case, we only need to apply the De Morgan laws to get the ``definition'' of $\vee$ as a combination of $\neg$ and  $\wedge$. However, for the rest of the semantics we do not have negation as ``constructor'', but $\vee$ distributes over $\wedge$ and the prefix operator (that is $\bigvee a \varphi_i= a \bigvee \varphi_i$), while negation is never applied to a formula $\varphi^{\prime} \in \mathcal{L}_Z^{\vee}$. Therefore, by floating to the top any $\vee$, using those distribution laws, a formula in $\mathcal{L}_Z^{\vee}$ becomes equivalent to a disjunction of formulas within the corresponding language $\mathcal{L}_Z$, and the equivalence of both logics follows.
\end{proof}

As we will see in this section, each of our logics is defined by a set of rules
and, as usual, only the formulas that can be obtained by finite application of
these rules belong to the logics. One important feature of our approach is that
instead of focusing on small sets of formulas characterizing each of the
semantics, we somewhat follow the opposite approach by including all the
formulas, from a certain family, that are preserved by each of the
semantics. This choice has many interesting side effects. In particular, we
will not need to look for adequate formulas reflecting the characteristics of
each of the semantics, but instead pick up from our ``repository'' of possible
formulas those that are preserved by the current semantics. Thus, we
characterize each of the semantics by means of the formulas that ``see'' the
kind of observations that define it. As a consequence, we know whether a
semantics is coarser than another by checking whether the logic characterizing
the former is included in the logic characterizing the latter. Moreover, by
using a larger logic we may find a formula expressing some property that
is preserved by the corresponding semantics, while if we settle on a smaller
logic we might need a collection of formulas to express a simple property.

Formally speaking, for each semantics defined by a preorder $\prec$ we have a
language $\mathcal{L} \subseteq \mathcal{L}_{HM}$ characterizing it: $\varphi
\in \mathcal{L}$ iff $((p \prec q \wedge p \models \varphi)$ implies $q \models
\varphi)$.  
However, it is not easy (nor specially illustrative) to capture the whole set
of formulas characterizing the semantics. 
Instead, we will consider sufficiently large families defined in a simple way
that provide natural characterizations of the different semantics and show the
relationship between them so that, as stated above, whenever a semantics is
finer than another, the logic characterizing the first will contain that for
the latter. 

As will become clear when we introduce our new logical characterizations,
Table~\ref{logic_table} readily presents the features that allow us to classify
the semantics in the spectrum in four categories: 
\begin{iteMize}{$\bullet$}
\item Bisimulation semantics, characterized by HML, that is closed under negation ($\neg$), so that the preorder defined is an equivalence (the bisimulation). The remaining semantics are defined by non-trivial preorders, i.e., the preorders are not equivalences and their logical characterizations are, of course, not closed under negation.
\item Simulation semantics (S, CS, RS, \ldots), characterized by branching observations, which will be reflected by the unrestricted use of the operator $\bigwedge$ in the formulas.
\item Linear semantics (T, F, R, \ldots), characterized by linear observations. We will get them by severely restricting the use of $\bigwedge$ and the use of the negation.
\item Deterministic branching semantics, corresponding to an intermediate class between branching and linear semantics, where determinism appears restricting the use of the operator $\bigwedge$ in combination with the prefix operator. The only semantics in this class in the classical spectrum is PW.
\end{iteMize} 
As already happened in Sections~\ref{sec:EquationalSemantics} and
\ref{observational-sem-sec}, our unified logical semantics will 
provide an enlarged spectrum---Figure~\ref{fig:extended-ltbts}.
In particular, we will show the logical characterization of revivals semantics,
introduced by Roscoe in \cite{Roscoe09} and already axiomatized in
\cite{FGP09mfps}.

\subsection{A new logical characterization of the most popular semantics} \label{logical_characterizations_popular_semantics}

Again, we start with the best known classical semantics, that is, those at the
layer of ready simulation in the spectrum. All of them use in some way the set
of formulas $\mathcal{L}_I=\{a\top \mid a \in Act\}$ that characterizes the
initial offers of a process. In
Section~\ref{logical_characterizations_coarsest_semantics} we will present the
logics for the rest of the semantics in a unified way, remarking how they are
obtained similarly to those in this section but working from the set
$\mathcal{L}_{N}$ of formulas associated to the corresponding constraint $N$.  

We will prove the equivalence between each of our logics and the corresponding
logical characterization defined by van Glabbeek, thus checking that our new
logical characterizations are indeed correct. But one of the intended goals of
our unification was to obtain direct and natural proofs. This will be
illustrated in Section~\ref{logic_observational_framework} by showing the
equivalence between each of our logical semantics and the corresponding
observational semantics of Section~\ref{observational-sem-sec}.
This will provide a new, single proof of their correctness without having to
resort to the characterizations defined by van Glabbeek.

\begin{defi}\label{def_logica}
\emph{Ready simulation semantics}. We define the set of formulas $\mathcal{L}^{\prime}_{RS}$ for ready simulation semantics by: \bran{\conjsimold{RS}{I}}{RS}
\emph{Ready trace semantics}. We define the set of formulas $\mathcal{L}^{\prime}_{RT}$ for ready trace semantics by: \lin{\conjyformparticular{RT}{I}}{RT}
\emph{Failure trace semantics}. We define the set of formulas $\mathcal{L}^{\prime}_{FT}$ for failure trace semantics by: \linfparticular{I}{FT}
\emph{Readiness semantics}. We define the set of formulas $\mathcal{L}^{\prime}_{R}$ for readiness semantics by: \lincparticular{I}{R} 
\emph{Failures semantics}. We define the set of formulas $\mathcal{L}^{\prime}_{F}$ for failures semantics by: \linfcparticular{I}{F}
\end{defi}

It is immediate that $\mathcal{L}^{\prime}_{RS}\subseteq \mathcal{L}_B$ and
hence ready simulation semantics is coarser than bisimulation equivalence. We
also have $\mathcal{L}^{\prime}_{F}\subseteq \mathcal{L}^{\prime}_{R}$,
$\mathcal{L}^{\prime}_{F}\subseteq \mathcal{L}^{\prime}_{FT}$,
$\mathcal{L}^{\prime}_{R}\subseteq \mathcal{L}^{\prime}_{RT}$,
$\mathcal{L}^{\prime}_{FT}\subseteq \mathcal{L}^{\prime}_{RT}$, and
$\mathcal{L}^{\prime}_{RT}\subseteq \mathcal{L}^{\prime}_{RS}$, which can be
interpreted in a similar way. Let us now focus our attention on the third rule
of the definition of $\mathcal{L}^{\prime}_{RS}$: the unrestricted use of
conjunction corresponds to the branching nature of the semantics. Moreover, the
two first rules allow to fix the set of offers of a process as $I$-simulations
impose. By contrast, the linear semantics only allow the use of conjunction to
join those simple formulas that fix the set of offers along a computation (in
the case of the readies-based semantics), or their over-approximations
(obtained by means of the negated formulas $\neg a\top$, in the case of the
failures-based semantics). Finally, notice how these simple formulas can only
be checked at the corresponding final state, for the two simpler coarser
semantics. 

Now, for $Z \in \{RS, RT, FT, R, F\}$, each of the logics $\mathcal{L}^{\prime}_{Z}$ is a superset of the corresponding logic $\mathcal{L}_{Z}$ defined in Table \ref{logic_table}. To be precise, for $FT$ and $F$ we need to remove the syntactic sugar used by van Glabbeek as stated below.

\begin{rem}
We have used in Section \ref{linear:sec} $X^c$ to denote the complementary of a set, because previously in Definition \ref{closures:def} we used the classic over line notation to refer to closures of sets $\calT\subseteq\textit{LGO}_N$. However, since we will not need those closure operators anymore we prefer to used the classic notation referring the complement of a set $X$ by $\overline{X}$.
\end{rem}

\begin{prop}
\hfill
\begin{enumerate}[\em(1)]
\item $\mathcal{L}_{RS}\varsubsetneq \mathcal{L}^{\prime}_{RS}$.
\item $\mathcal{L}_{RT}\subsetneq \mathcal{L}^{\prime}_{RT}$. 
\item $\mathcal{L}^{\prime}_{FT}\supseteq$ desugared$(\mathcal{L}_{FT})$, where the desugaring function removes the syntactic sugar used in $\mathcal{L}_{FT}$.
\item $\mathcal{L}_{R}\varsubsetneq \mathcal{L}^{\prime}_{R}$.
\item $\mathcal{L}^{\prime}_{F}\supseteq$ desugared($\mathcal{L}_{F}$), where the desugaring function removes the syntactic sugar used in $\mathcal{L}_F$.
\end{enumerate}
\end{prop}
\begin{proof}
Recall the definition of $\mathcal{L}_Z$ in Table~\ref{logic_table}.
\begin{enumerate}[(1)]
\item To prove that $\mathcal{L}_{RS}\subseteq \mathcal{L}^{\prime}_{RS}$, it is sufficient to show that each formula $\varphi_X= \bigwedge_{a\in X} a\top \wedge \bigwedge_{b \notin X} \neg b \top$ corresponding to $X\subseteq Act$ belongs to $\mathcal{L}^{\prime}_{RS}$. Both $a\top $ and $\neg b\top$ are in $\mathcal{L}^{\prime}_{RS}$ and the combination of these formulas with the operator $\wedge$ is also in the set $\mathcal{L}^{\prime}_{RS}$. For the inclusion to be proper, it is sufficient to notice that the formula $\neg b\top$ belongs to $\mathcal{L}^{\prime}_{RS}$ but not to the set $\mathcal{L}_{RS}$.
\item To prove that $\mathcal{L}_{RT}\subseteq \mathcal{L}^{\prime}_{RT}$ it is sufficient to show that for every $X \subseteq Act$ and any $\varphi \in \mathcal{L}_{RT}$, the formula $(\bigwedge_{a \in X} a\top \wedge \bigwedge_{b \notin X} \neg b\top) \wedge \varphi$ belongs to $\mathcal{L}^{\prime}_{RT}$. Note that $b \notin X$ is equivalent to $b \in \overline{X}$, so taking $X_1=X$ and $X_2=\overline{X}$ we have that the considered formula belongs to $\mathcal{L}^{\prime}_{RT}$. To prove that  $\mathcal{L}_{RT} \subset \mathcal{L}^{\prime}_{RT}$, it is sufficient to note that $(\neg b\top) \wedge \varphi$ belongs to $\mathcal{L}^{\prime}_{RT}\,$, by taking $X_1=\emptyset$ and $X_2=\{b\}$, but it does not belong to $\mathcal{L}_{RS}$.
\item In this case the result is trivial, since the definitions of $\mathcal{L}_{FT}$ and $\mathcal{L}_{FT}^{\prime}$ are almost the same, once the syntactic sugar is removed. The only difference is that $\perp \in \mathcal{L}_{FT}^{\prime}$, which obviously does not affect the inclusion.
\item To prove that $\mathcal{L}_{R}\subseteq \mathcal{L}^{\prime}_{R}$, it is sufficient to show that for every $X \subseteq Act$ the formula $\bigwedge_{a \in X} a\top \wedge \bigwedge_{b \notin X} \neg b\top$ belongs to $\mathcal{L}^{\prime}_{R}$. Note that the condition $b \notin X$ is equivalent to $b \in \overline{X}$, so taking $X_1=X$ and $X_2=\overline{X}$ we have that the considered formula belongs to $\mathcal{L}^{\prime}_{R}$. To check that $\mathcal{L}_{R} \varsubsetneq \mathcal{L}^{\prime}_{R}$, it is sufficient to note that the formula $\neg b\top$ belongs to $\mathcal{L}^{\prime}_{R}$ by taking $X_1=\emptyset$ and $X_2=\{b\}$, while it does not belong to $\mathcal{L}_{R}$.
\item Analogous to 3.\qedhere
\end{enumerate}    
\end{proof}

As stated earlier, in order to obtain more
natural characterizations, our logics typically contain large sets of
formulas. This is why in most cases our logics contain those proposed
by van Glabbeek. In order to prove the equivalence between ours and his, we
have to show that our additional formulas are in fact redundant and could be
safely removed. 
\begin{prop} \label{equiv_entre_logicas}
(1) $\mathcal{L}_{RS}\sim \mathcal{L}^{\prime}_{RS}$; (2) $\mathcal{L}_{RT}\sim \mathcal{L}^{\prime}_{RT}$;  (3) $\mathcal{L}_{FT}\sim \mathcal{L}^{\prime}_{FT}$; (4) $\mathcal{L}_{R}\sim \mathcal{L}^{\prime}_{R}$; and (5) $\mathcal{L}_{F}\sim \mathcal{L}^{\prime}_{F}$.
\end{prop}
\begin{proof}
\hfill
\begin{enumerate}[(1)]
\item[$(1)$] Any conjunction and negation of formulas in $\mathcal{L}_I$ can be obtained as the disjunction of the formulas $X$ describing all the ``compatible'' offers. These are those including the positive and negative information in the corresponding conjunction, i.e., $ a \top \sim \bigvee_{a \in X} X$; $\neg a \top \sim \bigvee_{a \notin X} X$.
Then, by applying Proposition~\ref{disyuncion_teo}, we obtain $ \mathcal{L}^{\prime}_{RS} \sim \mathcal{L}_{RS}$.
\item[$(2)$] We have shown that the formulas in $\mathcal{L}_{RT}$ are particular cases of the formulas in $\mathcal{L}^{\prime}_{RT}$: those that completely define the offers at the states along a computation (when we apply the second clause in the definition of $\mathcal{L}^{\prime}_{RT}$ with $X_2 = \overline{X_1}$).
In contrast, our more general formulas $(\bigwedge_{a\top \in X_1} a\top \wedge \bigwedge_{b\top \in X_2}\neg b\top) \wedge \varphi$, where $\varphi \in \mathcal{L}^{\prime}_{RT}$, could provide us with some partial information, combining both positive information $a\top \in X_1$ and negative information $b\top \in X_2$, which tells us that we are in an arbitrary state $X$ satisfying $X_1 \subseteq X \subseteq \overline{X_2}$. But we can replace these formulas by the disjunction of all the formulas describing any of these possible offers $X$. By repeating this procedure at each level of the formula, we finally obtain a disjunction of formulas in $\mathcal{L}_{RT}$. To conclude, it is enough to apply Proposition~\ref{disyuncion_teo}.
\item[$(3)$] We know $\perp = \neg \top=\bigvee_{i\in \emptyset}\varphi_i$, and applying Proposition~\ref{disyuncion_teo} we get the equivalence.
\item[$(4)$] Note that van Glabbeek allowed in $\mathcal{L}_{R}$ only ``normal form'' formulas from $\mathcal{L}^{\prime}_{R}$, which can give us information about the offers at the final state in a computation (when we apply the second clause in the definition of $\mathcal{L}^{\prime}_{R}$) or simply define these computations by means of the prefix operator (when we apply the third clause in the definition of $\mathcal{L}^{\prime}_{R}$).
However, our more general formulas $(\bigwedge_{a\top \in X_1} a\top \wedge \bigwedge_{b\top \in X_2}\neg b\top)$ can also provide us with some partial information about the final state, which could be both positive $a\top \in X_1$ and negative $b\top \in X_2$. In the (allowed) case $X_1 \bigcap X_2 \neq \emptyset$ we have that the formula is unsatisfiable. Otherwise, we are offering the actions $a$ corresponding to formulas $a\top$ in any $X\subseteq \mathcal{L}_{I}$ that satisfies $X_1 \subseteq X$ and  $X\subseteq \overline{X_2}$, and we can replace again the corresponding formula by a disjunction of formulas in $\mathcal{L}_{R}$.
\item[$(5)$] Analogous to 3.\qedhere
\end{enumerate}
\end{proof}

\noindent In the following, when we consider a logic $\mathcal{L}_{Z}$ and the index $Z$ refers to some concrete semantics, as is the case with $RS$, $RT$, $FT$, $R$, and $F$ above, by abuse of notation we will simply write $\sqsubseteq^{\prime}_{Z}$ instead of $\sqsubseteq_{\mathcal{L}^{\prime}_{Z}}$ for the preorder induced by the logic $\mathcal{L}^{\prime}_Z$.

\begin{thm}\label{teoequiv}
\hfill
\begin{enumerate}[\em(1)]
\item The logical semantics $\sqsubseteq^{\prime}_{RS}$ induced by the logic $\mathcal{L}^{\prime}_{RS}$ is equivalent to the observational branching semantics defined by $\leq_{I}^{b}$, generated by the set of branching general observations $BGO_I$.
\item For $Z\in \{F, FT, R, RT\}$, the logical semantics  $\sqsubseteq^{\prime}_{Z}$ induced by the logic $\mathcal{L}^{\prime}_{Z}$ is equivalent to the observational linear semantics $\leq_I^{l(Z)}$ in Definitions~\ref{orden_lN} and \ref{obs:def}.
\end{enumerate}
\end{thm}
\begin{proof}
It is a consequence of  Proposition~\ref{equiv_entre_logicas} and the results by van Glabbeek collected in Table~\ref{logic_table}, Theorem~\ref{denotational-main:thm}, and Proposition~\ref{failures-ftraces:prop}.
\begin{enumerate}[(1)]
\item We have already checked that our formulas are equivalent to van Glabbeek's: $\mathcal{L}^{\prime}_{RS} \sim \mathcal{L}_{RS}$. It is easy to show that once we have eliminated the unsatisfiable formulas in $\mathcal{L}^{\prime}_{RS}$ (those that simultaneously make two different offers, or perform an action that was not included in the corresponding offer) the remaining formulas in $\mathcal{L}^{\prime}_{RS}$ admit a normal form in the language $\mathcal{N(L_{RS})}$, which we define as follows:
\begin{iteMize}{$\bullet$}
\item if $X\subseteq Act$, $\{a_i\mid i\in I\} \subseteq X$, and $\varphi_i \in \mathcal{N(L_{RS})}$, then $(\bigwedge_{b\in X}b \top \wedge \bigwedge_{b \notin X} \neg b \top) \wedge \bigwedge_{i\in I}a_i \varphi_i \in \mathcal{N(L_{RS})}$;
\item if $\{a_i\mid i\in I\} \subseteq Act$ and $\varphi_i \in \mathcal{N(L_{RS})}$ then $\bigwedge_{i\in I}a_i \varphi_i \in \mathcal{N(L_{RS})}$.
\end{iteMize}

Within this set, consider the subset of formulas $\mathcal{CN(L_{RS})}$ which can be generated using the first clause in the above definition. We can establish an isomorphism between $\mathcal{CN(L_{RS})}$ and the set of possible branching general observations $BGO_I$. Moreover, it is easy to prove that if for every formula $\varphi \in \mathcal{CN(L_{RS})}$ we define $bgo_\varphi$ as the corresponding observation, then $\varphi \models p$ iff $bgo_\varphi \in BGO_I(p)$, from which it immediately follows that $\mathcal{CN(L_{RS})}$ characterizes the ready simulation semantics defined via $BGO_I$.

Now, to conclude the proof it is sufficient to show that $\mathcal{N(L_{RS})}$
and $\mathcal{CN(L_{RS})}$ are equivalent. Note that whenever we use the second
clause in the definition of $\mathcal{N(L_{RS})}$, we are ignoring the
possibility of specifying the offer $X$ at the state we are. As a consequence,
the offer could be any satisfying $\{a_i\mid i \in I \} \subseteq X$, for the
corresponding set $\{a_i \mid i \in I \}$. Then we can complete the associated
formula $\bigwedge_{i \in I} a_i\varphi_i$ by adding the disjunction
$\bigvee_{\{a_i / i\in I\} \subseteq X}(\bigwedge_{b \in X} b\top \wedge
\bigwedge_{b \notin X} \neg b \top)$. Floating all the disjunctions away we
obtain a disjunction of formulas in $\mathcal{N(L_{RS})}$, which ends the proof. 

\item \begin{iteMize}{$\bullet$}
\item If $Z=RT$, we know that $\mathcal{L}^{\prime}_{RT} \sim \mathcal{L}_{RT}$. It is easy to show that eliminating all the unsatisfiable formulas (those that simultaneously offer two different sets of actions,  or perform an action $a$ that is not included in the corresponding offer $X$) the rest of the formulas in $\mathcal{L}^{\prime}_{RT}$ admit a normal form in the language $\mathcal{N(L_{RT})}$, which we define as follows:
\begin{iteMize}{$-$}
\item if $X\subseteq Act$ then $(\bigwedge_{b\in X}b \top \wedge \bigwedge_{b \notin X} \neg b \top) \in \mathcal{N(L_{RT})}$;
\item if $X\subseteq Act$, $a\in X$, and $\varphi \in \mathcal{N(L_{RT})}$ then $(\bigwedge_{b\in X}b \top \wedge \bigwedge_{b \notin X} \neg b \top) \wedge a\varphi \in \mathcal{N(L_{RT})}$;
\item $\top \in \mathcal{N(L_{RT})}$;
\item if $a \in Act$ and $\varphi \in \mathcal{N(L_{RT})}$ then $a\varphi \in \mathcal{N(L_{RT})}$.
\end{iteMize}

As we did for the case of ready simulation, we could define the corresponding language of complete formulas $\mathcal{CN(L_{RT})}$. The formulas in $\mathcal{L}^{\prime}_{RT}$ that we obtained in the proof of Proposition~\ref{equiv_entre_logicas}, for the case of $RT$, are indeed in $\mathcal{CN(L_{RT})}$ because any subformula gives us some partial information about the offers at the corresponding state, which in the worst case could be empty. Therefore, when we translate this information into the language $\mathcal{L}^{\prime}_{RT}$ we obtain a disjunction between complete formulas in $\mathcal{CN(L_{RT})}$. We can easily establish the isomorphism between $\mathcal{CN(L_{RT})}$ and the domain $LGO_I$, and then prove that for every formula $\varphi \in \mathcal{CN(L_{RT})}$, if we define $lgo_\varphi$ as the corresponding observation, we have $\varphi \models p$ iff $lgo_\varphi \in LGO_I(p)$. From here it follows that $\mathcal{CN(L_{RT})}$ characterizes the ready simulation semantics defined via $LGO_I$.
To conclude the proof we need to show that $\mathcal{N(L_{RT})}$ and
$\mathcal{CN(L_{RT})}$ are equivalent, which is analogous to $\mathcal{N(L_{RS})}$ and $\mathcal{CN(L_{RS})}$ above.

\item $Z=FT$. $(\Rightarrow)$ 
Let $p$ and $q$ be such that $p\sqsubseteq^{\prime}_{FT}q$: we will show that $p\leq_I^{l\supseteq}q$.
Given an observation $X_0a_1X_1 \ldots a_nX_n \in LGO_I(p)$, we have a failure
trace $\overline{X_0}a_1\overline{X_1} \ldots a_n\overline{X_n}$ for the process $p$. Now, we consider the formulas $\varphi_n= \bigwedge_{a \in \overline{X}} \neg a \top$, $\varphi_i= \bigwedge_{a \in \overline{X_i}}\neg a \top \wedge a_{i+1}\varphi_{i+1}$ with $i \in 0..n-1$, and we have that $p\models \varphi_0$. Therefore $q\models \varphi_0$, which means that $\overline{X_0}a_1\overline{X_1} \ldots a_n\overline{X_n}$ is a failure trace of $q$. Then, there is some $Y_0a_1Y_2 \ldots a_nY_n \in LGO_I(p)$ with $Y_i \bigcap \overline{X_i} = \emptyset$ for all $i=0..n$ or, equivalently, $X_i\supseteq Y_i$ for all $i=0..n$. As a result, $LGO_I(p) \leq_I^{l\supseteq} LGO_I(q)$, which means $p \leq_I^{l\supseteq} q$.

$(\Leftarrow)$ 
Let us suppose that for all $X_0a_1X_1\ldots a_nX_n \in LGO_I(p)$ there exists $Y_0a_1Y_1\ldots$ $a_nY_n \in LGO_I(q)$ such that $X_i\supseteq Y_i$ for all $i=0..n$; we want to show that if $p\models \varphi$ then $q\models \varphi$, for all $\varphi \in \mathcal{L}^{\prime}_{FT}$.
If $p\models \varphi$, we can decompose $\varphi$ by means of a sequence of formulas, taking $\varphi= \varphi_n$, $\varphi_i = \bigwedge_{a\in X_2^{i}} \neg a \top \wedge a_{i}\varphi_{i-1}$ for $i\in 1..n$ and $\varphi_0= \bigwedge_{a \in X_2^{0}} \neg a \top$ . Therefore, $X_na_nX_{n-1}\ldots a_1X_0$ is a failure trace for the process $p$, so there exists $Z_na_nZ_{n-1}\ldots a_1Z_0\in LGO_I(p)$ with $Z_i \bigcap X_i=\emptyset$, and using that  $p \leq_I^{l\supseteq} q$, there exists some $Y_na_nY_{n-1} \ldots a_1Y_0 \in LGO_I(q)$ with $Y_i\subseteq Z_i$, so that $Y_i \bigcap X_i=\emptyset$ and then we get $q\models \varphi_n$.

\item If $Z=R$, using the result in the proof of Proposition~\ref{equiv_entre_logicas} for the case of R it is enough to show the result for the set of ``normal form'' formulas $\mathcal{N(L_{R})}$ defined by:
\begin{iteMize}{$-$}
\item if $X\subseteq Act$ then $(\bigwedge_{b\in X}b \top \wedge \bigwedge_{b \notin X} \neg b \top) \in \mathcal{N(L_{R})}$;
\item $\top \in \mathcal{N(L_{R})}$;
\item $a \in Act$ and $\varphi \in \mathcal{N(L_{R})}$ then $\varphi \in \mathcal{N(L_{R})}$.
\end{iteMize}
$(\Rightarrow)$ 
Let $p$ and $q$ be such that $p\sqsubseteq^{\prime}_{R}q$: we will show $p\leq_I^{lf}q$.
Given an observation $X_0a_1X_1 \ldots a_nX_n \in LGO_I(p)$, it corresponds to the readiness information $a_1 \ldots a_nX_n$ of $p$. Now, we consider the formulas $\varphi_n= \bigwedge_{a \in X} a \top \wedge \bigwedge_{a \notin X} \neg a \top$; $\varphi_{i-1}= a_{i}\varphi_i$ with $i \in 1 \ldots n-1$, and we have that $p\models \varphi_0$. Therefore $q\models \varphi_0$, and $a_1 \ldots a_nX_n$ is a readiness information of $q$ and, as a consequence, there is an observation $Y_0a_1Y_2 \ldots a_nY_n \in LGO_I(q)$ with $Y_n=X_n$, proving $p\leq_I^{lf}q$.

$(\Leftarrow)$ 
Let us suppose that for all $X_0a_1X_1\ldots a_nX_n \in LGO_I(p)$ there exists some $Y_0a_1Y_1\ldots a_nY_n \in LGO_I(q)$ such that $X_n=Y_n$. We want to show that if $p\models \varphi$ then $q\models \varphi$ for all $\varphi \in \mathcal{CN(L_{R})}$.
If $p\models \varphi$, we can decompose $\varphi$ taking $\varphi= \varphi_n$, $\varphi_{i} = a_{i}\varphi_{i-1}$, for all $i \in 1..n$, and $\varphi_0= \bigwedge_{a \in X_0} a \top \wedge \bigwedge_{a \notin X_0} \neg a \top$. Then we have that $a_na_{n-1}\ldots a_1X_0$ is a readiness information of $p$, so there exists some $Z_na_nZ_{n-1}\ldots a_1X_0\in LGO_I(p)$, and some $Y_na_nY_{n-1} \ldots a_1Y_0 \in LGO_I(q)$ with $Y_0=X_0$, from which we conclude that $q\models \varphi_n$.

\item $Z=F$. $(\Rightarrow)$
Let $p$ and $q$ be such that $p\sqsubseteq^{\prime}_{F}q$: we will show $p\leq_I^{lf\supseteq}q$.
Given an observation $X_0a_1X_1 \ldots a_nX_n \in LGO_I(p)$, it generates a (maximal) failure $a_1 \ldots a_n\overline{X_n}$ of the process $p$. Now, we consider the formulas $\varphi_n= \bigwedge_{a \in \overline{X}} \neg a \top$; $\varphi_{i+1}=a_{i+1}\varphi_i$ with $i \in 0..n-1$, and we have that $p\models \varphi_0$. Therefore, $q\models \varphi_0$, so $a_1 \ldots a_n\overline{X_n}$ is a failure information of $q$, and there is some  $Y_0a_1Y_2 \ldots a_nY_n \in LGO_I(q)$ with $Y_n \bigcap \overline{X_n} = \emptyset$, or equivalently $X_n\supseteq Y_n$, proving that $p\leq_I^{lf\supseteq}q$.

$(\Leftarrow)$
Let us suppose that for all $X_0a_1X_1\ldots a_nX_n \in LGO_I(p)$ there exists some $Y_0a_1Y_1\ldots a_nY_n \in LGO_I(q)$ such that $X_n\supseteq Y_n$. We want to show that if $p\models \varphi$ then $q\models \varphi$ for all $\varphi \in \mathcal{L}^{\prime}_{F}$.
If $p\models \varphi$, we can decompose $\varphi$ taking $\varphi= \varphi_n$, $\varphi_{i} = a_{i}\varphi_{i-1}$, with $i \in 1..n$, and $\varphi_0= \bigwedge_{a \in X_0} \neg a \top$. From $p\models \varphi$ we infer that $a_na_{n-1}\ldots a_1X_0$ is a failure information of the process \textit{p}, so there exists $Z_na_nZ_{n-1}\ldots a_1Z_0\in LGO_I(p)$ with $Z_0 \bigcap X_0=\emptyset$, and then there is some $Y_na_nY_{n-1} \ldots a_1Y_0 \in LGO_I(q)$ with $Y_n\subseteq Z_n$, so that $Y_n \bigcap X_n=\emptyset$, obtaining $q\models \varphi_n$.
\end{iteMize}
\end{enumerate}
\end{proof}

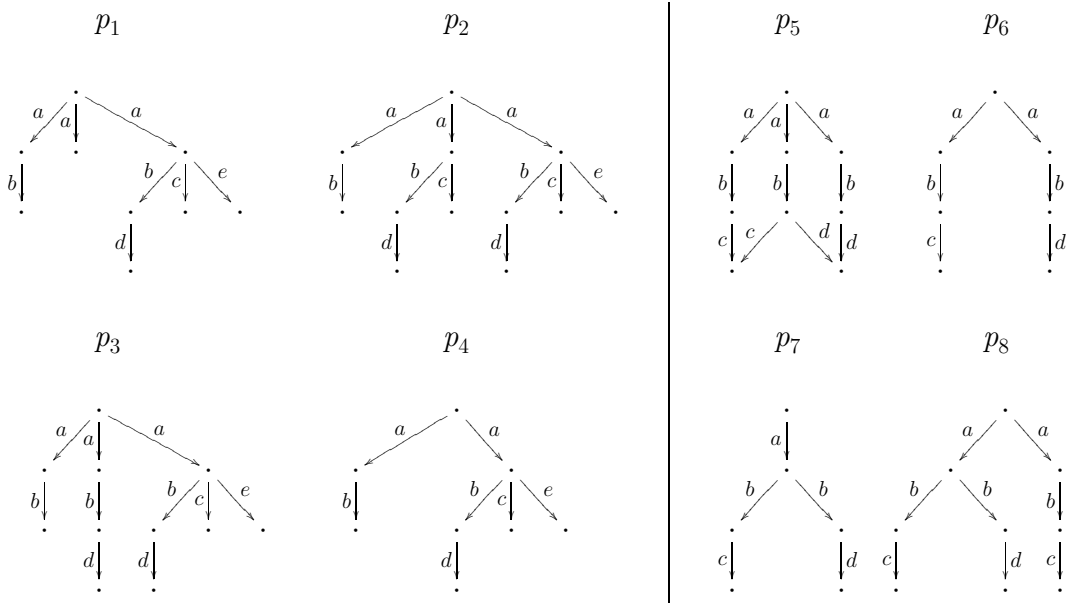
\begin{figure}[th]
\scalebox{0.57}{\huge
$\begin{array}{c c@{\hspace{3.25ex}\vline\hspace{3.25ex}} c c}
\textit{p}_1&\textit{p}_2 &\textit{p}_5 &\textit{p}_6\\ &&&\\
\xymatrix{
& &\cdot\ar[dl]_a \ar[d]_a\ar[drr]^a& & &\\
&\cdot\ar[d]_b  &\cdot & & \cdot \ar[dl]_b \ar[d]_c\ar[dr]^e &\\
& \cdot & &\cdot \ar[d]_d &\cdot &\cdot\\
& & & \cdot & &}
&
\xymatrix{
& & &\cdot\ar[dll]_a \ar[d]_a\ar[drr]^a& & &\\
&\cdot\ar[d]_b & &\cdot \ar[dl]_b \ar[d]_c & & \cdot \ar[dl]_b \ar[d]_c\ar[dr]^e &\\
&\cdot & \cdot \ar[d]_d&\cdot &\cdot \ar[d]_d &\cdot &\cdot\\
& & \cdot & & \cdot & &}
&
\xymatrix{
&\cdot \ar[dl]_a\ar[d]_a\ar[dr]^a&\\ 
\cdot \ar[d]_b& \cdot \ar[d]_b& \cdot \ar[d]^b\\ 
\cdot \ar[d]_c&\cdot \ar[dl]_c \ar[dr]^d& \cdot \ar[d]^d\\ 
\cdot &&\cdot }
&
\xymatrix{
&\cdot\ar[dl]_a\ar[dr]^a &\\
\cdot\ar[d]_b& & \cdot\ar[d]^b\\
\cdot\ar[d]_c& & \cdot \ar[d]^d\\
\cdot&&\cdot}
\\
&&&
\\
\textit{p}_3&\textit{p}_4 &\textit{p}_7 &\textit{p}_8 \\&&& \\
\xymatrix{
& & &\cdot\ar[dl]_a \ar[d]_a\ar[drr]^a& & &\\
& &\cdot\ar[d]_b &\cdot \ar[d]_b & & \cdot \ar[dl]_b \ar[d]_c\ar[dr]^e &\\
& & \cdot & \cdot \ar[d]_d &\cdot \ar[d]_d &\cdot &\cdot\\
& & &\cdot & \cdot & &}
&
\xymatrix{
& &\cdot\ar[dll]_a \ar[dr]^a&&\\
\cdot\ar[d]_b& & & \cdot \ar[dl]_b \ar[d]_c\ar[dr]^e &\\
\cdot & & \cdot  \ar[d]_d&\cdot & \cdot \\
& &  \cdot  & &}
&
\xymatrix{
&\cdot\ar[d]_a&\\
&\cdot\ar[dl]_b \ar[dr]^b&\\
\cdot\ar[d]_c& & \cdot \ar[d]^d\\
\cdot&&\cdot}
&
\xymatrix{
& &\cdot\ar[dl]_a\ar[dr]^a& &\\
&\cdot\ar[dl]_b \ar[dr]^b& &\cdot \ar[d]_b&\\
\cdot\ar[d]_c& & \cdot \ar[d]^d & \cdot \ar[d]_c&\\
\cdot&&\cdot &\cdot &}
\end{array}$}
\caption{A simple example to show the strength of the different logics}\label{example1} 
\end{figure}

\begin{exa}
Figure~\ref{example1} shows a collection of  examples to illustrate the differences between the semantics in the \textit{RS} layer of the spectrum. All the following equivalences can be checked by taking any arbitrary formula from the logic defining each of the semantics.
For readability, we omit the last $\top$ in all subformulas. Besides, $\sim_{X}$ (resp. $\nsim_{X}$), where \emph{X} is a set of indexes, represents any $\sim_{Z}$ (resp. $\nsim_{Z}$), with  $Z \in\emph{X}$. 
\begin{iteMize}{$\bullet$}
\item $p_1 \not\sqsubseteq^{\prime}_{F} p_2$ and $p_1\not\sqsubseteq^{\prime}_{\{R,\;FT,\;RT,\;RS\}} p_2$ because $p_1\models a(\neg b \wedge \neg c)$, but $p_2$ does not satisfy it.
\item $p_2\sim_{F}p_3$, but $p_2\not\sqsubseteq^{\prime}_{\{R,\;FT\}}p_3$ and thus $p_2\not\sqsubseteq^{\prime}_{\{RT,\;RS\}}p_3$, since $p_2$ satisfies $a(\neg e\wedge c)$ but $p_3$ does not.
\item $p_3\sim_{\{F,\;R\}}p_4$, but $p_3\not\sqsubseteq^{\prime}_{FT}p_4$ and thus $p_3\not\sqsubseteq^{\prime}_{\{RT,\;RS\}}p_4$, because $p_3$ satisfies $a(\neg c \wedge b(\neg e \wedge d))$ but $p_4$ does not.
\item $p_5\sim_{\{F,\;FT\}}p_6$, but $p_5\not\sqsubseteq^{\prime}_{R}p_6$ and thus $p_5\not\sqsubseteq^{\prime}_{\{RT,\;RS\}}p_6$, since $p_5$ satisfies $ab(c\wedge d)$ but $p_6$ does not.
\item $p_6\sim_{\{F,\;R,\;RT,\;FT\}}p_7$ but $p_7\not\sqsubseteq^{\prime}_{RS}p_6$, because $p_7$ satisfies $a(bc\wedge bd)$ but $p_6$ does not.
\item $p_7\sim_{\{F,\;R,\;RT,\;FT,\;RS\}}p_8$.
\end{iteMize}
\end{exa}

\subsection{Our new unified logical characterizations of the semantics}\label{logical_characterizations_coarsest_semantics}
Inspired by the semantics studied in
Section~\ref{logical_characterizations_popular_semantics}, next we define the
general format for the logics characterizing each of the semantics in the
spectrum. We start by enlarging the spectrum yet a bit more. 

\begin{defi}\hfill
\begin{enumerate}[(1)]
\item \emph{Universal semantics}. We define the set $\mathcal{L}^{\prime}_{U}$ of universal formulas that characterize the trivial semantics that identifies all the processes by $\mathcal{L}^{\prime}_{U} = \{ \top \}$.
\item \emph{Complete semantics}. We define the set $\mathcal{L}^{\prime}_{C}$
  of complete formulas  characterizing the semantics that only distinguishes
  the terminated processes from the non-terminated ones by
  $\mathcal{L}^{\prime}_{C} = \{ \top, \neg0\}$. 
\item \emph{Initial offer semantics}. 
  We define the set $\mathcal{L}^{\prime}_{I}$ of initial offer formulas
  characterizing the semantics that only observers the set of initial actions
  of a process by $\mathcal{L}^{\prime}_{I} = \{ \top, \neg0 \} \cup\{a\top
  \mid a \in \emph{Act}\}$. 
\end{enumerate}
\end{defi}
In the definition above the subformula $\neg 0$ is just syntactic sugar for the formula $\neg(\bigwedge_{a \in Act}\neg a \top)$. Therefore, once again all these new logics are sublogics of $\mathcal{L}_{HM}$ and, as a result, we do not need to define their semantics.

Note that $\mathcal{L}^{\prime}_I$ is a bit larger than the logic
$\mathcal{L}_I$ from
Section~\ref{logical_characterizations_popular_semantics}. Once again, this is
so in order to get a more uniform presentation of our logics: $\neg 0$ is
indeed redundant. By including it we immediately obtain that the complete semantics is coarser than the
initial offer semantics, because
$\mathcal{L}^{\prime}_{C}\subseteq\mathcal{L}^{\prime}_{I}$. Based on this
result we will also obtain that the complete simulation is coarser than
the ready simulation. 
Certainly, $\neg 0$ is redundant in $\mathcal{L}^\prime_{I}$ (but not in $\mathcal{L}^{\prime}_{C}$!), because by means of it we can only distinguish a process that cannot execute any action from any other that can execute someone. But using the corresponding $a\top$ formula we can also get that.

\subsubsection{The simulation semantics}
As repeatedly noted, the family of simulation semantics constitute the spine of the new
spectrum. All of them are defined in a homogeneous way thanks to the
notion of constrained simulation from \cite{FG08ifiptcs}. Next we present their
logical characterization. 

\begin{defi}\label{logica_sim}
Given a set of formulas $\mathcal{L}^{\prime}_{N}$ defining a semantics $N$, we define the set of formulas $\mathcal{L}^{\prime}_{NS}$ that characterizes the $N$-constrained simulation semantics by: \brannew{\conjsim{NS}{N}}{NS}
\end{defi}

Taking $N\in\{U, C, I\}$ we obtain $\mathcal{L}^{\prime}_{US}$,
$\mathcal{L}^{\prime}_{CS}$ and $\mathcal{L}^{\prime}_{IS}$, that we rewrite as
$\mathcal{L}^{\prime}_{S}$ and $\mathcal{L}^{\prime}_{RS}$ in
the first and last cases to emphasize the classic notation for simulation
semantics. From $\mathcal{L}^{\prime}_{S}$ we obtain
$\mathcal{L}^{\prime}_{SS}$, that we will denote as
$\mathcal{L}^{\prime}_{2S}$. To complete the collection of simulation semantics
in the spectrum we need $\mathcal{L}^{\prime}_{TS}$,
that will be based on $\mathcal{L}^{\prime}_{T}$, to be defined in the next
section. 

The definition above differs from the particular case of ready simulation in
Definition~\ref{def_logica} in the two first rules, by means of which we impose
that the process will traverse states which are in the corresponding
$N$-equivalence class all along the tree of computations checked by a formula
in $\mathcal{L}^{\prime}_{NS}$.  Note that the combination of positive and negated formulas allows us to shape each of these classes.
Next we state the equivalence between our logics for the simulation semantics and those by van Glabbeek's recalled in Table \ref{logic_table}. 

\begin{prop}
(1) $\mathcal{L}^{\prime}_{S} \sim \mathcal{L}_{S}$; (2) $\mathcal{L}^{\prime}_{CS} \sim \mathcal{L}_{CS}$; and (3) $\mathcal{L}^{\prime}_{2S} \sim \mathcal{L}_{2S}$.
\end{prop}
\begin{proof}
\hfill
\begin{iteMize}{$\bullet$}
\item[$(1)$] The clauses defining $\mathcal{L}^{\prime}_{S}$ and
  $\mathcal{L}_{S}$ produce the same set of formulas. The first two clauses in
  $\mathcal{L}^{\prime}_{S}$ only add the two trivial formulas $\top$ and $\neg
  \top$ because in $\mathcal{L}^{\prime}_{U} = \{\top\}$.
\item[$(2)$] Again, the sets of formulas produced by
  $\mathcal{L}^{\prime}_{CS}$ and $\mathcal{L}_{CS}$ are the same because the
  two first clauses of ${L}^{\prime}_{CS}$ can only generate $\top$,
  $\neg\top$, $0$ and $\neg 0$ from $\mathcal{L}^{\prime}_C =\{\top,
  \neg\top\}$. 
  $0$ is needed to reflect the second clause in the definition of
  $\mathcal{L}_{CS}$, while $\neg 0\equiv \bigvee_{a\in Act}a\top$ 
  so that any formula containing $\neg 0$ can be rewritten into a disjunction
  of formulas in $\mathcal{L}_{CS}$. 
\item[$(3)$] Once again, the sets generated by $\mathcal{L}^{\prime}_{2S}$ and
  $\mathcal{L}_{2S}$ are the same. The clause ``if $\sigma \in
  \mathcal{L}^{\prime}_S$ then $\sigma \in \mathcal{L}^{\prime}_{2S}$'' in
  $\mathcal{L}^{\prime}_{2S}$  does not generate any new formulas because
  $\mathcal{L}_{S}\subseteq \mathcal{L}_{2S}$ (the formulas in
  $\mathcal{L}_{S}$ are exactly those that can be created using only the last
  two clauses in the definition of $\mathcal{L}_{2S}$). 
\end{iteMize}
\end{proof}

\begin{rem}
We can use both positive formulas in $\mathcal{L}^{\prime}_{C}$ and their
negations for defining $\mathcal{L}^{\prime}_{CS}$ due to the fact that
$C$-constrained simulation can be built from 
the equivalence relation defined by $C$ as constraint. 
However, we could also use $\sqsubseteq_C$ as a constraint and then remove
the clause ``if $\sigma \in \mathcal{L}^{\prime}_{C}$ then $\sigma \in
\mathcal{L}^{\prime}_{S_C}$'', which generates $\neg 0 \in \mathcal{L}_{S_C}$.
The other clause, which generates $0 \in \mathcal{L}^{\prime}_{S_C}$, is
crucial and cannot be removed from the definition. 
These two facts also concur in the definition of the other simulation semantics
in the extended spectrum, for which we also present a logical characterization
including the two clauses above.   
\end{rem}

\subsubsection{Logical characterization of the linear semantics} 
\label{logica_lineales}
We start by defining the closure operators by means of which we express the
extent to which conjunction and negation can be used in the logical
characterizations of each of the linear semantics. 

\begin{defi}
Given a logical set $\mathcal{L}^{\prime}_{N}$ with $N \in \{U,C,I,T,S\}$, we
define: 
\begin{enumerate}[(1)]
\item Its \emph{symmetric closure} $\mathcal{L}_{N}^{\equiv}$ by: if $\sigma \in \mathcal{L}^{\prime}_{N}$ then $\sigma \in \mathcal{L}_{N}^{\equiv}$ and $\neg \sigma \in \mathcal{L}_{N}^{\equiv}$;  if $\sigma_i \in \mathcal{L}_{N}^{\equiv}$ for all $i\in I$ then $\bigwedge_{i\in I}\sigma_i \in \mathcal{L}_{N}^{\equiv}$.
\item Its \emph{negative closure}  $\mathcal{L}_{N}^{\neg}$ by: if $\sigma \in \mathcal{L}^{\prime}_{N}$ then $\neg \sigma \in \mathcal{L}_{N}^{\neg}$; if $\sigma_i \in \mathcal{L}_{N}^{\neg}$ for all $i\in I$ then $\bigwedge_{i\in I}\sigma_i \in \mathcal{L}_{N}^{\neg}$.
\item Its \emph{positive closure} $\mathcal{L}_{N}^{\surd}$ by: if $\sigma \in \mathcal{L}^{\prime}_{N}$ then $\sigma \in \mathcal{L}_{N}^{\surd}$;  if $\sigma_i \in \mathcal{L}_{N}^{\surd}$ for all $i\in I$ then $\bigwedge_{i\in I}\sigma_i \in \mathcal{L}_{N}^{\surd}$.
\end{enumerate}
\end{defi}

\begin{rem}
Obviously these closures make sense for any given logic $\mathcal{L}$, but we
prefer to restrict our attention to $\mathcal{L}^{\prime}_{N}$ since it will be
enough for our goal and gives rise to a simpler notation. 
\end{rem}

Whenever we have a bag of ``good'' properties (such as
$\mathcal{L}^{\prime}_{N}$ above), to assert by means of a single formula which
is the subset of properties that a certain element satisfies it is not enough
to assert that it satisfies each of them: we also need to
assert that it does not satisfy any of the rest. This is why we need formulas
in the symmetric closure. By contrast, if the only available formulas belong to
the negative (resp.~positive) closure, we can only assert that the element has
at most (resp.~at least) the enumerated properties. Next we present the unified 
logics for all the linear semantics in the spectrum. 

\begin{defi}\label{linear_logica}
Inspired by the orders $\leq_N^{l}$, $\leq_N^{l\supseteq}$, $\leq_N^{lf}$, and
$\leq_N^{lf\supseteq}$, we define the set of formulas
$\mathcal{L}^{\prime}_{\leq_N^{l}}$,
$\mathcal{L}^{\prime}_{\leq_N^{l\supseteq}}$,
$\mathcal{L}^{\prime}_{\leq_N^{lf}}$, and
$\mathcal{L}^{\prime}_{\leq_N^{lf\supseteq}}$, respectively,  by means of the
rules: 
\begin{enumerate}[(1)]
\item \lin{\conjyform{{\leq_N^{l}}}{N}}{{\leq_N^{l}}}
\item \linf{N}{{\leq_N^{l\supseteq}}}
\item \linc{N}{{\leq_N^{lf}}}
\item \linfc{N}{{\leq_N^{lf\supseteq}}}
\end{enumerate}
\end{defi}

\noindent Note that for the coarsest semantics (i.e. those
corresponding to plain refusals and plain readiness when $N=I$) we
only check for $N$ at the ``end'' of the formula because there are no
conjunctions in the corresponding languages
$\mathcal{L}^{\prime}_{\leq_N^{lf}}$ and
$\mathcal{L}^{\prime}_{\leq_N^{lf\supseteq}}$, except for those
stemming from the corresponding closures $\mathcal{L}_{N}^{\equiv}$
and $\mathcal{L}_{N}^{\neg}$.  The other two logics do introduce
additional conjunctions that allow to observe $N$ along the
computations.

We have used the negative and symmetric closures for the ``failures-based'' 
and``readies-based'' semantics, and we can use the positive closure to define
two new semantics that have not been considered earlier in this paper, nor
elsewhere as far as we know. 
For that we need to observe partial offers along a computation, or just at its
end, where $X$ is a partial offer of $p$ if $X \subseteq I(p)$. It is clear the
duality with respect to the failures semantics, where $F$ is a failure of $p$
if $I(p) \subseteq \overline{F}$. We can introduce these two new semantics at
each layer of the spectrum through the corresponding partial offers for
each $N \in \{U, C, I, T, S\}$. 

\begin{defi}\hfill
\begin{enumerate}[(1)]
\item The semantics of \emph{partial offer traces} for the constraint $N$ is that defined by the logic $\mathcal{L}^{\prime}_{\leq_N^{l\subseteq}}$ with: \linfmo{N}{\leq_N^{l\subseteq}}
\vspace{0.15cm}
\item The semantics of \emph{partial offers} for the constraint $N$ is that defined by the logic $\mathcal{L}^{\prime}_{\leq_N^{lf\subseteq}}$ with: \linfcmo{N}{{\leq_N^{lf\subseteq}}}
\end{enumerate}
\end{defi}

\noindent Duality between failures and partial offers causes the
picture of the complete layer of linear semantics for each \textit{N}
to become two diamonds that share the side corresponding to the
readies-based semantics. Now, recalling Theorem~\ref{teoequiv}.

\begin{prop}
\hfill
\begin{enumerate}[\em(1)]
\item $\mathcal{L}^{\prime}_{F}$ and $\mathcal{L}^{\prime}_{\leq_I^{lf\subseteq}}$ are  incomparable: $p\leq_I^{lf\supseteq} q$ does not imply $p\leq_I^{lf\subseteq}q$ and $p\leq_I^{lf\subseteq}q$ does not imply $p\leq_I^{lf\supseteq}q$.
\item $\mathcal{L}^{\prime}_{FT}$ and $\mathcal{L}^{\prime}_{\leq_I^{l\subseteq}}$ are incomparable: $p\leq_I^{l\supseteq}q$ does not imply $p\leq_I^{l\subseteq}q$ and $p\leq_I^{l\subseteq}q$ does not imply $p\leq_I^{l\supseteq}q$.
\end{enumerate}
\end{prop}
\begin{proof}
In fact, we have a stronger result by combining these two statements: if we
consider $p=ab+ac$, $q=a(b+c)$, and $r=p+q$, then $p=^{l\supseteq}_Ir$ but
$r\nleq_I^{lf\subseteq}p$, and $q=^{l\subseteq}_Ir$ but
$r\nleq_I^{lf\supseteq}q$.  
\end{proof}

Similar counterexamples exist for $N \in \{T,S\}$. However, for $N \in
\{U,C\}$, which produce the trace and the completed trace semantics,
respectively, it is easy  to prove that the six logics of the layer are 
equivalent. 

\begin{prop}
\hfill
\begin{enumerate}[\em(1)]
\item $\mathcal{L}^{\prime}_{\leq^{lf}_U} = \mathcal{L}^{\prime}_{\leq^{l}_U} = \mathcal{L}^{\prime}_{\leq^{l\supseteq}_U} = \mathcal{L}^{\prime}_{\leq^{l\subseteq}_U} = \mathcal{L}^{\prime}_{\leq^{lf\supseteq}_U} = \mathcal{L}^{\prime}_{\leq^{lf\subseteq}_U} = \mathcal{L}_{T}$
\item $\mathcal{L}^{\prime}_{\leq^{lf\supseteq}_C} = \mathcal{L}^{\prime}_{\leq^{lf\subseteq}_C} = \mathcal{L}^{\prime}_{\leq^{l\supseteq}_C} = \mathcal{L}^{\prime}_{\leq^{l\subseteq}_C} = \mathcal{L}^{\prime}_{\leq^{lf}_C} = \mathcal{L}^{\prime}_{\leq^{l}_C} = \mathcal{L}_{CT}$.
\end{enumerate}
\end{prop}
\begin{proof}
\hfill
\begin{enumerate}[(1)]
\item Trivial, since the sets of clauses defining
  $\mathcal{L}^{\prime}_{\leq^{lf}_U}$ and $\mathcal{L}_{T}$ are almost the
  same. Note that the clause ``if $\sigma \in \mathcal{L}_U^{\equiv}$ then
  $\sigma \in \mathcal{L}^{\prime}_{\leq^{lf}_U}$'' does not give rise to new
  formulas because $\mathcal{L}_U^{\equiv} = \{\top\}$. 
\item Note that the sets of clauses defining
  $\mathcal{L}^{\prime}_{\leq^{lf\supseteq}_C}$ and $\mathcal{L}_{CT}$ are
  the same but for the clause ``if $\sigma \in \mathcal{L}_C^{\neg}$ then
  $\sigma \in \mathcal{L}^{\prime}_{\leq^{lf\supseteq}_C}$''.
  On the one hand,
  this causes $\neg\top\in
  \mathcal{L}^{\prime}_{\leq^{lf\supseteq}_C}$ (which adds nothing) because 
  $\top\in\mathcal{L}^{\prime}_{C}$ and thus $\neg\top\in \mathcal{L}_C^{\neg}$. 
  On the other hand, 
  we also have $0\in \mathcal{L}^{\prime}_{\leq^{lf\supseteq}_C}$
  because $\neg0 \in \mathcal{L}^{\prime}_{C}$ and then $\neg\neg0 \in
  \mathcal{L}_C^{\neg}$.\qedhere
\end{enumerate}
\end{proof}

\begin{cor}
$\mathcal{L}^{\prime}_{\leq^{lf}_U}\sim \mathcal{L}_{T}$ and $\mathcal{L}^{\prime}_{\leq^{lf\supseteq}_C}\sim \mathcal{L}_{CT}$.
\end{cor}

An interesting result illustrating the generality of our characterizations
concerns one of the finest semantics in the classic spectrum: possible
futures. Possible futures is located in Figure~\ref{fig:ltbtsAmpliado} below
2-nested simulation because the more accurate trace simulation semantics was
not yet included in the spectrum; this is corrected in the spectrum in
Figure~\ref{fig:extended-ltbts}. 
Indeed, for $N=T$ we have the following result. 

\begin{prop}
$\mathcal{L}^{\prime}_{\leq^{lf}_T} = \mathcal{L}_{PF}$.
\end{prop}
\begin{proof}
Trivial, since the sets of clauses defining
$\mathcal{L}^{\prime}_{\leq^{lf}_T}$ and $\mathcal{L}_{PF}$ are almost the
same: our definition includes the clause ``$\top \in
\mathcal{L}^{\prime}_{\leq^{lf}_T}$'', which does not appear explicitly in
that of $\mathcal{L}_{PF}$ because it corresponds to the conjunction
of an empty set of formulas. 
\end{proof}

\begin{cor}
$\mathcal{L}^{\prime}_{\leq^{lf}_T}\sim \mathcal{L}_{PF}$.
\end{cor}

\subsubsection{Logical characterization of the deterministic branching semantics} \label{logica_ramificada}
Now we consider the deterministic branching semantics. In the classic spectrum
the only such semantics is possible worlds but, as we pointed out before, there
is one such semantics at each layer of the extended spectrum. 

In order to capture determinism we need to consider conjunctive formulas to express the desired branching, but only when it corresponds to a choice between different actions. This leads us to the following scheme: 
$$\textrm{if $X\subseteq \emph{Act}$ and $\varphi_a \in \mathcal{L}_{D_N}$ for all $a \in X$, then $\bigwedge_{a\in X} a\varphi_a \in \mathcal{L}_{D_N}$.}$$

\begin{defi}\label{logica_det}
For each $N \in \{U,C,I,T,S\}$, we define the formulas of $\mathcal{L}^{\prime}_{D_N}$ by: \dbran{\conjyform{D_N}{N}}{D_N}
\end{defi}

For $N=I$ we obtain the unified logical characterization of the possible worlds
semantics. 
\begin{prop}
$\mathcal{L}^{\prime}_{D_I}\supseteq \mathcal{L}_{PW}$.
\end{prop}
\begin{proof}
Analogous to the case of ready simulation semantics.
\end{proof}

\begin{prop}
$\mathcal{L}^{\prime}_{D_I}\nsim \mathcal{L}_{PW}$.
\end{prop}
\begin{proof}
This is a consequence of the fact that the original logical characterization of the 
possible worlds semantics, $\mathcal{L}_{PW}$, was wrong. 
For instance, taking $p=abc+a(bc+d)+ab$ and 
$q=a(bc+d)+ab$ then $p \not\equiv_{PW} q$ but
$p\sim_{\mathcal{L}_{PW}}q$, since $\mathcal{L}_{PW}$ cannot ``observe'' the
intermediate offer that makes the possible world $abc$ different from those of
$q$. 
By contrast, the formula $\varphi= a(\neg d \wedge bc) \in
\mathcal{L}^{\prime}_{D_I}$ is enough to distinguish $p$ and $q$, since 
$p\models \varphi$ and $q \not\models \varphi$.
\end{proof}

We postpone to Section \ref{logic_observational_framework} the proof of the equivalence between our observational and 
logical characterizations of the possible worlds semantics. As a consequence of this 
correspondence, we have that a logical characterization only works in the infinite case
if we restrict ourselves to image-finite processes.

\begin{table}[t]
\begin{center}
\scalebox{0.85}{\footnotesize
\begin{tabular}{|c|c|c|c|c|c||c|}
\hline
\backslashbox{Formulas}{Constraints ($\mathcal{N}$)}& U & C & I & T & S & B\\
\hline
$\top  \in \mathcal{L}^{\prime}_{\mathcal{N}}$ & $\bullet$ & $\bullet$ & $\bullet$ & $\bullet$ & $\nu$ & $\nu$\\
\hline
$\textbf{$\neg \top$ = $\perp$}\in \mathcal{L}^{\prime}_{\mathcal{N}}$& $\nu$ &$\nu$&$\nu$ & $\nu$ & $\nu$ & $\nu$\\
\hline
$\textbf{$\neg0$}\in \mathcal{L}^{\prime}_{\mathcal{N}}$& &$\bullet$ & $\bullet$ & $\nu$ & $\nu$ & $\nu$\\
\hline
$a \in Act \Rightarrow \hspace{0.075cm} a\top \in \mathcal{L}^{\prime}_{\mathcal{N}}$ & & &$\bullet$ &$\nu$ &$\nu$ &$\nu$\\
\hline
$\varphi \in \mathcal{L}^{\prime}_{\mathcal{N}}, \hspace{0.075cm} a \in Act \Rightarrow$ & & & &\multirow{2}{*}{$\bullet$} &\multirow{2}{*}{$\bullet$} &\multirow{2}{*}{$\bullet$} \\
$a \varphi \in \mathcal{L}^{\prime}_{\mathcal{N}}$& & & & & &\\
\hline
$\varphi_i \in \mathcal{L}^{\prime}_{\mathcal{N}} \hspace{0.075cm} \forall i \in I \Rightarrow$ & & & & &\multirow{2}{*}{$\bullet$} &\multirow{2}{*}{$\bullet$}\\
$\bigwedge_{i \in I} \varphi_i \in \mathcal{L}^{\prime}_{\mathcal{N}}$ & & & & & &\\
\hline
$\varphi \in \mathcal{L}^{\prime}_{\mathcal{N}} \Rightarrow$ & & & & & &\multirow{2}{*}{$\bullet$}\\
$\neg \varphi \in \mathcal{L}^{\prime}_{\mathcal{N}}$ & & & & & &\\
\hline
\end{tabular}}
\vspace{1ex}
\caption{Logical characterizations of the semantics used as constraints.} \label{our_constraint_table}
\vspace{-0.2cm}
\end{center}
\end{table}

\begin{table}[th]
\begin{center}
\scalebox{0.85}{\footnotesize
\begin{tabular}{|c|c|c|c|c|c|c||c}
\hline
\multirow{4}{*}{\backslashbox{Formulas}{Semantics ($\mathcal{Y_N}$)}}& \multirow{2}{*}{$\mathbf{\leq_\textit{N}^{lf\supseteq}}$} & \multirow{2}{*}{$\mathbf{\leq_\textit{N}^{lf}}$} & \multirow{2}{*}{$\mathbf{\leq_\textit{N}^{l\supseteq}}$} & \multirow{2}{*}{$\mathbf{\leq_\textit{N}^{l}}$} &\multirow{2}{*}{$D_N$} & \multirow{2}{*}{$NS$} & \multicolumn{1}{c|}{\multirow{2}{*}{$N \in \{U,C,I,T,S\}$}}\\ 
& & & & & & & \multicolumn{1}{c|}{}\\
\cline{2-8}
& \multirow{2}{*}{F} & \multirow{2}{*}{R} & \multirow{2}{*}{FT} & \multirow{2}{*}{RT} & \multirow{2}{*}{PW} & \multirow{2}{*}{RS} & \multicolumn{1}{c|}{\multirow{2}{*}{when $N=I$}}\\
& & & & & & & \multicolumn{1}{c|}{}\\
\hline
$\top  \in \mathcal{L}^{\prime}_{{\mathcal{Y_N}}}$ & $\bullet$  & $\bullet$ & $\bullet$ & $\bullet$ & $\bullet$ & $\nu$ &\\
\cline{1-7}
$\varphi \in \mathcal{L}^{\prime}_{{\mathcal{Y_N}}}, \hspace{0.075cm} a \in Act \Rightarrow$ & \multirow{2}{*}{$\bullet$}  & \multirow{2}{*}{$\bullet$} & \multirow{2}{*}{$\bullet$} & \multirow{2}{*}{$\bullet$} & \multirow{2}{*}{$\nu$} & \multirow{2}{*}{$\bullet$} &\\
$a \varphi \in \mathcal{L}^{\prime}_{{\mathcal{Y_N}}}$ &  &  &  &  &  &  &\\
\cline{1-7}
$\varphi \in \mathcal{L}_{N}^{\neg}\Rightarrow$ & \multirow{2}{*}{$\bullet$}  & \multirow{2}{*}{$\nu$} & \multirow{2}{*}{$\nu$}  & \multirow{2}{*}{$\nu$}  & \multirow{2}{*}{$\nu$}  & \multirow{2}{*}{$\nu$}  &\\
$\varphi \in \mathcal{L}^{\prime}_{\mathcal{Y_N}}$ &  &  &  &  &  &  &\\
\cline{1-7}
$\varphi \in \mathcal{L}_{N}^{\equiv}\Rightarrow$ &  & \multirow{2}{*}{$\bullet$}  &  & \multirow{2}{*}{$\nu$}  & \multirow{2}{*}{$\nu$}  & \multirow{2}{*}{$\nu$}  &\\
$\varphi \in \mathcal{L}^{\prime}_{\mathcal{Y_N}}$ &  &  &  &  &  &  &\\
\cline{1-7}
$\varphi \in \mathcal{L}^{\prime}_{{\mathcal{Y_N}}}, \hspace{0.075cm} \sigma \in \mathcal{L}_{N}^{\neg}\Rightarrow$ &  &  & \multirow{2}{*}{$\bullet$}  & \multirow{2}{*}{$\nu$}  & \multirow{2}{*}{$\nu$}  & \multirow{2}{*}{$\nu$}  &\\
$\sigma \wedge \varphi \in \mathcal{L}^{\prime}_{\mathcal{Y_N}}$ &  &  &  &  &  &  &\\
\cline{1-7}
$\varphi \in \mathcal{L}^{\prime}_{{\mathcal{Y_N}}}, \hspace{0.075cm} \sigma \in \mathcal{L}_{N}^{\equiv}\Rightarrow$ &  &  &  & \multirow{2}{*}{$\bullet$}  & \multirow{2}{*}{$\bullet$}  & \multirow{2}{*}{$\nu$}  &\\
$\sigma \wedge \varphi \in \mathcal{L}^{\prime}_{\mathcal{Y_N}}$ &  &  &  &  &  &  &\\
\cline{1-7}
$X\subseteq Act, \hspace{0.075cm} \varphi_a \in \mathcal{L}^{\prime}_{\mathcal{Y_N}} \hspace{0.075cm} \forall a \in X \Rightarrow$ &  &  &  &  & \multirow{2}{*}{$\bullet$} &\multirow{2}{*}{$\nu$} &\\
$\bigwedge_{a \in X} a\varphi_a \in \mathcal{L}^{\prime}_{\mathcal{Y_N}}$ &  &  &  &  &  &  &\\
\cline{1-7}
$\varphi_i \in \mathcal{L}^{\prime}_{{\mathcal{Y_N}}} \hspace{0.075cm} \forall i \in I \Rightarrow$ &  &  &  &  &  & \multirow{2}{*}{$\bullet$} &\\
$\bigwedge_{i \in I} \varphi_i \in \mathcal{L}^{\prime}_{{\mathcal{Y_N}}}$ &  &  &  &  &  &  &\\
\cline{1-7}
$\varphi \in \mathcal{L}_{N} \Rightarrow$ &  &  &  &  &  & \multirow{2}{*}{$\bullet$}   &\\
$\varphi \in \mathcal{L}^{\prime}_{\mathcal{Y_N}}$ &  &  &  &  &  &  &\\
\cline{1-7}
$\varphi \in \mathcal{L}_{N} \Rightarrow$ &  &  &  &  &  & \multirow{2}{*}{$\bullet$}  &\\
$\neg \varphi \in \mathcal{L}^{\prime}_{\mathcal{Y_N}}$ &  &  &  &  &  &  &\\
\cline{1-7}
\end{tabular}}
\vspace{1ex}
\caption{Our new logical characterizations for the semantics at each level of the spectrum.} \label{our_logic_table}
\vspace{-0.2cm}
\end{center}
\end{table}

In Tables~\ref{our_constraint_table} and \ref{our_logic_table} we present our
results in a three-dimensional way. Table~\ref{our_logic_table} shows the rules
defining the logics characterizing each of the semantics at each layer of the
spectrum. On top of it also appears, as example, the classic notation for 
the corresponding semantics represented when $N=I$.
Table~\ref{our_constraint_table} contains the logics that characterize the
constraint governing each of these layers.  There are 
two semantics that are included in both tables, in order to emphasize their double
role as ``main'' and ``auxiliary'' semantics. However they are disguised under
different names: this is the case of $T = {\leq_U^{l}}$ (in fact, it is also equal to the other
three linear $U$-semantics) and $S=US$. 

\section{Relating the unified logics and the unified observational model}\label{logic_observational_framework}

In this section we will relate our unified logical characterizations and the
unified observational semantics. As indicated in
Section~\ref{sec:Preliminaries}, we have to restrict ourselves to image-finite
processes; as a byproduct, the finite parts of each of the corresponding
languages, that are obtained by intersection with $\mathcal{L}^{f}_{HM}$,
provide us with a pure finite logical characterization of the semantics.
However, it is convenient in the first part of this Section to consider still the full (infinitary) logic characterizing each of the semantics.

\begin{defi}[\textbf{Normal formulas $\mathcal{N(L)}$}]\label{def:normal}
\hfill
\begin{enumerate}[(1)]
\item Given a set of formulas $\mathcal{L}$ whose outermost operator is not
  conjunction,  the set $\mathcal{N(L)}$ of induced \emph{normal formulas} is
  defined by: 
  \begin{iteMize}{$\bullet$}
  \item $\top\in \mathcal{N(L)}$;
  \item if $\Gamma_1, \Gamma_2\subseteq \mathcal{L}, \{a_i \mid i\in I\} \subseteq Act$, and $\varphi_i \in \mathcal{N(L)}$, then $(\bigwedge_{\sigma\in \Gamma_1}\sigma \wedge \bigwedge_{\sigma \in \Gamma_2} \neg \sigma) \wedge \bigwedge_{i\in I}a_i \varphi_i \in \mathcal{N(L)}$.
  \end{iteMize}
\item For each $N \in \{U,C,I,T,S\}$ and each $\mathcal{Y_N} \in \{NS, \leq_N^{l}, \leq_N^{l\supseteq}, \leq_N^{lf}, \leq_N^{lf\supseteq}, \leq_N^{l\subseteq}, \leq_N^{lf\subseteq}, D_N\}$ in the spectrum, we define the set of normal formulas $\mathcal{N_{Y_N}(L_N^{\prime \prime})} \subseteq \mathcal{L}^{\prime}_{\mathcal{Y_N}}$  as $\mathcal{N_{Y_N}(L_N^{\prime \prime})}=\mathcal{N(L_N^{\prime\prime})}\bigcap \mathcal{L}^{\prime}_{\mathcal{Y_N}}$, where $\mathcal{L}_{N}^{\prime \prime}$ is the set of formulas in $\mathcal{L}^{\prime}_{N}$ whose outermost operator is not conjunction.\end{enumerate}
\end{defi}

\begin{rem}
The clause in Definition~\ref{def:normal}.1 is more involved than it
appears. Initially, we can apply it with $I=\emptyset$ to obtain the first
(non-trivial) normal formulas and then recursively to obtain more
complex normal formulas; note that the two first subformulas
stem always from the original set $\mathcal{L}$. 
By abuse of notation, when some of the elements in our normal
formulas do not appear in the corresponding set
$\mathcal{L}^{\prime}_{\mathcal{Y_N}}$, we assume that these formulas have been
extended by conjunction with $\top$ using the fact that $\bigwedge_{\sigma \in
  \emptyset} \sigma$ is another syntactic form to express $\top$. 


Also note that infinite conjunction is allowed in the two first
subformulas. As a consequence, if we consider the tree-like form of these (possibly infinitary) 
formulas they could have infinite depth. However, if we define the normal depth of formulas in $\mathcal{N(L_N)}$
as that obtained by counting the recursive nesting in the application of
Definition~\ref{def:normal}, then any normal formula has finite normal depth,
and the set they form can be explored by structural induction. 
\end{rem}

\begin{thm}\label{equivalencia_logica_normal}
Each set of normal formulas $\mathcal{N_{Y_N}(L_N^{\prime \prime})}$ associated
to the semantics in the spectrum  is equivalent to the full set of formulas
$\mathcal{L}^{\prime}_{\mathcal{Y_N}}$. 
\end{thm}
\begin{proof}
By structural induction, all the formulas in
$\mathcal{L}^{\prime}_{\mathcal{Y_N}}$ admit a normal formula in the sense of
Definition~\ref{def:normal}, that is obtained by gathering the subformulas and
applying Proposition~\ref{disyuncion_teo}. 
\end{proof}

\begin{defi}
The set of \emph{complete normal} formulas $\mathcal{CN(L)}$ (resp., the set
of complete normal formulas associated to each semantics in the spectrum,
$\mathcal{CN_{Y_N}(L_N^{\prime \prime})}$) is the set of normal formulas
(resp., the set of normal formulas associated to each semantics in the
spectrum) for which the rule in Definition~\ref{def:normal} is applied with
$\Gamma_2= \overline{\Gamma_1}$.
\end{defi}

Now we prove that infinite conjunction in Definition~\ref{def:normal} can be approximated by finite conjunction.

\begin{thm}\label{aproximacion}
If we restrict ourselves to image-finite processes, for each denumerable set of formulas $\mathcal{L}$, any complete normal formula
$\varphi \in \mathcal{CN(L)}$ can be approximated by  a set of finite normal
formulas $\{ \varphi^{k} \mid k \in \nat\}$ that only use finite conjunction,
that is,  $p\models \varphi$ iff $p \models \varphi^{k}$ for all $k \in \nat$.
\end{thm}
\begin{proof}
We define the sequence $\varphi^{k}$ by structural induction on the normal
depth of $\varphi$: 
\begin{iteMize}{$\bullet$}
\item $\varphi= (\bigwedge_{\sigma\in \Gamma_1}\sigma \wedge \bigwedge_{\sigma
    \in \overline{\Gamma_1}} \neg \sigma)$. We consider a fixed enumeration of
  the set $\mathcal{L}=\{\sigma_n \mid n \in \nat\}$, and define
  $\mathcal{L}^{\leqslant n}= \{\sigma_j \in \mathcal{L} \mid
  j\leqslant n\}$. Then, for each $k\in \nat$: 
  $$\varphi^{k}= \bigwedge_{\sigma \in \Gamma_1 \cap \mathcal{L}^{\leqslant k}} \sigma \wedge \bigwedge_{\sigma \in \overline{\Gamma_1}
    \cap \mathcal{L}^{\leqslant k}} \neg \sigma.$$
We have $p\models \varphi$ $\Leftrightarrow$ $(p\models \sigma\; \forall \sigma \in \Gamma_1$ and $p\not\models \sigma\;\forall \sigma \not\in \Gamma_1)$ and $p\models \varphi^k \; \forall k \in \nat$ $\Leftrightarrow$ $(p\models \sigma\; \forall \sigma \in \Gamma_1\cap\mathcal{L}^{\leqslant k}$ and $p\not\models \sigma\;\forall \sigma \in \overline{\Gamma_1}\cap\mathcal{L}^{\leqslant k})$ and the result follows from a the equality
$$\Gamma_1=\Gamma_1\cap\mathcal{L}=\Gamma_1\cap(\bigcup_{k\in\nat}\mathcal{L}^{\leqslant{n}}).$$

\item $\varphi= (\bigwedge_{\sigma\in \Gamma_1}\sigma \wedge \bigwedge_{\sigma
    \in \overline{\Gamma_1}} \neg \sigma) \wedge  \bigwedge_{i \in I} a_i
  \varphi_i$. By structural induction we can assume that the result is true for
  any subformula $\varphi_i$. Then we define $\varphi^{k}= \bigwedge_{\sigma
    \in \Gamma_1 \cap \mathcal{L}^{\leqslant k}} \sigma \wedge
  \bigwedge_{\sigma \in \overline{\Gamma_1} \cap \mathcal{L}^{\leqslant
      k}} \neg \sigma \wedge \bigwedge_{i \in I} a_i\varphi^{k}_i$. Now, if we
  decompose $\varphi$ as $\varphi_I \wedge \varphi_{II}$ (taking
  $\varphi_{II}=\bigwedge_{i\in I}a_i\varphi_i$, and analogously for the set of
  approximations) we have that $p\models \varphi^{k} $ iff $p\models
  \varphi^{k}_I$ and $p\models \varphi^{k}_{II}$. If $p\models \varphi^{k}$
  then $p\models \varphi^{k}_I$ for all $k \in \nat$ and arguing as in
  the base case above we conclude that $p\models \varphi_I$. Any image-finite
  process $p$ can be decomposed as $p=\sum_{a_i \in Act} \sum_{j=1}^{m_i}
  a_i^{j}p_i^{j}$, and we have $p\models \varphi^{k}_{II}$ iff for all $i$ there
  exists $j$ with $a_i=a_i^{j}$ and $p_i^{j} \models \varphi_i^{k}$. Then, if
  $p\models \varphi_{II}^{k}$ for all $k \in \nat$, for each $i$ there
  exists some $j \in 1..m_i$ such that  $p_i^{j} \models \varphi_i^{k}$
  for infinitely many $k$, but this means that $p_i^{j} \models
  \varphi_i^{k}$ for all $k \in \nat$ and then, by  the induction
  hypothesis, $p_i^{j} \models \varphi_i$ thus getting $p\models
  \varphi$. 
\end{iteMize}
\end{proof}

\begin{defi}
For each $N \in \{U,C,I,T,S\}$ and each $\mathcal{Y_N} \in \{NS, \leq_N^{l}, \leq_N^{l\supseteq}, \leq_N^{lf}, \leq_N^{lf\supseteq}, \leq_N^{l\subseteq}, \leq_N^{lf\subseteq},$ $D_N\}$ in the spectrum, we define the finite logic for the semantics $\mathcal{L}^f_{\mathcal{Y_N}}$ as $\mathcal{L}^\prime_{\mathcal{Y_N}}\cap\mathcal{L}^f_{HM}$.
\end{defi}

\begin{cor} \label{cor_aprox}
For each $N \in \{U,C,I,T,S\}$ and each $\mathcal{Y_N} \in \{NS, \leq_N^{l}, \leq_N^{l\supseteq}, \leq_N^{lf}, \leq_N^{lf\supseteq}, \leq_N^{l\subseteq}, \leq_N^{lf\subseteq},$ $D_N\}$ in the spectrum, if we restrict ourselves to the set of image-finite processes we have $\mathcal{L}^f_{\mathcal{Y_N}}\sim \mathcal{L}^\prime_{\mathcal{Y_N}}$.
\end{cor}
\begin{proof}
We only need to apply Theorem \ref{aproximacion}. The only non trivial case is when $N=S$, where we have to apply twice the Theorem, using also the fact that $CN(\mathcal{L}^{\prime\prime}_S)\sim CN(\mathcal{L}^f_S)$, because $\mathcal{L}^{\prime\prime}_S\sim\mathcal{L}^f_S$.
\end{proof}

\begin{thm}\label{isomorfia_observaciones}
For each $N \in \{U,C,I,T,S\}$ and each $\mathcal{Y_N} \in \{NS, \leq_N^{l}, \leq_N^{l\supseteq}, \leq_N^{lf}, \leq_N^{lf\supseteq}, \leq_N^{l\subseteq}, \leq_N^{lf\subseteq}, D_N\}$ in the spectrum there exists a correspondence between the set of complete normal formulas $\mathcal{CN_{Y_N}(L_N^{\prime \prime})}$ and the corresponding domain of observations $\Omega GO_N$ with $\Omega \in \{B,L\}$. This correspondence $\leftrightarrow$ satisfies that $\varphi \leftrightarrow \theta$ implies that $(p \models \varphi$ iff $\theta \in \Omega GO_N(p))$.
Moreover:
\begin{enumerate}[\em(1)]
\item The set of complete normal formulas $\mathcal{CN_{NS}(L_N^{\prime \prime})}$ (resp. $\mathcal{CN_{D_N}(L_N^{\prime \prime})}$) and the domain of branching general observations $BGO_N$ (resp. $dBGO_N$) are isomorphic, that is, $\leftrightarrow$ is one to one.
\item The set of complete normal formulas $\mathcal{CN}_{\leq_N^{l}}\mathcal{(L_N^{\prime \prime})}$, $\mathcal{CN}_{\leq_N^{l\supseteq}}\mathcal{(L_N^{\prime \prime})}$ and the domain of linear general observations $LGO_N$ are isomorphic, that is, $\leftrightarrow$ is one to one.
\item The set of complete normal formulas $\mathcal{CN}_{\leq_N^{lf}}\mathcal{(L_N^{\prime \prime})}$ (resp. $\mathcal{CN}_{\leq_N^{lf\supseteq}}\mathcal{(L_N^{\prime \prime})}$) and the quotient domain $LGO_N/ _{\simeq_{N}^{lf}}$ (resp. $LGO_N/ _{\simeq_{N}^{lf\supseteq}}$) are isomorphic, that is, $\leftrightarrow^{-1}$ is injective and $\varphi \leftrightarrow \theta$ iff $\theta \simeq_N^{lf\supseteq} \theta_{\varphi}$, for some adequate $\theta_{\varphi}$.
\end{enumerate}
\end{thm}
\begin{proof}
\hfill
\begin{enumerate}[(1)]
\item
As can be seen in Figure~\ref{labeledtree-in-proof},
a branching observation is a labeled tree whose nodes are local observations
and whose arcs are labeled by actions.  

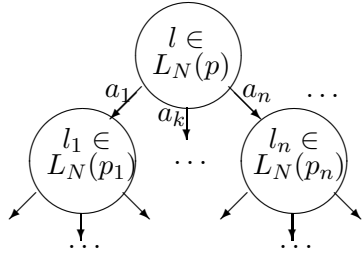
\begin{figure}
\begin{picture}(110,90)
\put(46,78){$l \in$} \put(40,68){$L_N(p)$} \put(84,40){$l_n \in$} \put(78,30){$L_N(p_n)$} \put(6,40){$l_1 \in$} \put(0,30){$L_N(p_1)$}
\put(22,59){$a_1$} \put(42,50){$a_k$}  \put(74,59){$a_n$}
\put(8,2){$\ldots$} \put(48,34){$\ldots$}  \put(88,2){$\ldots$}   \put(98,59){\normalsize{$\ldots$}}
\put(50,75){ \circle{40}}
\put(36,63){\vector(-1,-1){12}}  \put(53,55){\vector(0,-1){12}}  \put(69,63){\vector(1,-1){12}}
\put(90,35){ \circle{40}}
\put(76,23){\vector(-1,-1){10}}   \put(93,15){\vector(0,-1){10}}   \put(109,23){\vector(1,-1){10}}
\put(10,35){ \circle{40}}
\put(-4,23){\vector(-1,-1){10}}   \put(13,15){\vector(0,-1){10}}   \put(29,23){\vector(1,-1){10}}
\end{picture}
\caption{A branching observation.}\label{labeledtree-in-proof} 
\end{figure}

The general form of any complete normal formula in $\mathcal{CN_{NS}(L_N)}$ is
$(\bigwedge_{\sigma\in \Gamma}\sigma \wedge \bigwedge_{\sigma \notin \Gamma}
\neg \sigma) \wedge \bigwedge_{i\in I}a_i \varphi_i$, with $\varphi_i \in
\mathcal{CN_{NS}(L_N)}$ for all $i \in I$.  
Since the language $\mathcal{L}^{\prime}_N$ characterizes the semantics used to
get the local observations, we can associate to each complete formula
$(\bigwedge_{\sigma\in \Gamma}\sigma  \wedge \bigwedge_{\sigma \notin \Gamma}
\neg \sigma)$ the corresponding local observation $l \in L_N$. 
Then, by structural induction, we obtain the observation associated to
each formula $\varphi_i \in \mathcal{CN_{NS}(L_N)}$, thus getting the branching
general observation $BGO_N$ associated to the given formula. It is easy to see
that this correspondence is indeed a bijection. 

The case for $\mathcal{CN_{D_N}(L_N)}$ is analogous, but now it is not allowed to have repeated actions in the arcs leaving any node of an observation; this is obviously reflected in the form of the formulas in the corresponding language.

\item The case for $\mathcal{CN}_{\leq_N^{l}}\mathcal{(L_N)}$ is similar to the
  previous one, but now the obtained (degenerated) tree is just a single
  branch corresponding to a $lgo$ in $LGO_N$.  

For $\mathcal{CN}_{\leq_N^{l\supseteq}}\mathcal{(L_N)}$ the general form of  a
complete normal formula is $\varphi=(\top \wedge \bigwedge_{\sigma \notin \Gamma}
\neg \sigma) \wedge a\varphi^{\prime}$, with $\varphi^{\prime}\in
\mathcal{CN}_{\leq_N^{l\supseteq}}\mathcal{(L_N)}$. If we close the set
$\Gamma$ by derivability obtaining $\Gamma^{\prime}$ and then consider its
complement $\overline{\Gamma^{\prime}}$, we can consider the local observation
\textit{l} that satisfies all the formulas in $\overline{\Gamma^{\prime}}$ and
none in $\Gamma^{\prime}$. The linear general observation \textit{lgo}
corresponding to $\varphi$ is then recursively defined as $\langle
l,\{(a,lgo^{\prime})\}\rangle$ where $lgo^{\prime}$ is the linear general
observation corresponding to $\varphi^{\prime}$. 
 
To proceed in the opposite direction, we just need to take as $\Gamma$ the
complement of the set of formulas in $\mathcal{L}^{\prime}_N$ satisfied by the
local observation \textit{l} at the root of the given $LGO_N$, and then 
proceed in a recursive way. 

\item In this case, the general form of a complete normal formula in
  $\mathcal{CN}_{\leq_N^{lf}}\mathcal{(L_N)}$ is $\varphi=\top \wedge a_1(\ldots
  (\top \wedge a_{n-1}(\top \wedge a_n(\bigwedge_{\sigma\in \Gamma}\sigma
  \wedge \bigwedge_{\sigma \notin \Gamma} \neg \sigma)\ldots)$. Now we
  establish a correspondence between the set of local observations $L_N$ and
  the sets $\Gamma \subseteq \mathcal{L} _N$ as done in cases (1) and (2) above,
  and then define the correspondence $\leftrightarrow$ by ignoring the
  values of all the intermediate local observations in the considered $lgo$,
  keeping only the local observation at the end. 

For $\mathcal{CN}_{\leq_N^{lf\supseteq}}\mathcal{(L_N)}$ we just need to apply the same procedure above combined with the ideas along the proof for $\mathcal{CN}_{\leq_N^{l\supseteq}}\mathcal{(L_N)}$.\qedhere
\end{enumerate}
\end{proof}

\begin{rem}
It came as a surprise to notice that the $lgo's$ in $LGO_N$ are in a
bijective relation both with the complete normal formulas in
$\mathcal{N}_{\leq_N^{l}}\mathcal{(L_N)}$ and those in
$\mathcal{N}_{\leq_N^{l\supseteq}}\mathcal{(L_N)}$, so let us consider the case
$N=I$ to explain this fact. A $\textit{cnf}$ in
$\mathcal{N}_{\leq_I^{l}}\mathcal{(L_I)}$ specifies the corresponding local
observation $I(p)\subseteq \mathcal{P}(Act)$ by means of a formula
$(\bigwedge_{\sigma \in \Gamma}\sigma \wedge \bigwedge_{\sigma \notin \Gamma}
\neg \sigma$), where the formulas in $\Gamma$ are just the elements of the
corresponding set $I(p)$ while those in $\overline{\Gamma}$ correspond to its
complement. 
When considering the failure trace semantics, the formulas in
$\mathcal{N}_{\leq_I^{l\supseteq}}\mathcal{(L_I)}$ only contain the 
part $\bigwedge_{\sigma \notin \Gamma}\neg \sigma$ corresponding to the complement
$\overline{I(p)}$. Since in this case the
sets of \textit{lgo}'s could be assumed to be closed with respect
to the order $\leq_N^{l\supseteq}$ in Definition~\ref{obs:def}, soundness is
retained after ``assuming'' that any formula $\bigwedge_{\sigma \notin
  \Gamma}\neg \sigma$ ``generates'' the observation associated to $\Gamma$,
even though some of the formulas $\sigma \in \Gamma$ may not be
satisfied when the corresponding $I(p)$ is smaller. But
for the failures and failure trace semantics we can proceed by closing the set
of offers upwards with respect to $\subseteq$ and no new failure is introduced. 
\end{rem}

\begin{thm}
For each $N \in \{U,C,I,T,S\}$ and each $\mathcal{Y_N} \in \{NS, \leq_N^{l}, \leq_N^{l\supseteq}, \leq_N^{lf}, \leq_N^{lf\supseteq}, \leq_N^{l\subseteq}, \leq_N^{lf\subseteq}, D_N\}$ in the spectrum, if we restrict ourselves to image-finite processes,
the logical semantics   $\sqsubseteq^{f}_{\mathcal{Y_N}}$ induced by the
logic $\mathcal{L}^{f}_{\mathcal{Y_N}}$, is equivalent to the corresponding observational semantics in
Definitions~\ref{bgo:def}, \ref{lbo:def} and \ref{branch_obs_def}.  
In order to unify our notation, here we will denote by $GO_N$ the corresponding
semantic domain. 
\end{thm}
\begin{proof}
By Theorem~\ref{equivalencia_logica_normal},
$\mathcal{L}^{\prime}_{\mathcal{Y_N}} \sim \mathcal{N_{Y_N}(L_N)}$, and from
Theorem~\ref{isomorfia_observaciones} we get the isomorphism between the set
$\mathcal{CN_{Y_N}(L_N)}$ and the corresponding set of general observations
$GO_N$. 

To finish the proof, we just need to show that $\mathcal{N_{Y_N}(L_N)}$ and
$\mathcal{CN_{Y_N}(L_N)}$ are equivalent.  Any consistent formula  in
$\mathcal{N_{Y_N}(L_N)}$ ($\Gamma_1 \bigcap \Gamma_2= \emptyset$) provides only
some partial information about the states in a computation, so that the
concrete values of these states are characterized by a set $\Gamma$ with
$\Gamma_1\subseteq \Gamma \subseteq \overline{\Gamma_2}$. Therefore, we can
replace $\Gamma_1$ and $\Gamma_2$ by $\Gamma$ and $\overline{\Gamma}$,
respectively, adding the disjunction over all the possible values of $\Gamma$,
to characterize the set of processes specified by the formula. Now it is enough
to float the disjunction up to obtain a disjunction of formulas in
$\mathcal{CN_{Y_N}(L_N)}$, and applying Proposition~\ref{disyuncion_teo} we get the equivalence between the two sets of
formulas.  Finally, we only need to apply Corollary \ref{cor_aprox} to conclude.
\end{proof}

\begin{cor}
\hfill
\begin{enumerate}[\em(1)]
\item The unified logical semantics in Definition~\ref{logica_sim} is equivalent to the $N$-simulation semantics.
\item The unified logical semantics in Definition~\ref{linear_logica}.1 is equivalent to the $N$-ready trace semantics.
\item The unified logical semantics in Definition~\ref{linear_logica}.2 is equivalent to the $N$-failure trace semantics.
\item The unified logical semantics in Definition~\ref{linear_logica}.3 is equivalent to the $N$-readiness semantics.
\item The unified logical semantics in Definition~\ref{linear_logica}.4 is equivalent to the $N$-failure semantics.
\item The unified logical semantics in Definition~\ref{logica_det} is equivalent to the $N$-deterministic branched semantics.
\end{enumerate}
Moreover, if we restrict ourselves to image-finite processes we have also an equivalence with the corresponding finite logical semantics.
\end{cor}
\begin{proof}
Since it was proved in Section~\ref{observational-sem-sec} that any
observational semantics characterizes the corresponding (classical) semantics
in the (extended) ltbt spectrum, the desired equivalence between our (unified)
logical characterizations and the classical semantics is an immediate corollary. 
\end{proof}

\section{On the real diamond structure}
\label{sec:RealDiamond}

This section is a practical proof of the suitability of our unification
work. Some recently proposed semantics that were not in the original
ltbt spectrum are nicely included in our extended spectrum, which shows why and
how the old spectrum has to be expanded. Our unified approach immediately
absorbs these new semantics and the results about the different
characterizations are easily extended to cover them. 
We warmly thank Roscoe for pointing out to us his work on the stable
revivals semantics \cite{ReedEtAl07,Roscoe09}, where an
endeavor for an adequate presentation of the notion of responsiveness for a 
CSP-like language is made.
(Responsiveness had been previously studied by Fournet et al. in 
\cite{FournetEtAl04} for CCS, under the name of stuck-freeness.)

When faced with the diamond shape of the collection of linear semantics that are associated to each simulation semantics in the
extended spectrum, it would be natural to expect it to reflect the structure of
a lattice.
Then, failure semantics would be the greatest lower bound of the 
readiness and failure
trace semantics, while ready trace semantics would be the corresponding lowest
upper bound.
However, both intuitions turn out to be wrong and a new semantics finer than failures
and another one coarser than ready trace can be found: together with readiness and failure trace, they do constitute a lattice.

Let us first consider the case of the lowest upper bound.
We postulate that the axiomatization of the associated semantics is obtained by instantiating our general axiom with
the conjunction of the two conditions $M_R$ and $M_{FT}$:
\[
M_{R\land FT}(x,y,w) \iff I(x) \supseteq I(y) \ \textrm{and}\ I(w)\subseteq I(y).
\]
We denote with $\sqsubseteq_{R\land FT}$ the order axiomatized by the 
corresponding axiom $(\textit{ND}^{R\land FT})$.
\begin{defi}
The \emph{readiness and failure trace semantics}, or \emph{join semantics $R\wedge FT$}, is that defined by the order
$\sqsubseteq_{R\land FT}$ generated by the set of
axioms $\{ \textrm{$B_1$--$B_4$}, (RS), (\textit{ND}^{R\land FT})\}$.
\end{defi}

\begin{prop}
The ready trace semantics is strictly finer than the readiness and failure
trace semantics.
\end{prop}
\begin{proof}
${\sqsubseteq_{RT}} \subseteq {\sqsubseteq_{R\land FT}}$ is an immediate 
consequence of Proposition~\ref{static-1:prop} and the fact that condition
$M_{RT}$ implies both $M_R$ and $M_{FT}$, and hence also $M_{R\land FT}$.
To show that ${\sqsubseteq_{RT}} \not\subseteq {\sqsubseteq_{R\land FT}}$,
let us take $w = \cero$, $y = b$, and $x = bB' + c$; then we have:
\[
\underbrace{a(bB + bB' + c)}_p \sqsubseteq_{R\land FT}
                   \underbrace{a(bB' +c) + abB}_q
\]
but, if $I(B) \neq I(B')$,
\[
a(bB + bB' + c) \not\sqsubseteq_{RT} a(bB' +c) + abB
\]
because $\{a\}a\{b,c\} b I(B) \in 
\textit{ReadyTraces}(p)\setminus\textit{ReadyTraces}(q)$.
\end{proof}

It is clear that the readiness and failure trace semantics is finer than 
both the readiness and the failure trace semantics; to show that it is
actually the coarsest upper bound we need to prove that 
${\sqsubseteq_{R\land FT}}= {\sqsubseteq_R\cap \sqsubseteq_{FT}}$. Even if the axiom $(\textit{ND}^{R\land FT})$ was created with this goal in mind, this cannot be
easily shown using only algebraic arguments.
Instead, it is trivial to obtain the observational characterization of the 
desired semantics by gathering together the failure trace and the
ready observations.
Based on Definition~\ref{obs:def}, we can define the corresponding
order $\leq_N^{l{\supseteq} \land f}$ by taking 
\[
\calT \leq_N^{l{\supseteq} \land f} \calT' \iff
\calT \leq_N^{l{\supseteq}} \calT' \ \textrm{and}\ \calT\leq_N^{lf}\calT'.
\]

A direct characterization can be obtained as follows.
We combine both kinds of observations into a single family of decorated traces that we 
call \emph{failure trace with final ready sets}, by considering failure sets
all along the trace except at the end of it, where we introduce the corresponding
ready set. 

\begin{defi}
We define the order $\leq_N^{l{\supseteq}\land f}$ by
\[
\begin{array}{lll}
\calT\leq_N^{l{\supseteq}\land f} \calT'& \Longleftrightarrow&
\textrm{for all $X_0a_1\dots X_n\in\calT$
        there is some $Y_0a_1\dots Y_n\in\calT'$}\\
&&\textrm{with $X_n=Y_n$
          and $X_i\supseteq Y_i$, for $i\in 0..n-1$}.
\end{array}
\]
\end{defi}

\begin{prop}
The semantics defined by the order $\sqsubseteq_{R\land FT}$ coincides with that
defined by $\leq_I^{l{\supseteq}\land f}$ and is thus the lowest upper bound of the readiness
and failure trace semantics.
\end{prop}
\begin{proof}
Similar to that of Theorem~\ref{sandc:thm}.
\end{proof}

Let us finally consider the logical characterization of this semantics. It is
clear that the conjunction of two semantics should be characterized in a
logical way by simply considering the union of the logics that characterize
both semantics (although there could possibly be a more compact presentation). 

\begin{defi}\label{wedge_vee_logica}
We define the set of formulas $\mathcal{L}^{\prime}_{\leq_I^{l\supseteq\wedge
    f}}$ as that generated by the clauses:
\lin{\conjyformneg{\leq_I^{l\supseteq\wedge f}}{I}; \item
  \conjbis{\leq_I^{l\supseteq\wedge f}}{I}}{\leq_I^{l\supseteq\wedge f}} 
\end{defi}

\begin{prop}
The logical semantics  $\sqsubseteq^{\prime}_{\leq_I^{l\supseteq\wedge f}}$ induced by the logic $\mathcal{L}^{\prime}_{\leq_I^{l\supseteq\wedge f}}$ is equivalent to the observational semantics defined by ${\leq_I^{l\supseteq\wedge f}}$.
\end{prop}
\begin{proof}
We just need to check that $\mathcal{L}^{\prime}_{\leq_I^{l\supseteq\wedge f}}$ = $\mathcal{L}^{\prime}_{\leq^{l\supseteq}_I} \cup \mathcal{L}^{\prime}_{\leq^{lf}_{I}}$, which is immediate.
\end{proof}
By replacing the \textit{I} above by the generic \textit{N}, we get the definitions and results for the general case.

The axiomatic characterization of the greatest lower bound of the readiness and failure trace
semantics is much simpler: we simply put together the axioms for the orders 
defining both semantics. 

\begin{defi}
The \emph{meet semantics $R\vee FT$} is that defined by the order $\sqsubseteq_{R\lor FT}$ generated by the set of axioms
$\{ \textrm{$B_1$--$B_4$}, (RS),$ $(\textit{ND}^R), (\textit{ND}^{FT})\}$.
\end{defi}

If we define $M_{R\lor FT}$ as $M_R\lor M_{FT}$, that is,
$M_{R\lor FT}(x,y,w)$ holds if $I(x) \supseteq I(y)$ or $I(w) \subseteq I(y)$, 
we have the following characterization of $\sqsubseteq_{R\lor FT}$.
\begin{prop}
The order $\sqsubseteq_{R\lor FT}$ is that generated by the set of axioms
$\{ \textrm{$B_1$--$B_4$}, (RS),$ $(\textit{ND}^{R\lor FT})\}$, where
$(\textit{ND}^{R\lor FT})$ is the instantiation of the generic axiom 
$(\textit{ND})$ with the condition $M_{R\lor FT}$.
\end{prop}

\begin{prop}
The semantics defined by the order $\sqsubseteq_{R\lor FT}$ is the finest
semantics that is coarser than both the readiness and the failure trace
semantics.
\end{prop}
\begin{proof}
Obvious since any semantics coarser than the readiness semantics has to
satisfy $\{ \textrm{$B_1$--$B_4$}, (RS), (\textit{ND}^R)\}$, any one coarser 
than failure trace must satisfy 
$\{ \textrm{$B_1$--$B_4$}, (RS), (\textit{ND}^{FT})\}$,
and $M_{R\lor FT}$ is equivalent to $M_R\lor M_{FT}$.
\end{proof}

Once again, the semantics defined by $\sqsubseteq^{R\lor FT}$ is not included
in the ltbt spectrum and neither in our extended one; in particular, it is
different from the failures semantics.
To prove this fact we make essential use of the notion of \emph{revival},
as defined by Reed, Roscoe, and Sinclair \cite{ReedEtAl07}.
Revivals are sequences $a_1,\dots,a_n (X,a)$ where $a_1,\dots,a_n$ is a trace
of the corresponding process after which the action $a$ is
offered, but the set of actions $X$ is refused.

\begin{prop}
The meet semantics $R\vee FT$ is strictly finer than failure semantics.
\end{prop}
\begin{proof}
The inclusion ${\sqsubseteq}^{R\lor FT} \varsubsetneq {\sqsubseteq^F}$ is obvious since failures semantics is coarser than 
both the readiness and the failure trace semantics.
To show that the inclusion is strict, note that any two processes related 
by $\sqsubseteq^{R\lor FT}$ do not only have the same failures but also 
the same revivals.
This is indeed the case since all the axioms $u\preceq v$ in 
$\{ \textrm{$B_1$--$B_4$}, (RS), (\textit{ND}^R), (\textit{ND}^{FT})\}$
preserve the revivals, which means 
$\textit{Revivals}(\sigma(u)) \subseteq \textit{Revivals}(\sigma(v))$ for
every ground substitution $\sigma$, and the revivals order is a precongruence
for the operators in BCCSP.
For instance, for $(\textit{ND}^{FT})$ we need to prove that 
$\textit{Revivals}(\sigma(a(x+y))) \subseteq \textit{Revivals}(\sigma(ax))\cup
\textit{Revivals}(\sigma(a(y+w)))$ whenever $I(\sigma(w))\subseteq 
I(\sigma(x))$.
It is clear that the only non-trivial case occurs when 
$a(X,b)\in \textit{Revivals}(\sigma(a(x+y)))$; then we have
$(X,b)\in \textit{Revivals}(\sigma(x+y))$ so that 
$X \in\textit{Failures}(\sigma(x)) \cap \textit{Failures}(\sigma(y))$ and
$b\in I(\sigma(X))$ or $b\in I(\sigma(y))$.
In the first case, $a(X,b)\in \textit{Revivals}(\sigma(ax))$ whereas, in
the second, $X\in\textit{Failures}(\sigma(x+y))$ and therefore 
$a(X,b)\in \textit{Revivals}(\sigma(a(x+y)))$.
The case for $(\textit{ND}^R)$ is simpler.
Once we know that $\sqsubseteq^{R\lor FT}$ preserves the revivals we only need
to observe that the revivals cannot be obtained from the failures of a process.
In particular, we have $ab\sqsubseteq^F a+a(b+c)$, but 
$a(\{c\},b) \in  \textit{Revivals}(ab)\setminus\textit{Revivals}(a+a(b+c))$.
\end{proof}

\comment{The semantics $\sqsubseteq^{R\lor FT}$, though not in the ltbt spectrum, is not
completely new since it coincides with the revivals semantics (at least
for BCCSP).
To prove this, we first give a}
Next we present the characterization of the revivals semantics
in terms of our observational framework.

\begin{defi}
We define the order $\leq_N^{l{\supseteq}\lor f}$ by
\[
\begin{array}{lll}
\calT \leq_N^{l{\supseteq}\lor f} \calT' &\Longleftrightarrow &
\textrm{for all $X_0a_1\dots X_n\in\calT$}\\
&&\textrm{there is $\{Y_0a_1Y_1\dots Y_n^j\mid j\in J\}\subseteq \calT'$
        such that $X_n = \bigcup_{j\in J} Y_n^j$}.
\end{array}
\]
\end{defi}

\begin{prop}
For all $p,q \in \textit{BCCSP}$, $\textit{Revivals}(p)\subseteq
\textit{Revivals}(q)$ if and only if 
$\textit{LGO}_I(p) \leq_I^{l{\supseteq}\lor f} \textit{LGO}_I(q)$.
\end{prop}
\begin{proof}
Note that $\leq_I^{l{\supseteq}\lor f}$ can be equivalently defined as
\[
\begin{array}{lll}
\calT\leq_I^{l{\supseteq}\lor f}\calT' &\Longleftrightarrow &
\textrm{for all $X_0a_1\dots X_n\in\calT$ and
        for all $a\in X_n$}\\
&&\textrm{there is $Y_0a_1\dots Y_n\in\calT'$
        with $a\in Y_n$ and $Y_n\subseteq X_n$}.
\end{array}
\]
Now, since $a_1\dots a_n(X,a)\in\textit{Revivals}(p)$ if and only if
there exists $X_0a_1\dots X_n\in \textit{LGO}_I(p)$ such that
$a\in X_n$ and $X_n\cap X = \emptyset$, we obtain the desired
characterization.
\end{proof}

\begin{defi}
Given $\calT\subseteq \textit{LGO}_N$, $\ol{\calT}^{{\supseteq}\lor f}$ is defined
as
\[
\ol{\calT}^{{\supseteq}\lor f} = 
\{ X_0a_1\dots X_n\mid 
\textrm{there is $\{ Y_0a_1\dots Y_n^j \mid j\in J\}\subseteq \calT$
        with $X_n= \bigcup_{j\in J}Y_n^j\}$}.
\]
\end{defi}
This clearly indicates that $\leq_I^{l{\supseteq}\lor f}$ is in between 
$\leq_I^{lf\supseteq}$, defining the failures semantics, and 
$\leq_I^{lf}$, defining readiness semantics. 
This is useful for the proof of the axiomatic characterization of the
revivals semantics.

\begin{thm}
The revivals semantics defined by $\sqsubseteq_I^{l{\supseteq}\lor f}$ is
axiomatized by $\{ \textrm{$B_1$--$B_4$}, (RS), 
(\textit{ND}^{R\lor FT})\}$.
\end{thm}
\begin{proof}[Proof sketch]
It is quite similar to that of Theorem~\ref{sandc:thm} for the case of failures
semantics and, hence, also similar to the characterization of that semantics
by means of acceptance trees \cite{Hen88ATP} (and where the closure of the
set of offers with respect to both union and convex closure is a critical
argument), and this is why we only sketch it.
In connection to that, recall that the application of the particular case of 
$(\textit{ND})$ corresponding to $(\textit{ND}^{FT})$ allowed us to join
arbitrary states after the same trace, while that corresponding to 
$(\textit{ND}^R)$ allowed us to obtain a common continuation after the
same action at any state reachable by the same trace.
All this can be done now using $(\textit{ND}^{R\lor FT})$; however, we cannot
add to an arbitrary state an action offered at another state reachable by the 
same trace since to do that we needed the unlimited strength of axiom
$(\textit{ND})$.
\end{proof}

Note that for the join semantics $R\wedge FT$ the logical approach was the most direct way of defining it, whereas its equational characterization needed more care.
For the meet semantics $R \vee FT$, the situation is just the opposite. As we
have seen, $R\vee FT$ is axiomatized by putting together the axioms for $R$ and
those for $FT$; in contrast, the logic characterizing $R\vee FT$ is obtained by
cleverly selecting the common part of the logics characterizing both $R$ and
$FT$. 
If we had defined the logical semantics by considering all the formulas from
HML that are preserved by each semantics, then we could take the intersection
of these sets as the logical semantics of any meet semantics. Since we defined
our logical semantics by considering only a ``basis'' that generates the
corresponding full set, we cannot simply take their intersection. 

\begin{defi}
We define the set of formulas $\mathcal{L}^{\prime}_{\leq_I^{l\supseteq\vee f}}$ as that generated by the clauses: \lin{\conjU{\leq_I^{l\supseteq\vee f}}{I}}{\leq_I^{l\supseteq\vee f}}
\end{defi}

Note that in the second clause of this definition we have relaxed the condition in the definition of $\mathcal{L}^{\prime}_{R}$ by considering an arbitrary failure (that defined by the set $J$), but only a positive offer (the action appearing in $\sigma$). This is how the revivals semantics becomes slightly finer than the failures semantics.

\begin{prop}
The logical semantics  $\sqsubseteq^{\prime}_{\leq_I^{l\supseteq\vee f}}$ induced by the logic $\mathcal{L}^{\prime}_{\leq_I^{l\supseteq\vee f}}$ is equivalent to the observational semantics defined by ${\leq_I^{l\supseteq\vee f}}$.
\end{prop}
\begin{proof}
In this case we have taken $\mathcal{L}^{\prime}_{\leq_I^{l\supseteq\vee f}} = \mathcal{L}^{\prime}_{\leq^{l\supseteq}_I} \cap \mathcal{L}^{\prime}_{\leq^{lf}_{I}}$.
Then, to prove that it defines $R \vee FT$ it is enough to check that $p \not\sqsubseteq^{l\supseteq \vee f}_{I} q$ implies that there exists $\varphi$ in $\mathcal{L}^{\prime}_{\leq^{l\supseteq \vee f}_{I}}$ such that $p\models \varphi$ and $q \not\models \varphi$, which is almost immediate.
\end{proof}
Again, by replacing the \textit{I} above by the generic \textit{N}, we get the definitions and results for the general case.

We can generalize most of the results obtained for the refusal semantics when
$N=I$ to any reasonable local observation function such as $T$ or $S$, once we
interpret  
$\subseteq$ as the corresponding order and $=$ as the induced equivalence.
However, in order to define the adequate observational characterization of
the revivals semantics for a local observation (or constraint) $N$, we should
look for the adequate ``elements'' of the universe of observations. 
This leads us to traces when $N$ is $T$, but it is not so clear
how to define those ``elements'' for a non-extensional semantics such as
that obtained when $N$ is $S$.

Let us conclude this section with a look at the picture in 
Figure~\ref{fig:diamond} showing the real structure of the full (bidimensional!)
diamond, that should be included in all the upper levels of the extended ltbt 
spectrum.

\begin{figure}[tbp]
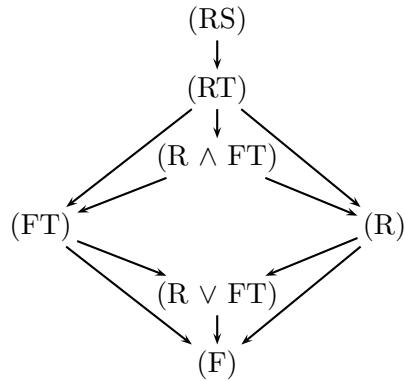

\begin{center}
{
\psset{nodesep=2pt}
\begin{tabular}{c@{\qquad}ccc@{\qquad}c}
  & & & \\ & & \rnode{rs}{(RS)} & \\
  & & & \\  & & \rnode{rt}{(RT)} & \\
  & & & \\  & & \rnode{raft}{(\textrm{R $\land$ FT})} & \\
  & & & \\ \rnode{ft}{(FT)} & & & & \rnode{r}{(R)}\\
  & & & \\  & & \rnode{roft}{(\textrm{R $\lor$ FT})} & \\
  & & & \\ & & \rnode{f}{(F)} & 
  \ncline{->}{rs}{rt}
  \ncline{->}{rt}{ft}
  \ncline{->}{rt}{r}
  \ncline{->}{ft}{f}
  \ncline{->}{r}{f}
  \ncline{->}{rt}{raft}
  \ncline{->}{raft}{ft}
  \ncline{->}{raft}{r}
  \ncline{->}{ft}{roft}
  \ncline{->}{r}{roft}
  \ncline{->}{roft}{f}
\end{tabular} 
}
\end{center}
 \caption{The real diamond below ready simulation.}
 \label{fig:diamond}
\end{figure}

\section{Operational semantics}
\label{sec:OperationalSemantics}
In this section we explain how to develop the semantics in the spectrum
in an operational way. Certainly, this presentation could be argued to
be ad-hoc at times since some
``high-level'' conditions are required in the SOS-like rules for some of the
semantics. Moreover, the style of presentation at this Section is certainly less precise and detailed than in the previous ones.
However, we believe it still provides some additional insight on 
the common properties of the semantics
and also establishes a connection with our previous work on (bi)simulations
up-to \cite{FG07sos,FG09ic} as a way to get coinductive characterizations
of any ``reasonable'' process semantics.

Structural operational semantics was introduced by G.~Plotkin in 1981, even
though his seminal work was not published in a journal until 2004 \cite{Plo04}. 
In Section~\ref{sec:Preliminaries} we already presented a basic operational
semantics for our processes as a starting point for the definition of all the
semantics in the spectrum: a small-step semantics that collects the
(atomic) actions executed by the processes into the corresponding
transition system.  
By contrast, all the operational semantics in this section will be
big-step semantics which directly return the adequate semantic values 
defining each of the semantics. They are generated by means of SOS-like rules that
obtain these values in a compositional way.  
An extensive presentation of structural operational semantics covering all its 
variants can be found in \cite{AFV01}.

\subsection{Local simulations up-to}
\label{sec:LocalSimulations}

In order to characterize all the reasonable behavior preorders in
a coinductive way we need to generalize constrained $N$-simulations
(Definition~\ref{def:constrainedsimulation}) with $N$-simulations up-to an
order $\ltp$.

\begin{defi}
  \label{def:ConstrainedSimulationUpto}
   Let $\ltp$ be a behavior preorder and
   $N$ a relation over processes.
   We say that a binary relation
   $S$ over processes
   is an \emph{$N$-simulation up-to $\ltp$} if $S\subseteq N$
   and  $S$ is a simulation
   up-to $\ltp$. Or equivalently, in a coinductive way,
   whenever we have $pSq$ we also have:
   \begin{iteMize}{$\bullet$}
   \item for every $a$, if $p\tran{a}p'_a$ then there exist $q',q'_a$ such that
     $q\gtp q'\tran{a} q'_a$ and $p'_a S q'_a$;
   \item $pNq$.
   \end{iteMize}
   We say that process $p$ is $N$-simulated up-to $\ltp$
   by process $q$, or that process $q$ $N$-simulates
   process $p$ up-to $\ltp$, written
   $p\ltlbi^N_{\ltp}q$, if there exists an $N$-simulation up-to $\ltp$, $S$,
   such that $pSq$.
\end{defi}

We often just write $\ltlbi^N$,
instead of $\ltlbi^N_{\ltp}$, when the behavior preorder
is clear from the context.

We proved in~\cite{FG07sos} that all the preorders  defining the
semantics in the ltbt spectrum can be
characterized as $N$-simulations up-to the corresponding
equivalence relation $\eip$, where $N$ is the constraint defining
the coarsest simulation semantics finer than the given semantics.
For instance, the result for the semantics between failures semantics and
ready simulation was the following.
\begin{thm}[\cite{FG07sos}]
\label{the:LtlbiThenPreorder} For every behavior preorder $\ltp$
satisfying the axiom ($RS$) and ${\ltp} \subseteq I$, we have
$p\ltp q$ if and only if $p\ltlbi^I_\ltp q$.
\end{thm}

Table~\ref{table:constraints} shows the constraints defining
the adequate constrained simulation order finer than each of the
semantics in the linear time-branching time spectrum. Obviously, they coincide with the layer of the extended spectrum at which each semantics appear.

\begin{table}[th]
    \centering
    {\small
    \begin{tabular}[h]{|c|c|c|c|c|c|c|c|c|c|c|c|c|}
      \cline{2-13}
      \multicolumn{1}{c|}{}&\enspace T\enspace&\enspace S\enspace &\thinspace CT\thinspace&\thinspace CS\thinspace&\enspace F\enspace &\enspace R\enspace &\thinspace FT\thinspace&\thinspace RT\thinspace&\thinspace PW\thinspace&\thinspace RS\thinspace&\thinspace PF\thinspace&\thinspace 2N\thinspace\\
      \hline
      \thinspace$C_O$\thinspace&$U$&$U$&$C$&$C$&$I$&$I$&$I$&$I$&$I$&$I$&$V$&$W$\\
      \hline
    \end{tabular}
}
    $$
    \begin{array}{cclcccl}
     pUq & \Longleftrightarrow & true& & pVq & \Longleftrightarrow & p \eip_T q\\
     pCq & \Longleftrightarrow & (p=\cero \textrm{ iff } q=\cero)&\mbox{\qquad}&
     pWq & \Longleftrightarrow & p \eip_S q\\
     pIq & \Longleftrightarrow & I(p)=I(q)
    \end{array}
    $$
    \caption{Constraints for the semantics in the ltbt spectrum.}
    \label{table:constraints}
\end{table}

Note that Theorem~\ref{the:LtlbiThenPreorder} is more subtle that it could appear: it characterizes a given preorder with a constrained simulation upto the preorder itself (Definition~\ref{def:ConstrainedSimulationUpto}). Therefore, there are several semantics that share the same constraint. This characterization is indeed rather technical and the key point is that it allows to express any behavior preorder in a simulation-like fashion. We have used this characterization to prove many useful statements in our previous work\footnote{For instance in \cite{FG08ifiptcs} (Theorem 10) we provided an axiomatization for any behavior preorder starting from the equations of the corresponding equivalence.} and we will use it again several times in the current Section.

In our proof of the completeness of the axiomatizations for the
linear semantics in the spectrum in Section~\ref{rnoef:sec} we
used a notion of normal form which, roughly, was defined by applying repeatedly 
to any term $p$ the axiom $(\textit{ND}_\eip)$ from right to left, for as
long as possible.
Propositions~\ref{hnfx-aux:prop} and \ref{hnfx-aux2:prop} were then 
the key results to complete the proof, and also lie behind the intuition for 
introducing now the notion of \emph{local I-simulation up-to}.

\begin{defi}
 For $Z\in\{F,R,FT,RT\}$ and $p =\sum_a\sum_i ap^i_a$, whenever we have
 a pair of indices $i$, $j$ and a decomposition $ p^j_a = r^j_a + s^j_a $
 with $M_Z(p^j_a, r^j_a , s^j_a)$ we say that $p$ is 1-locally
 $Z$-equivalent to $q=p+a(p^i_a + r^j_a)$, and we write $p \eip^{l1}_Z q$. 
 We say  that $p$ and $q$ are locally $Z$-equivalent when
 they are related by the reflexive and transitive closure of
 $\eip^{l1}_Z$, and then we write $p \eip^{l}_Z q$.

 For $Z\in\{F,R,FT,RT\}$ we refer to the $I$-simulations up-to $\eip^{l}_Z$ as local $I$-simulations up-to $\eip_Z$. 
 We say that process $p$ is locally
 $I$-simulated up-to $\eip_Z$ by process $q$, or that process $q$ locally 
 $I$-simulates process $p$ up-to $\eip_Z$, written
 $p \ltlbi^I_{\eip^{l}_Z}q$, if there exists a local $I$-simulation up-to 
 $\eip_Z$, $S$, such that $pSq$.
\end{defi}

Local $I$-simulations up-to are enough to characterize the 
linear semantics in $\{F,R,FT,RT\}$.
Note that we cannot get a local notion of bisimulation up-to 
equivalent to our unrestricted notion of bisimulation up-to.

\begin{prop}
\label{prop:localisis}
 For $Z\in\{F,R,FT,RT\}$ we have $p\ltp_Z q$ if and only if 
 $p\ltlbi^I_{\eip^{l}_Z} q$.
\end{prop}
\begin{proof}
The implication from right to left is an immediate consequence of
Theorem~\ref{the:LtlbiThenPreorder}. For the other, note
that $\{ (p,q) \mid p\ltp_Z q \} $ is a local $I$-simulation up-to
$\eip_Z$. Indeed, for any $p \tran{a} p^i_a $ we have $q
\eip^{l}_Z \textit{hnf\/}^Z(q)$ and taking 
$\textit{hnf\/}^Z(q) = \sum_a\sum_i ah^j_a$
there exists some $j$ such that $\textit{hnf\/}^Z(q) \tran{a}h^j_a$ and
$p^i_a \ltp_Zh^j_a$.
\end{proof}

\begin{exa}
Let us consider the processes $p = abc+abd$ and $q=a(bc+bd)$.
We have $p\equiv_F q$ and we can check that $p\ltlbi^I_{\equiv^l_F} q$ since
$p\sqsubseteq_{RS} q$.
In order to prove that we also have $q\ltlbi^I_{\equiv^l_F} p$, we apply 
$\equiv^l_F$ to $p$ to obtain $p\equiv^l_F p+q$ and then we obtain
$q\sqsubseteq_{RS} p$.

By contrast, if we wanted to apply our bisimulation up-to characterization
to prove directly that $p\equiv_F q$ then we would have to turn $q$ into $q+p$
in order to simulate the transition $p\tran{a}bc$. 
This would correspond to the local application of $(\textit{ND\/}^F_\equiv)$
combined with that of 
\[
(RS_\equiv)\quad I(x) = I(y) \Longrightarrow a(x+y) \simeq a(x+y) + ax.
\]
But if we replace the action $a$ by a larger prefix $a_1\dots a_n$ then 
we should also modify the process $q'=a_1\dots a_n(bc+bd)$ in a non-local
way in order to obtain $q''= q'+p'$, so that we could suitably simulate
the transition
$p'=a_1\dots a_nbc+a_1\dots a_nbd\tran{a_1} a_2\dots a_nbc$.
Certainly, this is not necessary when checking $p'\equiv_F q'$ by means
of local simulations up-to.
\end{exa}

The coinductive characterization of the semantics by means of
simulations up-to has at least two important advantages over that of using
bisimulations up-to.
First, we can characterize the orders defining the semantics and not just
the induced equivalences; and second, we can use a local variant of the
up-to mechanism so that we only need to rely on the equivalence relation
$\equiv^l_Z$ for the up-to part.

\subsection{Operational rules for the linear semantics of processes}
\label{sec:sos}

In Section~\ref{sec:LocalSimulations} we have introduced and
proved some results that establish the framework using which we achieve our 
goal: to define for each of the classic linear
semantics an operational semantics over BCCSP terms in such a way that we can use
constrained simulations to characterize the considered semantics. For instance, if we consider the case of the failures preorder
$\ltp_F$, we are going to define a new operational semantics for
BCCSP terms $(\proc,\act,\Rightarrow_F)$ such that $p\ltp_F q$ if
and only if $q$ ready simulates $p$ in
$(\proc,\act,\Rightarrow_F)$.

Next we will concentrate first on the diamond of linear semantics
coarser than ready simulation. All these semantics are based on
the observation of the initial set of actions of each process,
that can be obtained by application of the SOS-like rules in
Figure~\ref{fig:initials}.

\begin{figure}[tbp]
\[
\begin{array}{ccccc}
  \cero\flee_I \emptyset &\mbox{\qquad} & 
  ap\flee_I \{a\} &\mbox{\qquad} & 
 \frac{\displaystyle p\flee_I A\quad q\flee_I B}{\displaystyle p+q\flee_I A\cup B}
\end{array}
\]
  \caption{Rules that compute the set of initial actions of a process.}
  \label{fig:initials}
\end{figure}


The rules in Figure~\ref{fig:sos} define the transition relation
$\Longrightarrow_Z$ that induces the operational semantics to
characterize each of the $Z$-semantics.  The transition relation
$\longleftrightarrow_Z$ is an auxiliary relation that captures the
iterated application of the axiom $(\textit{ND\/}^Z_\eip)$. Rules
(RF) and (TR) define reflexivity and transitivity of the relation
$\longleftrightarrow_Z$. Finally, the rule (CL) combines the
auxiliary relation $\longleftrightarrow_Z$ and the original
operational transition relation $\longrightarrow$ (see
Figure~\ref{fig:OperationalSemanticsBCCSP}), to define the new labeled
transitions $\Longrightarrow_Z$.

\begin{figure}[t]
\[
\begin{array}{lr}
\multicolumn{2}{c}{(\mbox{ND})\;\; \frac{\displaystyle p\flee_I A_p\quad  
  q\flee_I A_q\quad  r\flee_I A_r\quad  
  M_Z(A_p,A_q,A_r)}{\displaystyle 
                   ap+a(q+r)+s\longleftrightarrow_Z ap+a(q+r)+a(p+q)+s}} \\ \\
(\mbox{RF})\;\;  p\longleftrightarrow_Z p&
  (\mbox{TR})\;\; \frac{\displaystyle p\longleftrightarrow_Z q\quad  
  q\longleftrightarrow_Z r}{\displaystyle p\longleftrightarrow_Z r}\\ \\
(\mbox{CL})\;\;\frac{\displaystyle 
                     p\longleftrightarrow_Z p'\quad 
                     p'\tran{a}q}{\displaystyle p\Tran{a}_Z q  }\\
\end{array}
\]
  \caption{Operational semantics characterizing the linear semantics.}
  \label{fig:sos}
\end{figure}

\begin{defi}\label{def:sos}
For $Z\in\{F,R,FT,RT\}$,
the operational semantics for BCCSP terms is given by
the labeled transition system $(\proc,\act,\Longrightarrow_Z)$ where the
transition 
relation $\Longrightarrow_Z$ is defined by the rules in Figure~\ref{fig:sos}.
\end{defi}

By abuse of notation, we have written $M_Z(A_p,A_q,A_r)$ to express that
we check $M_Z(p,q,r)$ using the initials computed by $\longrightarrow_I$.

The relation $\Longrightarrow_Z$ has some interesting properties. 
First, it is an extension of the original transition system.

\begin{prop}
  \label{pro:preservetraces}
  For $Z\in\{F,R,FT,RT\}$, $p$ and $q$ BCCSP processes,
  and $\alpha$ a sequence of actions in $\act$,
  we have that
  $p\Tran{\alpha}q$ implies $p\stackrel{\alpha}{\Longrightarrow_Z} q$.
\end{prop}

Although usually some new transitions appear, the set of
initial actions of any process always remains the same.
\begin{cor}
  \label{cor:preserveinitials}
  For $Z\in\{F,R,FT,RT\}$ and for any BCCSP process $p$, we have
  $I^{{}^{\rightarrow}}(p) = I^{{}^{\Rightarrow_{Z}}}(p)$.
\end{cor}

It is also clear that, for any $Z\in\{F,R,FT,RT\}$, the auxiliary
relation $\longleftrightarrow_Z$ preserves the equivalence
$\eip_Z$ because the rule $(ND)$ corresponds to the application of
axiom $(\textit{$I$-ND}^Z_{\eip})$, which is sound with respect to
$\eip_Z^I$.

\begin{prop}
  \label{pro:sosequivalence}
   For $Z\in\{F,R,FT,RT\}$ and any two BCCSP processes $p$ and $q$, we have
   $p\longleftrightarrow_Z q$ implies $p\eip_Zq$.
\end{prop}

Now we prove the main theorem in this section, that asserts
that for each of the semantics in the considered diamond we can define the 
corresponding operational semantics as stated in Figure~\ref{fig:sos}. 

\begin{thm}
\label{thm:operationalsimulation}
 For $Z\in\{F,R,FT,RT\}$ and any two BCCSP processes $p$ and $q$,
 we have
 \[
 p\ltp_Z q \iff p\ltp^{\Rightarrow_Z}_{RS} q.
 \]
\end{thm}
\begin{proof}
We will apply our characterization of the orders $\ltp_Z$ by means
of local $I$-simulations up-to at Proposition~\ref{prop:localisis} to
 show that $p\ltp^{\Rightarrow_Z}_{RS} q$ implies
$p\ltlbi^I_{\eip^{l}_Z} q$. This is because any ready simulation
over the transition system $\Longrightarrow_Z$ is also a local
$I$-simulation up-to $\eip_Z$. Indeed, if $R$ is a ready simulation
over the transition system $\Longrightarrow_Z$, and $pRq$, then whenever we
have
 $p \tran{a} p' $ we also have  $p \Tran{a}_Z p' $, and therefore
 there is some $q \Tran{a}_Z q'$ with $p'Rq'$. By definition of the
 transition system $\Longrightarrow_Z$, there is some process $q''$ such
 that $q\longleftrightarrow_Z q''$ and $q''\tran{a}q'$. Then we
 also have $q \eip^{l}_Z q''$, and thus $R$ is indeed a local
 $I$-simulation up-to $\eip_Z$.

To prove that $p\ltlbi^I_{\eip^{l}_Z} q$ implies
$p\ltp^{\Rightarrow_Z}_{RS} q$, we will check that the relation 
$\ltlbi^N_{\eip^{l}_Z}$ is a ready simulation over the transition
relation $\Longrightarrow_Z$. If $p \ltlbi^N_{\eip^{l}_Z}q$,
whenever $ p \Tran{a}_Z p'$ we have some process $p''$ such that
$p\longleftrightarrow_Zp''$ and $p''\tran{a}p'$. Then we also
have $p \eip_Z p''$, and so $p'' \ltlbi^N_{\eip^{l}_Z} q$. From
$p''\tran{a}p'$ we now obtain that there are processes $q'$ and $q''$ such that $q \eip^{l}_Z
q''$, $q''\tran{a}q'$, and therefore we also have
$q\longleftrightarrow_Z q''$, thus concluding the proof.
\end{proof}

As a consequence of our negative results at the end of 
Section~\ref{sec:LocalSimulations},
it is not possible to obtain an operational semantics locally defined from
that which characterizes the linear semantics by means of bisimilarity.
However, this can be done if we use mutual similarity instead of bisimulation.

Certainly, the fact that the characterizations in terms of bisimilarity 
cannot be defined in a local way is related to the fact that the 
transition systems generated by application of the algorithm in
\cite{CH93} are larger than those generated by our local 
transformation here.
Unfortunately, it is true that our presentation does not magically lead (at 
least at the theoretical
level) to more efficient algorithms to decide the equivalences with respect
to the linear semantics (which are known to be quite hard to decide).
Obviously, this is related to the fact that simulation is harder than
bisimulation \cite{KM02concur}.
Even so, these are just theoretical worst case bounds, and it is nice to know
that in practice we can apply a local transformation to generate the
transition systems characterizing those semantics by means of 
the simulation orders, that in many concrete cases will not be too difficult 
to decide.

\subsection{Characterizing the semantics corresponding to other 
constraints}
\label{cscoc:sec}

Let us start by considering the case of the universal constraint $U$.
As discussed in Section~\ref{tcss:sec}, if we use $U$ in 
the condition $M_Z$ it is clear that all the semantics in the corresponding
diamond collapse into a single one: trace semantics.
It is immediate to realize that the transition system to characterize
it in terms of plain simulations is the same transition system 
$\Longrightarrow_F$ that we use to characterize the failures semantics
by means of ready simulations.

\begin{thm}
The trace preorder $\sqsubseteq_T$ coincides with the simulation order on the
transition system $\Longrightarrow_F$, that is, $p\sqsubseteq_T q$ iff
$p\sqsubseteq^{\Rightarrow_F}_S q$.
\end{thm}

Even if this coincidence is a simple fact that reflects the relation between 
traces and failures semantics, it contributes to clarify it.
In plain words, failures
semantics is just traces semantics enriched by the observation of initials,
so that the plain simulation order that implies the trace order becomes the 
ready simulation order.

For other, finer observers such as $T$ we can also characterize the 
corresponding semantic orders, such as possible and impossible futures, 
in terms of local simulations up-to.
We can use that result to justify that the corresponding
transition systems $\Longrightarrow^T_Z$ would characterize the
semantic orders $\sqsubseteq^T_Z$ in terms of $T$-simulations that
preserve the set of traces of the simulated process.
In this case the corresponding operational characterization has to include rules for the 
computation of the set of traces $T(p)$ and this cannot certainly be done
for infinite processes.
But out of the computation of these sets, the rest of the rules for the
generation of the corresponding transition systems $\Longrightarrow^T_Z$
are also valid, and their local character is still present.

\subsection{Application: trace deterministic normal forms}
\label{atdnf:sec}

As a simple application we present the example used by Klin
in~\cite{Kli04}, that we already used in \cite{FG05concur} to
illustrate our coinductive characterization of the behavior
preorders by means of our bisimulations up-to.

\begin{defi}
For any process $p=\sum_a\sum_i ap^i_{a}$ the \emph{deterministic
form} of $p$ is defined as $\deter{p}=\sum_a a\deter{\sum_i
p^i_{a}}$.
\end{defi}

We wish to prove that $p$ and $\deter{p}$ are trace equivalent. We
will do it by proving that they are simulation equivalent over the
transition system $\Longrightarrow_F$.

\begin{prop}
\label{prop:menordetTra}
  For any process $p$ we have  $p\ltpF \deter{p}$.
\end{prop}
 \begin{proof}
     We will prove that $p\ltp^{\Rightarrow_F}_{S} \deter{p}$ by showing
     that $R=\{(p,\deter{p+q}) \mid \textrm{$p$, $q$ processes}\}$ is a 
     simulation for the transition system $\Longrightarrow_F$.
     For $q=\sum_a\sum_j aq^j_{a}$ we have $\deter{p+q}= \sum_a\deter{\sum_i
      p^i_{a} + \sum_j q^j_{a}}$. Then, for any $p \Tran{a}_F p'$ we have
     $p=p^i_{a} + \sum_k r^k_{a}$, for some index $i$ and $p^k_{a} = r^k_{a}
      + s^k_{a}$ a decomposition of any of the rest of the summands of $p$.
      We have $ \deter{p+q} \tran{a} \deter{\sum_i ap^i_{a} +\sum_j
      aq^j_{a}} =  \deter{(p^i_{a} + \sum_k r^k_{a})+ (\sum_k
      r^k_{a} + \sum_j q^j_{a})}$, so that we also have $
      \deter{p+q} \Tran{a}_F \deter{(p^i_{a} + \sum_k r^k_{a})+ (\sum_k
      r^k_{a} + \sum_j q^j_{a})}$, with $(p^i_{a} + \sum_k
      r^k_{a},\deter{(p^i_{a} + \sum_k r^k_{a})+ (\sum_k
      r^k_{a} + \sum_j q^j_{a})}) \in R$.
      \end{proof}

\begin{prop}
  For any process $p$ we have  $\deter{p} \ltpF p$.
\end{prop}
  \begin{proof}
   We will prove that $\deter{p}\ltp^{\Rightarrow_F}_{S}p$ by showing that
   $R=\{(\deter{p},p)\}$ is a simulation for the transition
    system $\Longrightarrow_F$. Since $\deter{p}$ is deterministic for
    each $a\in Act$ there is a unique transition $\deter{p} \Tran{a}_F
    \deter{\sum_i p^i_{a}}$. By applying the definition of
    $\Tran{a}_F$ we have $p \Tran{a}_F \sum_i p^i_{a}$, and
    clearly we have $( \deter{\sum_i p^i_{a}},\sum_i p^i_{a}) \in
    R$.
    \end{proof}

Although this is a very simple example, it is
interesting to compare the proof above with that in
\cite{FG05concur}. This proof is simpler and more
natural, mainly because the proof obligations to check
bisimulations forced us to remove the sub-terms that were not in
the chosen transition when we had to simulate it. This is
not necessary for any of the two simulations that
are needed to check mutual simulation, as done above. Obviously, this is
also related to the impossibility to obtain a notion of local
bisimulation up-to characterizing the equivalence under any of the
linear semantics.

\section{Conclusions and some future work}
\label{sec:conclusions}

Throughout this paper we have provided a global outline of process 
semantics from different points of view, each of which reveals some of the 
key ingredients for a more uniform comprehension of those semantics.
We have noted that the family consisting of the simulation 
semantics---constrained simulations, in its generalized version---plays 
an essential role in the class of process semantics, becoming the cornerstone
for sorting and classifying the remaining semantics.

From a framework in which, based on observational trees, denotational semantics
are assigned ---Section~\ref{observational-sem-sec}---we 
have been able to prove that the spectrum of process semantics can
be structured by means of layers that are induced by the simulations.
Each layer is dominated by a simulation semantics that determines the
finest distinction available for that layer.
The remaining semantic families are also described by abstracting
or simplifying the observations needed for the corresponding layer.
In particular, below each constrained simulation there appear the corresponding
versions for each of the  
classic linear semantics---failures, readiness, failure trace, and ready
trace---and, as we saw in Section~\ref{sec:RealDiamond}, other semantics
are also explained within our framework.

This observational characterization allowed us to offer a new insight into
the axiomatic characterization of the
semantics---Sections~\ref{sec:EquationalSemantics}
and~\ref{rnoef:sec}---revealing a uniformity lacking in all previous studies. 
To characterize any of the orders that define a process semantics,
we have proved that it is enough to use two parametric axioms:
one of the required axioms is that for the generalized simulation of the
corresponding layer while the other, when it is present, has to do with 
the reduction of non-determinism that is carried out in each semantics.

Analogously, in Sections~\ref{sec:LogicalCharacterization} 
and~\ref{logic_observational_framework} we showed how to characterize
process semantics by means of sets of Hennesy-Milner logic formulas out
of their observational characterization, and finally we have also discussed a
unified operational presentation of the semantics in the extended spectrum.

One of the more obvious lines for future work would be to consider those
semantics that allow for an inner, non-visible action, known as \emph{weak
  semantics}.
Actually, some promising results have already been obtained that make clear
the regularity and generality present in the domain of weak semantics.
In particular, in \cite{CFG08} it is proved that it is possible to apply
to weak semantics the algorithm to obtain axiomatic characterizations 
of semantic equivalences from the axioms for corresponding order
\cite{DeFrutosEtAl08b}.
And \cite{AFGI11ictac,AFGI12} provides a detailed study of the axiomatization of 
weak simulation semantics.

Let us also cite here the recent work by Anti Valmari \cite{valmari2012all},
where he presents the full catalogue of (weak) linear-time congruences for
finite state systems. Certainly, it is interesting to limit somehow the class
of ``reasonable'' semantics for processes, but this has not been so much the 
intention of our work in this paper. In fact, it is interesting to note that
the results in the paper referenced above limit the set of semantics to 
explore in a quite personal way: for instance, the semantics of failure traces
and that of ready traces are not included in the category, because Valmari
(implicitly) considers that they are not ``linear-time enough''. 

Another interesting approach consists in the use of coalgebras---following the
work, among others, of Jesse Hughes and Bart Jacobs~\cite{HJ04}---where
powerful categorical techniques allow to connect the idea of simulation with
that of bisimulation, which is central in the coalgebraic setting.
More concretely, these techniques were successfully used in \cite{FPF08} to
relate classic and probabilistic bisimulation.

\bibliographystyle{plain}

\end{document}